\documentclass[a4paper,11pt]{article}

\usepackage{graphicx}
\usepackage{float}
\usepackage[T1]{fontenc}
\usepackage{epsfig}
\usepackage{xcolor}
\usepackage{amsmath,jheppub,mathtools,amsthm}
\usepackage{subfloat}
\usepackage{amsfonts}
\usepackage{braket}
\usepackage{cleveref}
\usepackage{epstopdf}
\usepackage{caption}
\usepackage{subcaption}
\usepackage{enumitem}
\usepackage{comment}
\usepackage[numbers]{natbib}
\usepackage[titletoc,toc,title]{appendix}
\usepackage{hyperref}
\hypersetup{
	colorlinks=true,
	linkcolor=blue,
	filecolor=red,      
	urlcolor=blue,
	citecolor=blue
}

\newtheorem{proposition}{Proposition}

\title{Multipartite Non-local Magic and SYK Model}

\author[a,b]{Vinay Malvimat,}
\author[c]{Matthieu Sarkis,}
\author[a]{Yena Suk,}
\author[a,b,d]{Junggi Yoon}

\affiliation[a]{Department of Physics, College of Science, Kyung Hee University, Seoul 02447, Republic of Korea}
\affiliation[b]{Research Institute for Basic Sciences, Kyung Hee University, Seoul 02447, Republic of Korea}
\affiliation[c]{Department of Physics and Materials Science, University of Luxembourg, L-1511 Luxembourg City, Luxembourg}
\affiliation[d]{International Center for Quantum Matter, Kyung Hee University, Seoul 02447, Republic of Korea}

\emailAdd{vinaymalvimat@khu.ac.kr}
\emailAdd{matthieu.sarkis@uni.lu}
\emailAdd{yenahapphy@khu.ac.kr}
\emailAdd{junggi.yoon@khu.ac.kr}

\date{}

\Crefname{figure}{Fig.}{Figs.}

\abstract{We investigate the structure of quantum magic in interacting disordered fermionic systems, quantifying non-stabilizerness via the fermionic stabilizer R\'enyi entropy (SRE). To resolve the distribution of magic across different scales, we introduce a \emph{multipartite non-local magic} functional, constructed from an inclusion-exclusion combination of subsystem contributions. This measure serves as a fine-grained diagnostic, isolating genuinely global contributions and revealing nontrivial interactions between local and collective supports of magic. 
We illustrate the measure on paradigmatic multipartite states and apply these diagnostics to the Sachdev-Ye-Kitaev model and its variants. Crucially, for thermal/typical ensembles, we observe a marked disparity between Thermal Pure Quantum (TPQ) states and the thermal density matrix. This reveals a \textit{concealed complexity}: the immense computational hardness characterizing the unitary evolution is encoded in the specific microstructure of the black hole microstates, while being washed out in the coarse-grained thermodynamic description. Furthermore, in $\mathcal N=2$ supersymmetric SYK, we show that while fortuitous BPS states exhibit intermediate stabilizer complexity, the multipartite measure unveils a rich, sector-dependent pattern of global correlations, distinguishing them from generic chaotic states.

}

\begin{document} 
	\maketitle
	\flushbottom

\pagebreak
\section{Introduction}\label{sec_intro}

Quantifying the classical simulability of quantum states lies at the heart of understanding the extent to which quantum computers can outperform classical devices. As quantum platforms move steadily toward scalable architectures, identifying the precise features that render quantum states difficult to simulate becomes increasingly crucial. States that evade such efficient simulation possess a form of computational ``nonclassicality'' collectively referred to as \emph{quantum magic} \cite{Veitch2014,HowardCampbell2017}. A foundational result in this direction is the Gottesman-Knill theorem, which shows that stabilizer states and all states reachable from them by Clifford operations can be efficiently simulated on a classical computer \cite{Gottesman:1998hu,PhysRevA.70.052328}. . Various measures have been developed to characterize this magic, each aiming to quantify how far a quantum state lies beyond the classically tractable stabilizer framework \cite{Kenfack:2004ges,Nystrom:2024oeq,park2025localwignermassmapsintegrated}.

For the present work, we focus on one such measure, the \emph{Stabilizer Rényi Entropy} (SRE)~\cite{LeonePRL2022,LeonePRA2024}. This quantity provides a quantitative assessment of how far a quantum state deviates from the set of stabilizer states. Concretely, it is defined by expanding the density matrix of a state in the Pauli operator basis or Majorana basis and examining the distribution of its stabilizer-compatible components. States with a highly concentrated Pauli/Majorana distribution as in stabilizer states possess low SRE and are therefore classically simulable. In contrast, non-stabilizer states spread their weight over many Pauli/Majorana strings, yielding a larger SRE and signaling a greater degree of ``magic'' or nonclassical complexity. The SRE thus serves as a powerful indicator of the classical intractability of a quantum state. 

In practice, the SRE has become a broadly used tool of non-stabilizerness across many systems. In quantum optics, it can track the time evolution of atomic magic in the Jaynes-Cummings model~\cite{shuangshuang2022dynamics}, while in quantum chemistry it has been used to measure non-stabilizerness in molecular bonding~\cite{sarkis2025molecules}. In many-body settings it has been applied to permutationally invariant systems and kinetically constrained Rydberg-atom arrays~\cite{passarelli2024nonstabilizerness,smith2406non}, to fermionic problems, including strongly interacting models and fermionic Gaussian states~\cite{Bera:2025pfp,collura2412quantum}, as well as to hybrid boson-fermion systems~\cite{sarkis2025magic}. It also quantifies magic growth and spreading in random circuits and generic ergodic dynamics~\cite{turkeshi2025magic,tirrito2024anticoncentration}, and extends naturally to non-Hermitian regimes where it helps design protocols for producing highly magic states~\cite{martinezazcona2025magic}.

In this broader context, it is natural to ask whether black hole microstates or, more generally, quantum states arising in quantum field theories with holographic black hole duals are easy or hard to simulate classically. Put differently, do such states carry substantial quantum magic \cite{White:2020zoz,Cao:2024nrx,Basu:2025mmm,Basu:2025uxw}? The SYK model and its numerous variants provide an ideal setting in which to explore these questions. The SYK$_4$ model, in particular, has yielded deep insights into holography and quantum chaos due to its maximally chaotic dynamics and emergent gravitational features \cite{Sachdev_1993,Polchinski:2016xgd,Jevicki:2016bwu,Maldacena:2016hyu}. Its mass-deformed extension \cite{Banerjee:2016ncu,Garcia-Garcia:2017bkg,Nosaka:2018iat,Nandy:2022hcm} and sparse variants \cite{Xu:2020shn,Orman:2024mpw} further enrich this landscape by exhibiting transitions between chaotic and integrable behavior.
 These properties make SYK-type models a compelling playground for studying how quantum magic behaves in systems that mimic aspects of black holes. In this work, we therefore investigate the stabilizer Rényi entropy of states generated by these models and analyze how their magic content evolves across different dynamical and thermodynamic regimes (see also \cite{Bera:2025pfp,Jasser:2025myz,Zhang:2025rky}).  At the same time, growing evidence points to a deep connection between chaos and quantum magic \cite{Passarelli:2024lpm,PhysRevD.106.126009}, suggesting that non-classical resources may play a fundamental role in quantum chaotic dynamics. Furthermore, to refine the notion of magic, it is essential to distinguish contributions arising from genuinely global correlations from those generated by local or few-body structures. Measures that quantify only the total non-stabilizerness of a state may obscure how magic is distributed across its multipartite degrees of freedom. A multipartite non-local extension of the SRE is therefore required to isolate the component of magic that originates solely from long-range, collective correlations beyond any local stabilizer deviations.

In the present work, we take a step toward addressing these questions by introducing a multipartite non-local extension of the stabilizer Rényi entropy, designed to isolate the component of quantum magic that originates from genuinely global correlations. This refined measure allows us to distinguish between local non-stabilizerness and the collective, long range structure that plays a central role in holographic and chaotic systems. We then apply both the conventional SRE and our multipartite extension to a range of SYK models including the canonical SYK$_4$, its mass-deformed and sparse variants. We investigate how mass deformation and sparseness influence the generation and structure of quantum magic in these models.
Subsequently, we examine these quantities in the ${\cal N}=2$ supersymmetric extension to probe how magic is generated and distributed in systems featuring fortuitous BPS states, which are proposed black-hole microstate candidates \cite{Fu:2016vas,Chang:2024lxt}. By comparing the SRE and its multipartite counterpart across BPS, $Q$-exact, and typical states within fixed charge sectors, we reveal how supersymmetric cohomology and fortuity influence their nonclassical complexity. This, in turn, allows us to assess the classical simulability of states emerging in chaotic, integrable, and supersymmetric regimes and to elucidate the structure of quantum magic in these models.

In \Cref{sec 2}, we review the conventional stabilizer Rényi entropy and introduce a multipartite non-local extension designed to capture global contributions to magic. In \Cref{sec 3} we apply these measures to the SYK$_4$ model, analyzing their behavior under real-time evolution and in thermal ensembles. \Cref{sec 4} investigates the impact of mass deformation on quantum magic, while \Cref{sec 5} explores how sparseness modifies the structure of magic in sparse SYK variants. In \Cref{sec 6} we examine the distribution of magic in the ${\cal N}=2$ supersymmetric SYK model, highlighting features associated with fortuitous BPS states. Finally, \Cref{sec 7} summarizes our findings and outlines potential future directions.

\section{Non-Stabilizerness and Stabilizer Renyi Entropy}\label{sec 2}

\subsection{Non-stabilizerness in finite Majorana systems and the stabilizer Renyi entropy}
\label{sec:majorana_sre}

    The resource theory of non-stabilizerness (or ``magic'') formalizes the gap between classically simulable operations and universal quantum computation. In its canonical form for qubits, the free states are the \emph{stabilizer states}, the free unitaries are the \emph{Clifford group}, and free measurements and classical randomness complete the classically simulable operations~\cite{Veitch2014,HowardCampbell2017}. In this framework, any deviation from the stabilizer realm quantifies a computational resource non-stabilizerness captured by a hierarchy of monotones such as robustness of magic and the \emph{stabilizer Rényi entropy} (SRE)~\cite{LeonePRL2022,LeonePRA2024}.
    
    In the fermionic setting, particularly for systems composed of a finite number of Majorana modes, one can formulate a parallel structure where the free operations are the \emph{fermionic Clifford group} which is the stabilizer group of the \emph{Majorana group} and the free states are the corresponding \emph{Majorana stabilizer states}. In this section, we present the stabilizer Rényi entropy defined via the \emph{Majorana spectrum}, which quantifies non-stabilizerness for the fermionic systems to be studied in the rest of the paper.

\subsection{Clifford Majorana group and Majorana stabilizer states}
\label{sec:stab_majorana}

    Consider a system of $2n$ Majorana operators
    \begin{equation}
    \label{eq:majorana-frame}
        \Gamma \equiv \{\gamma_1,\dots,\gamma_{2n}\},\qquad \gamma_j^\dagger=\gamma_j,\qquad \{\gamma_j,\gamma_k\}=2\delta_{jk}\,,
    \end{equation}
    acting on a $2^n$-dimensional Hilbert space, with even-parity superselection imposed. 
    The choice of the ordered set $\Gamma$ will be referred to as the \emph{Majorana frame}; in what follows, all notions are defined relative to this fixed frame. Denote by
    \begin{equation}
    \label{eq:strings}
        \Gamma_S = i^{\frac{|S|(|S|-1)}{2}} \prod_{j\in S}\gamma_j,\qquad S\subseteq [2n],\qquad \Gamma_\varnothing=1\,,
    \end{equation}
    the Hermitian Majorana strings. They furnish an orthonormal operator basis for the Hilbert–Schmidt inner product $\langle A,B\rangle=2^{-n}\text{Tr}(A^\dagger B)$.
    
    The \emph{Majorana group} $\mathcal{M}_\Gamma$ is generated by all $\Gamma_S$ together with overall phases $\{\pm 1,\pm i\}$. Its normalizer inside $\mathrm{U}(2^n)$ defines the \emph{Majorana Clifford group}
    \begin{equation}
    \label{eq:clifford}
        \mathcal{C}_\Gamma = \{U\in \mathrm{U}(2^n)\,|\, U\mathcal{M}_\Gamma U^\dagger=\mathcal{M}_\Gamma\}\,.
    \end{equation}
    Equivalently, $U\in \mathcal{C}_\Gamma$ iff $U$ acts as a signed permutation on the fixed Majoranas:
    \begin{equation}
    \label{eq:permutation}
        U\,\gamma_j\,U^\dagger = \sigma_j\,\gamma_{\pi(j)},\qquad \sigma_j\in\{\pm 1\},\ \ \pi\in S_{2n}\,.
    \end{equation}
    
    A \emph{stabilizer subgroup} is a maximal Abelian subgroup $\mathcal{S}\subset \mathcal{M}_\Gamma$ generated by $n$ independent, commuting, \emph{even-parity} strings\footnote{Physical states are assumed to respect fermionic parity superselection.}. The unique common $+1$ eigenstate of $\mathcal{S}$ is a \emph{stabilizer state}. With the frame fixed, the set of pure stabilizer states $\mathrm{STAB}_\Gamma$ is finite\footnote{For qubits, for which no superselection rule applies, the number of stabilizer states is given by $
    |\mathrm{STAB}_\Gamma|= 2^n \prod_{k=1}^n (2^k+1)$.}:
    \begin{equation}
    \label{eq:stab-count}
        |\mathrm{STAB}_\Gamma| = 2^n\,(2^n-1)\,\prod_{k=1}^{n-1} \big(2^{\,k}+1\big)\,.
    \end{equation}
    For instance for $2$ modes with frame $\{\gamma_1,\gamma_2\}$, the even-parity nontrivial string is $i\gamma_1\gamma_2$. The two maximal Abelian subgroups are
    \begin{equation}
        \mathcal{S}_\pm = \{1,\pm i\gamma_1\gamma_2\}\,,
    \end{equation}
    with stabilizer states $|0\rangle$ and $|1\rangle=b^\dagger|0\rangle$ in the Fock basis $b=(\gamma_1+i\gamma_2)/2$. 
    
    The free operations of the resource theory are:
    (i) unitaries in $\mathcal{C}_\Gamma$; 
    (ii) projective measurements of commuting even-parity strings from $\mathcal{M}_\Gamma$; and 
    (iii) classical randomness/post-processing. These are efficiently simulable and mirror the qubit stabilizer setting.

\subsection{Majorana spectrum and the stabilizer Rényi entropy}

    For a parity-even density operator $\rho$, expand it in the orthonormal basis $\{\Gamma_S\}$:
    \begin{equation}
        \rho = \frac{1}{2^n}\sum_{S\subseteq[2n]} c_S \,\Gamma_S,\qquad c_S =\text{Tr}(\rho\,\Gamma_S)\,.
    \end{equation}
    Define the \emph{Majorana spectrum} relative to the frame $\Gamma$,
    \begin{equation}
    \label{eq:maj-spectrum}
        p_S = \frac{|c_S |^2}{2^n},\qquad \sum_{S} p_S=1\,.
    \end{equation}
    For $\alpha\geq 1$, the Rényi-$\alpha$ entropy of the spectrum defines the \emph{stabilizer Rényi entropy}:
    \begin{equation} 
    \label{eq:def_SRE}
        \mathrm{SRE}_\alpha(\rho)=\frac{1}{1-\alpha}\log \sum_S p_S^\alpha  -  n\log 2.
    \end{equation}
    where in analogy with the qubit case, we normalized by the stabilizer baseline $n\log 2$ to ensures that all stabilizer states have $\mathrm{SRE}_\alpha=0$, so SRE quantifies non-stabilizerness. We will focus on the case $\alpha=2$ in the rest of the paper.
    
    To extend the definition to mixed states in a resource-theoretic sense, one needs to adopt the convex-roof-like construction. This approach involving a computationally heavy optimization step over the space of purifications, one typically adopts the following definition for mixed states:
    \begin{equation}
    \label{eq:def_SRE_mixed}
        \mathrm{M}_2 \equiv \mathrm{SRE}_2 - S_2 \ ,
    \end{equation}
    by simply substracting away the contribution of the Rényi entropy $S_2(\rho)=-\log\mathrm{Tr}(\rho^2)$. Though not a magic monotone per se, this quantity proves to be a good proxy of magic for mixed states. More explicitely, we have:
    \begin{equation}
    \label{eq:def_SRE_mixed_simple}
        \mathrm{M}_2(\rho) = -\log\left[\frac{\sum_S c_S^4}{\sum_S c_S^2}\right]\ .
    \end{equation}

\begin{comment}

    Let us also define the notion of linearized SRE by expanding the log to first order:
    \begin{equation} 
    \label{eq:def_LSRE}
        \mathrm{LSRE}_2(\rho)=1-2^n\sum_S p_S^2 = 1 - \frac{1}{2^n}\sum_S c_S^4
    \end{equation}
    and the corresponding definition for mixed states, either by substracting away the linearized Renyi entropy or by directly linearizing (\ref{eq:def_SRE_mixed_simple}) \textcolor{purple}{MS: Guys, I am not sure which of the following two definitions is the most meaningful, what do you think? I think they should give very similar results though.}: \textcolor{blue}{Add reference and choose one of them. We don't use this quantity in current paper.}
    \begin{equation} 
    \label{eq:def_mixedLSRE}
        \mathrm{LM}_2(\rho)=
        \left\{
                \begin{array}{ll}
                  1 - \frac{\sum_S c_S^4}{\sum_S c_S^2}\\
                  \text{or} \\
                  -\frac{1}{2^n}\sum_S\left(c_S^4 - c_S^2\right)
                \end{array}
              \right.
    \end{equation}

\end{comment}

\subsection{Multipartite Non-Local Magic}

    \subsubsection{Definition and Properties}

    Drawing inspiration from the multipartite mutual information, given an $n$-partite system, we define the \textit{multipartite non-local magic} as:
    \begin{equation}\label{MNLDefin}
        \mathrm{M}^{(n)}_\textsc{nl}(\rho_{12\dots n}) := \sum_{a=1}^n(-1)^{n-a}\sum_{i_1<\dots<i_a=1}^n \mathrm{M}(\rho_{i_1\dots i_a})\ ,
    \end{equation}
    where $\mathrm{M}$ denotes $\mathrm{M}_2$ (from now on we omit the subscript for notational simplicity), and where we used obvious notations for the partial density matrices of the subsystems. The definition follows a very natural pattern avoiding multiple counting of parts contribution. 
    For instance in the tripartite case, we have:
    \begin{equation}
        \mathrm{M}^{(3)}_\textsc{nl}(\rho_{123}) := \mathrm{M}(\rho_{123}) - \mathrm{M}(\rho_{12}) - \mathrm{M}(\rho_{13}) - \mathrm{M}(\rho_{23}) + \mathrm{M}(\rho_{1}) + \mathrm{M}(\rho_{2}) + \mathrm{M}(\rho_{3})
    \end{equation}
    Slightly more compact expression in terms of subsets of $[n]\equiv\{1,2,\dots, n\}$ reads:
    \begin{equation}
    \label{eq:def-InM-new}
        \mathrm{M}^{(n)}_\textsc{nl}(\rho_{[n]}) =  \sum_{\emptyset \neq S \subseteq [n]} (-1)^{\,n-|S|}\,\mathrm{M}\left(\rho_S\right).
    \end{equation}
    with again obvious subset notations for partial density matrices.
    
    Let us prove here a few properties satisfied by the non-local magic. The properties naturally derive from the properties of the SRE and combinatorics. Let us therefore list some basic properties satisfied by $\mathrm{M}$:
    \begin{proposition}
    $\mathrm M$ satisfies the following properties:
    \begin{enumerate}
        \item (Nonnegativity) $\mathrm{M}(\rho)\ge 0$ for all states $\rho$.
        \item (Additivity) $\mathrm{M}(\rho\otimes\sigma) = \mathrm{M}(\rho) + \mathrm{M}(\sigma)$. \label{prop:additivity}
        \item (Invariance under Clifford unitaries) $\mathrm{M}(U\rho U^\dagger) = \mathrm{M}(\rho)$ for all  Clifford unitaries $U$. \label{prop:invariance}
        \item (Existence of zero-magic reference states) There exist states $\tau$ with $\mathrm{M}(\tau)=0$. For the concrete choice $\mathrm{M}=\mathrm{SRE}_2-S_2$, any pure stabilizer state has $\mathrm{M}(\tau)=0$, and some mixed states (e.g. maximally mixed) also satisfy $\mathrm{M}(\tau)=0$.
    \end{enumerate}
    \end{proposition}
    Equipped with these properties, we have
    \begin{proposition}[Tripartite identity]
    For three parties,
    \begin{equation}
      \mathrm{M}^{(3)}_\textsc{nl}(\rho_{123})
      = -\mathrm{M}^{(2)}_\textsc{nl}(\rho_{12}) - \mathrm{M}^{(2)}_\textsc{nl}(\rho_{13}) + \mathrm{M}^{(2)}_\textsc{nl}(\rho_{1(23)}).
    \end{equation}
    Analogous identities hold if one singles out
    parts $2$ or $3$ instead of $1$, where the notation $\rho_{1(23)}$ means that the bipartition $(1)$ vs. $(23)$ of the system is considered.
    \end{proposition}
    
    \begin{proof}
        Obvious.
    \end{proof}
        
    \begin{proposition}[States with vanishing local magic]
    Suppose $\rho_{[n]}$ is such that
    \begin{equation}
      \mathrm{M}(\rho_S) = 0
      \qquad\text{for all nonempty } S\subseteq[n].
    \end{equation}
    Then
    \begin{equation}
      \mathrm{M}^{(n)}_\textsc{nl}(\rho_{[n]}) = 0.
    \end{equation}
    \end{proposition}
    
    \begin{proof}
        Obvious.
    \end{proof}
    
    This property applies, for example, to certain special families such as the three-qubit GHZ state.
    
    \begin{proposition}[Product across a nontrivial partition, general $n$]
    Let $[n]$ be partitioned into $k\ge 2$ nonempty disjoint blocks
    \begin{equation}
      [n] = B_1 \,\cup\, B_2 \,\cup\, \cdots \,\cup\, B_k.
    \end{equation}
    Suppose
    \begin{equation}
      \rho_{1\cdots n} = \bigotimes_{j=1}^k \rho_{B_j}.
    \end{equation}
    Then
    \begin{equation}
      \mathrm{M}^{(n)}_\textsc{nl}(\rho_{[n]}) = 0.
    \end{equation}
    \end{proposition}
    
    \begin{proof}
    For any subset $S\subseteq [n]$, we have
    \begin{equation}
      \rho_S = \bigotimes_{j=1}^k \rho_{S\cap B_j}.
    \end{equation}
    By additivity (\ref{prop:additivity}),
    \begin{equation}
      \mathrm{M}(\rho_S) = \sum_{j=1}^k \mathrm{M}(\rho_{S\cap B_j}).
    \end{equation}
    Hence
    \begin{equation}
    \begin{aligned}
      \mathrm{M}^{(n)}_\textsc{nl}(\rho_{[n]})
      &= \sum_{\emptyset\neq S\subseteq [n]} (-1)^{\,n-|S|} \,\mathrm{M}(\rho_S) = \sum_{j=1}^k \sum_{\emptyset\neq S\subseteq [n]} (-1)^{\,n-|S|}
         \,\mathrm{M}(\rho_{S\cap B_j}).
    \end{aligned}
    \end{equation}
    Fix a block $B_j$. For any $\emptyset\neq T\subseteq B_j$, group all subsets
    $S$ such that $S\cap B_j = T$. These are of the form
    $S = T\cup U$ with $U\subseteq \bar{B}_j := [n]\setminus B_j$, and $\bar{B}_j$ is
    nonempty because $k\ge 2$. For such $S$, $|S| = |T| + |U|$. Thus the contribution from a fixed nonempty $T\subseteq B_j$ is
    \begin{equation}
      M(\rho_T) \sum_{U\subseteq \bar{B}_j} (-1)^{\,n-|T|-|U|}
      = M(\rho_T)(-1)^{\,n-|T|} \sum_{U\subseteq \bar{B}_j} (-1)^{|U|}.
    \end{equation}
    But one has
    \begin{equation}
      \sum_{U \subseteq C_j} (-1)^{|U|} = \sum_{k=0}^m \;\sum_{\substack{U \subseteq C_j\\ |U|=k}} (-1)^{|U|} = \sum_{k=0}^m \;\sum_{\substack{U \subseteq C_j\\ |U|=k}} (-1)^k = \sum_{k=0}^m \binom{m}{k} (-1)^k = 0,
    \end{equation}
    where we used the binomial identity in the last equality.  Therefore every such grouped contribution vanishes, and the whole sum is zero.
    \end{proof}
    
    Thus $\mathrm{M}^{(n)}_\textsc{nl}$ detects only magic that is genuinely shared among all $n$ parties simultaneously, in the sense that it vanishes whenever the global state factorizes across any nontrivial partition.
    
    \begin{proposition}[Local free-unitary invariance]
    Let $U = \bigotimes_{i=1}^n U_i$ where each $U_i$ is a stabilizer (Clifford)
    unitary acting on party $i$, and define $\rho' = U\rho U^\dagger$.
    Then
    \begin{equation}
      \mathrm{M}^{(n)}_\textsc{nl}\left(\rho'_{[n]}\right) = \mathrm{M}^{(n)}_\textsc{nl}\left(\rho_{[n]}\right).
    \end{equation}
    \end{proposition}
    
    \begin{proof}
    For any $S\subseteq[n]$,
    \begin{equation}
      \rho'_S = U_S \rho_S U_S^\dagger,
      \qquad
      U_S := \bigotimes_{i\in S} U_i.
    \end{equation}
    By (\ref{prop:invariance}),
    \begin{equation}
    \mathrm{M}\left(\rho'_S\right) = \mathrm{M}\left(\rho_S\right).
    \end{equation}
    Substituting into \eqref{eq:def-InM-new},
    \begin{equation}
      \mathrm{M}^{(n)}_\textsc{nl}\left(\rho'_{[n]}\right)
      = \sum_{\emptyset\neq S\subseteq[n]} (-1)^{\,n-|S|} \,\mathrm{M}\left(\rho'_S\right)
      = \sum_{\emptyset\neq S\subseteq[n]} (-1)^{\,n-|S|} \,\mathrm{M}(\rho_S)
      = \mathrm{M}^{(n)}_\textsc{nl}\left(\rho_{[n]}\right).
    \end{equation}
    \end{proof}
    
    \begin{proposition}[Additivity over independent $n$-partite states]
    Let $\rho^{(A)}_{[n]}$ and $\rho^{(B)}_{[n]}$ be two independent
    $n$-partite states. Consider the combined state
    \begin{equation}
      \rho^{(AB)}_{[n]}
      = \rho^{(A)}_{[n]}\otimes\rho^{(B)}_{[n]},
    \end{equation}
    where each part's Hilbert space is enlarged but the number of parts remains $n$. Then
    \begin{equation}
      \mathrm{M}^{(n)}_\textsc{nl}\left(\rho^{(AB)}\right) 
      = \mathrm{M}^{(n)}_\textsc{nl}\left(\rho^{(A)}\right) + \mathrm{M}^{(n)}_\textsc{nl}\left(\rho^{(B)}\right).
    \end{equation}
    \end{proposition}
    
    \begin{proof}
    For any $S\subseteq[n]$,
    \begin{equation}
      \rho^{(AB)}_S
      = \rho^{(A)}_S \otimes \rho^{(B)}_S.
    \end{equation}
    Additivity (\ref{prop:additivity}) gives
    \begin{equation}
      \mathrm{M}\left(\rho^{(AB)}_S\right) = \mathrm{M}\left(\rho^{(A)}_S\right) + \mathrm{M}\left(\rho^{(B)}_S\right).
    \end{equation}
    Therefore
    \begin{equation}
    \begin{aligned}
      \mathrm{M}^{(n)}_\textsc{nl}\left(\rho^{(AB)}\right)
      &= \sum_{\emptyset\neq S\subseteq[n]} (-1)^{\,n-|S|} \,\mathrm{M}\left(\rho^{(AB)}_S\right) = \sum_{\emptyset\neq S\subseteq[n]} (-1)^{\,n-|S|} 
         \left[ \mathrm{M}\left(\rho^{(A)}_S\right) + \mathrm{M}\left(\rho^{(B)}_S\right) \right] \\
      &= \mathrm{M}^{(n)}_\textsc{nl}\left(\rho^{(A)}\right) + \mathrm{M}^{(n)}_\textsc{nl}\left(\rho^{(B)}\right).
    \end{aligned}
    \end{equation}
    \end{proof}
    
    \subsubsection{Examples}

    Let us introduce two tripartites states, the GHZ and W states defined respectively as:
    \begin{equation}
            |\mathrm{GHZ}\rangle = \frac{|000\rangle + |111\rangle}{\sqrt{2}}, \qquad |W\rangle = \frac{|001\rangle + |010\rangle + |100\rangle}{\sqrt{3}}\,.
    \end{equation}
    By symmetry, we only need to worry about, say, $\rho_1$ and $\rho_{12}$.

    \paragraph{GHZ:} 

    $|{\rm GHZ}\rangle$ is a pure stabilizer state, hence
    \begin{equation}
      \mathrm{SRE}_2(\rho_{123}) = 0, \qquad S_2(\rho_{123}) = 0, \qquad \mathrm{M}(\rho_{123}) = 0.
    \end{equation}
    By tracing out, one finds
    \begin{equation}
      \rho_1 = \frac{\mathbb{I}_2}{2},
    \end{equation}
    \begin{equation}
      \rho_{12} = \frac{1}{2}\big( |00\rangle\langle 00| + |11\rangle\langle 11| \big).
    \end{equation}
    For a single-qubit maximally mixed state $\rho = \mathbb{I}_2/2$,
    \begin{equation}
      \mathrm{Tr}(\rho^2) = \frac{1}{2} \quad\Rightarrow\quad S_2(\rho) = -\log\left(\frac{1}{2}\right) = \log 2.
    \end{equation}
    Only $\mathrm{Tr}(\rho I)=1$ is nonzero among Pauli expectations, so
    \begin{equation}
      \frac{1}{2}\sum_{P\in\{I,X,Y,Z\}} \mathrm{Tr}(\rho P)^4 = \frac{1}{2} \quad\Rightarrow\quad \mathrm{SRE}_2(\rho) = -\log\left(\frac{1}{2}\right) = \log 2.
    \end{equation}
    Hence
    \begin{equation}
      \mathrm{M}(\rho_1) = 0.
    \end{equation}
    For $\rho_{12}$ the eigenvalues are $\{1/2,1/2,0,0\}$, so
    \begin{equation}
      \mathrm{Tr}\big[(\rho_{12})^2\big] = \frac{1}{2} \quad\Rightarrow\quad S_2(\rho_{12}) = \log 2.
    \end{equation}
    Among two-qubit Pauli operators, only $I\otimes I$ and $Z\otimes Z$ have nonzero expectation equal to $1$, hence
    \begin{equation}
      \frac{1}{4}\sum_{P}\mathrm{Tr}(\rho_{12} P)^4 = \frac{1}{2} \quad\Rightarrow\quad \mathrm{SRE}_2(\rho_{12}) = \log 2.
    \end{equation}
    Thus
    \begin{equation}
      \mathrm{M}(\rho_{12}) = \mathrm{M}(\rho_{13}) = \mathrm{M}(\rho_{23}) = 0.
    \end{equation}
    Putting everything together,
    \begin{equation}
      \mathrm{M}^{(3)}_{\textsc{nl}}(\rho_{123}) = 0.
    \end{equation}

    \paragraph{W state:}
    
   For the bipartite mixed-state SRE, the contributions arise from two-qubit Pauli operators. Among the single-site terms, only $Z\otimes I$ and $I\otimes Z$ have nonvanishing expectation values, each equal to $1/9$. In addition, the correlators $X\otimes X$, $Y\otimes Y$, and $Z\otimes Z$ contribute with expectation values $4/9$, $4/9$, and $1/9$, respectively. Consequently, the bipartite stabilizer R\'enyi entropies are
\begin{equation}
    \mathrm{M}(\rho_{12}) 
    = \mathrm{M}(\rho_{13})
    = \mathrm{M}(\rho_{23})
    = \log\!\left[\frac{45}{29}\right].
\end{equation}

Similarly, the single-party SREs evaluate to
\begin{equation}
    \mathrm{M}(\rho_1)
    = \mathrm{M}(\rho_2)
    = \mathrm{M}(\rho_3)
    = \log\!\left[\frac{45}{41}\right].
\end{equation}

The three-party stabilizer R\'enyi entropy is found to be
\begin{align}
    \mathrm{M}(\rho_{123})
    = \log\!\left[\frac{9}{5}\right].
\end{align}

Putting these results together, as in the earlier example, the non-local stabilizer R\'enyi entropy (or non-local magic) is
\begin{align}
    \mathrm{M}_{\mathrm{NL}}(\rho_{123})
    = -0.451.
\end{align}

\begin{table}[h]
\centering
\begin{tabular}{|c|c|}
\hline
Number of  qubits &   W state \\
\hline
 3 &$- 0.451$  \\
 4  &$-0.4274$ \\
 5 &$0.0922$ \\
 6  & $0.2261$ \\
 7  & $0.05627$ \\
 8  & $-0.02696$ \\
 9  & $-0.00194$ \\

\hline
\end{tabular}
\caption{Non-local SRE for W states.}
\label{tab:entropy-GHZW}
\end{table}
    Note that for any $n$-partite GHZ and cluster states this quantity is zero as they are stabilizer states. However quite interestingly for generalized GHZ state it turns out to be nonzero.
    \begin{align}
        \ket{gGHZ}=\cos \theta\ket{000..}+\sin \theta\ket{111...}
    \end{align}
    In this case when the SRE of the full state is always
    \begin{align}
        M(  \ket{gGHZ})=\log\left[\frac{8}{7+\cos (8\theta)}\right]
    \end{align}
    And all the lower party SRE is given by
    \begin{align}
     M(\rho)  = \log \left[\frac{2 \left(\cos ^4(2 \theta )+1\right)}{\cos (4 \theta )+3}\right]
    \end{align}
    Therefore non-local magic or SRE for a generalized GHZ state differs based on whether the total number of qubits is odd or even 
    \begin{align}
        \mathrm{M}^{\mathrm{odd}}_{\textsc{nl}}(\ket{gGHZ})&=\log\left[\frac{8}{7+\cos(8\theta)}\right]\\
          \mathrm{M}^{\mathrm{even}}_{\textsc{nl}}(\ket{gGHZ})&=\log\left[\frac{8}{7+\cos(8\theta)}\right]+2 \log \left[\frac{2 \left(1+\cos ^4(2 \theta )\right)}{3+\cos (4 \theta )}\right]
    \end{align}
    The behavior with theta is plotted in Fig.~\ref{NLSREgGHZ}.

    \begin{figure}[H]
\centering\includegraphics[width=0.6\linewidth]{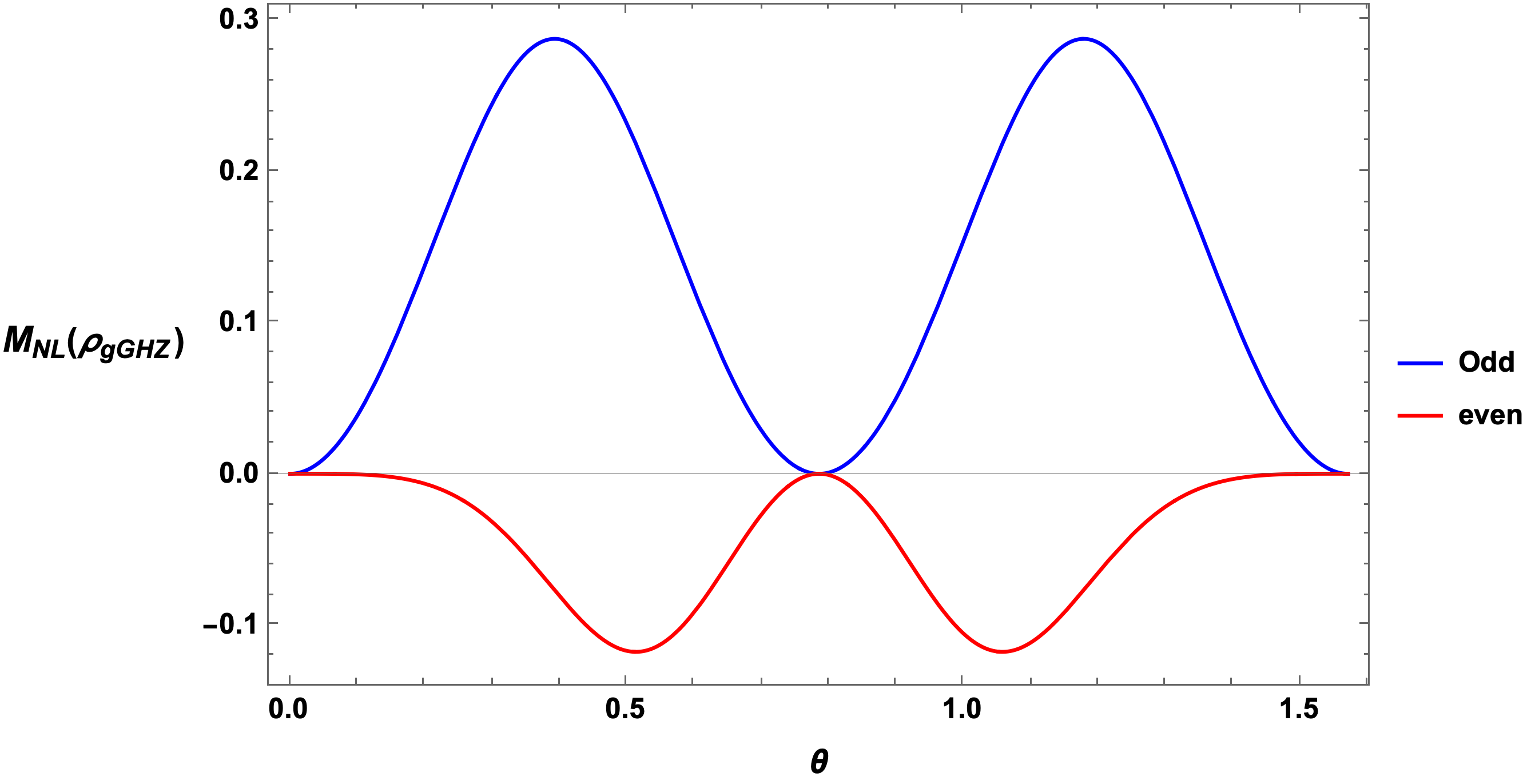}
    \caption{Non-local SRE for a generalized GHZ state}\label{NLSREgGHZ}
    \end{figure}

In the bipartite setting, it is well known that maximally entangled states such as the Bell states can nevertheless carry zero magic, since they belong to the class of stabilizer states. This phenomenon becomes even more striking in the multipartite context. For example, GHZ states, which are regarded as maximally genuinely tripartite entangled in the three-qubit setting, as quantified by measures such as the concurrence triangle and the $L$-entropy \cite{Xie:2021hsy,Basak:2024uwc,Ahn:2025bdm}, remain stabilizer states and therefore possess no magic. A similar situation arises for certain highly entangled four-party states, such as cluster states, which also exhibit maximal multipartite entanglement but zero magic due to their stabilizer nature. These examples underscore the fact that entanglement and magic quantify distinct forms of nonclassicality and need not align.

    \subsubsection{Comments on negativity}

    Even though the underlying magic $\mathrm{M}(\rho)$ is nonnegative for every state $\rho$, the combination defining $\mathrm{M}^{(n)}_{\textsc{nl}}$ uses alternating signs. In other words, $\mathrm{M}^{(n)}_{\textsc{nl}}$ is not a simple sum of positive contributions; it is an interference pattern of magical contributions from different subsystem scales, in which lower-order contributions are systematically subtracted to isolate a genuine $n$-body piece. As a consequence, cancellations can be strong enough to make $\mathrm{M}^{(n)}_{\textsc{nl}}$ negative.

    Let us illustrate in the tripartite case. For $n=3$, we can rewrite the definition as
    \begin{equation}
        \begin{aligned}
              \mathrm{M}^{(3)}_{\textsc{nl}}(\rho_{123})
              &= \mathrm{M}(\rho_{123})
               - \Big[
                  \mathrm{M}(\rho_{12}) + \mathrm{M}(\rho_{13}) + \mathrm{M}(\rho_{23})
                  - \mathrm{M}(\rho_1) - \mathrm{M}(\rho_2) - \mathrm{M}(\rho_3)
                 \Big].
        \end{aligned}
    \end{equation}
    The bracketed term contains all contributions that can be constructed from local and pairwise magic. The difference $\mathrm{M}^{(3)}_{\textsc{nl}}$ is then the part of the magic that cannot be explained by such lower-order contributions. If $\mathrm{M}^{(3)}_{\textsc{nl}}(\rho_{123})>0$, then the tripartite magic of $\rho_{123}$ is larger than what one would predict from single and two-body magics. There is a genuinely \emph{synergistic} component of magic that only appears when all three parties are considered together. If instead $\mathrm{M}^{(3)}_{\textsc{nl}}(\rho_{123})<0$, then the combination of single- and two-body magics overestimates the magic that is truly available in the tripartite system. In this case, the lower-order magics are redundant: they contain overlapping information about the same underlying non-stabilizer structure. When these overlaps are subtracted away, the net connected three-body magic becomes negative.

    The two examples above illustrate the meaning of the sign of the multipartite non-local magic. For the W state, $\mathrm{M}^{(3)}_{\textsc{nl}}(\rho_{123})<0$. The system possesses magic at all scales (global, pairwise and local), but the lower-order magic is highly redundant and overlaps strongly with the global magic. After subtracting these overlaps, the connected three-body contribution is negative. This describes a redundant or \emph{frustrated} distribution of magic. Instead, the generalized GHZ state exhibits non-negative $\mathrm{M}^{(3)}_{\textsc{nl}}$ indicating that the non-stabilizer resource is genuinely tripartite and cannot be reconstructed from any combination of one and two-qubit reduced states.

    In general, the multipartite non-local magic $\mathrm{M}^{(n)}_{\textsc{nl}}$ should be viewed not as a basic magic monotone (which it is not, as we saw), but as a tool revealing how the magic resource is distributed across different scales. Its magnitude captures the strength of the genuinely $n$-body connected magic, while its sign distinguishes between synergistic (positive), factorizable (zero) and redundant (negative) patterns of non-stabilizer correlations. The subtle allocation of magic among the different scales of the system being strongly system size dependent, one should not be surprised in the sign of $\mathrm{M}^{(n)}_{\textsc{nl}}$ being strongly sensitive to $n$ in general, as can already be seen at the level of the two families of states we considered above.

Note that there exists a multipartite generalization of mutual information analogous to that of SRE in \cref{MNLDefin}, obtained by replacing $M$ with entanglement entropies of respective subsystems. As an illustrative example, the tripartite information can be rewritten in terms of bipartite mutual informations as  
\begin{align}\label{Idef}
I_3(A:B:C) &:= S(A)+S(B)+S(C) - S(AB) - S(BC) - S(AC) + S(ABC) \notag\\
&= I(A:B) + I(A:C) - I(A:BC)\, .
\end{align}
The negativity of the tripartite information admits a natural interpretation: it represents a strengthened form of monogamy of entanglement, a property that is not satisfied by generic quantum states. Remarkably, holographic states are known to exhibit strictly negative tripartite information \cite{Hayden:2011ag,Gharibyan:2013aha}. By contrast, it remains unclear whether the negativity of multipartite quantum magic admits a similar monogamous interpretation. This is an interesting issue that we leave for further exploration in the near future.

    \section{The SYK model} \label{sec 3}
   Let us begin with a brief review of the Sachdev--Ye--Kitaev (SYK) model and the specific variants that will be relevant in the forthcoming sections. The SYK model describes a system of $N$ Majorana fermions with all-to-all random interactions. The interaction strengths are drawn independently from a Gaussian distribution, making the model a paradigmatic example of a strongly interacting quantum system. For a general $q$-body interaction, the SYK Hamiltonian is given by \cite{Sachdev_1993,Polchinski:2016xgd,Jevicki:2016bwu,Maldacena:2016hyu}. 
\begin{align}
    H_{\mathrm{SYK}_q}
    = i^{q/2} 
    \sum_{1 \le i_1 < i_2 < \cdots < i_q \le N}
    J_{i_1 i_2 \cdots i_q}\,
    \chi_{i_1}\chi_{i_2}\cdots\chi_{i_q},
\end{align}
where the $\chi_i$ are Majorana fermion operators satisfying $\{\chi_i,\chi_j\}=\delta_{ij}$. The couplings $J_{i_1 i_2 \cdots i_q}$ are drawn from a Gaussian distribution with zero mean and variance 
\begin{align}
    \left\langle J_{i_1 i_2 \cdots i_q}^2\right\rangle=\frac{(q-1)!J^2}{N^{q-1}} .
\end{align}
In the remainder of this section, we will analyze the time-evolution dynamics and finite-temperature properties of both the stabilizer R\'enyi entropy and the multipartite non-local SRE in the SYK$_4$ model. The analogous study for the various deformations of the SYK$_4$ model will be carried out in the sections that follow.

\subsection{Time Evolution: SRE and multipartite non-local SRE}

We begin our numerical analysis by examining the time evolution of the stabilizer R\'enyi entropy for the SYK$_4$ model, initialized in a product state,
\begin{equation}
   \ket{\psi(t)} = e^{-i H_{\mathrm{SYK}_4} t}\ket{000\cdots 0}.
\end{equation}
To gain insight into the contributions of Majorana strings of different lengths, we consider the \emph{restricted} SRE, defined as
\begin{equation}\label{res SRE}
    M_{\textrm{res}}(\rho)
    = -
      \log\!\left( \sum_{|S|=0}^{S_{\max}} p_S^2 \right)
      - n \log 2-S_2.
\end{equation}

In \cref{ResSRESYK4T}, we present the restricted SRE in \cref{RestSREtime7}, SRE for restricted number of qubits in \cref{SREtime7}, Multipartite non-local SRE in \cref{NLSREtime7} for the time-evolved state under the SYK$_4$ Hamiltonian, illustrating how the results depend on the cutoff $S_{\max}$, indicated on the right-hand side of the figure. Following this, we present the time evolution of the SRE and the multipartite non-local SRE for subsystem sizes ranging from one qubit to the full $N$-qubit system.

Physically, the restriction in \eqref{res SRE} should be read as a string-length resolved diagnostic of magic generation. In the qubit representation, the Majorana monomials organize into Pauli strings whose support size tracks how many degrees of freedom participate in the state’s non-stabilizer structure. The quantities $p_S$ therefore play a role analogous to a size distribution: at early times one expects the weight to remain concentrated at small $S$ (magic created by few-body strings), while chaotic dynamics transfers weight to larger $S$ (magic becomes genuinely many-body).  In this sense, the restricted SRE directly probes the scrambling of non-stabilizerness across the Hilbert space.

\begin{figure}[H]
  \centering
  
   \begin{subfigure}{.32\linewidth}
    \includegraphics[height=2.8cm,width=\linewidth]{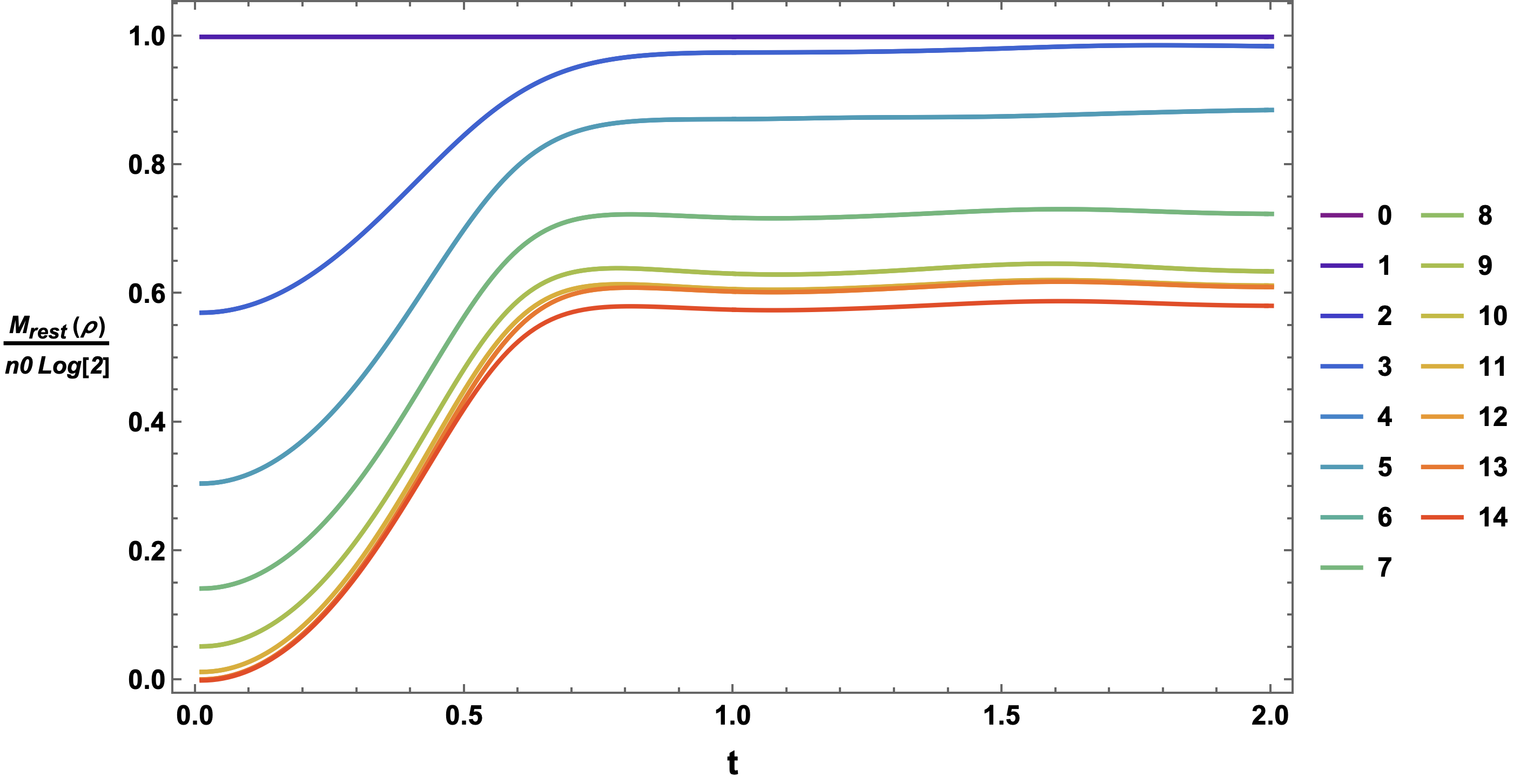}
      \caption{Restricted SRE}\label{RestSREtime7}
  \end{subfigure}\quad
  \begin{subfigure}{.32\linewidth}
    \includegraphics[width=\linewidth]{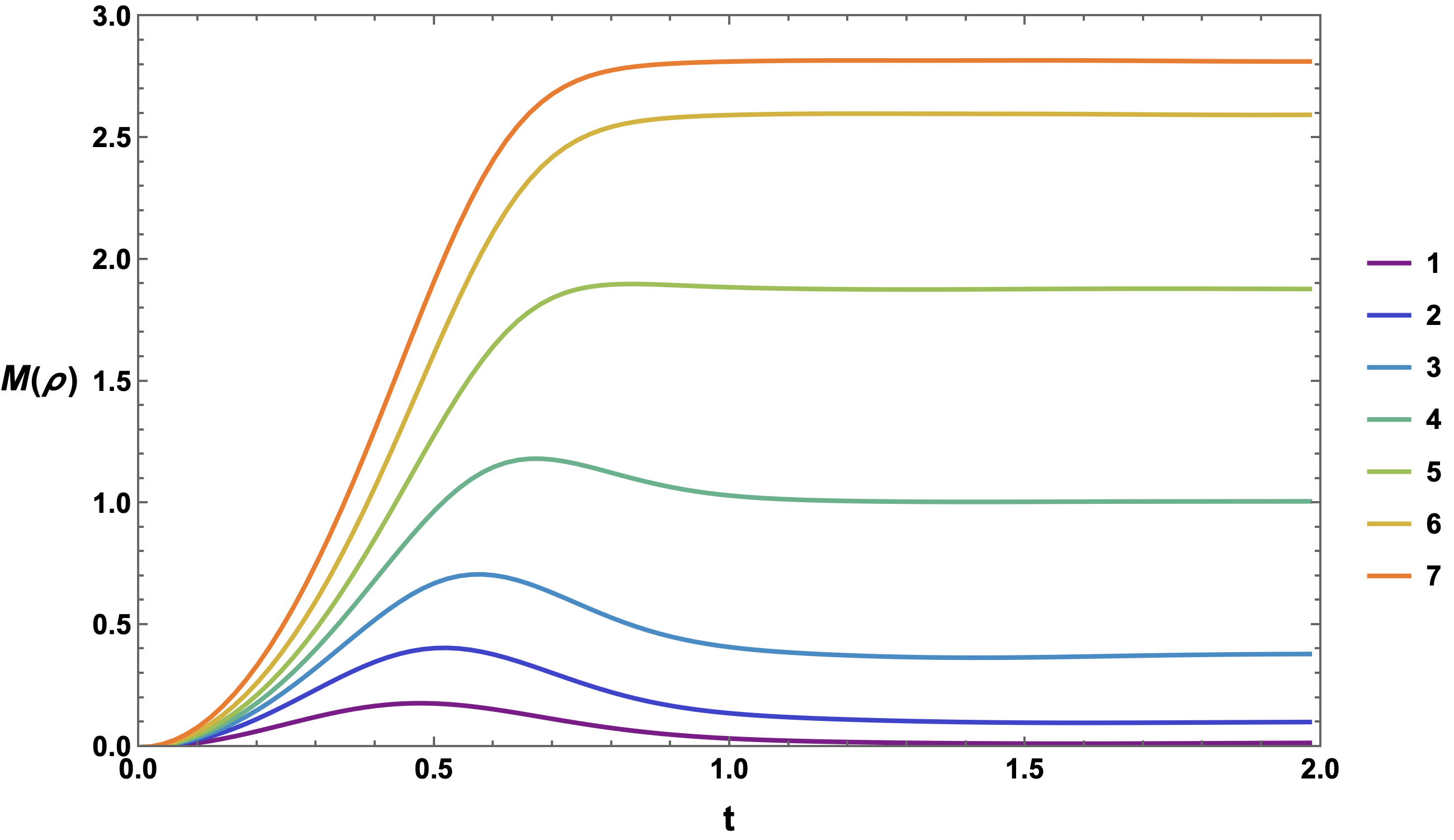}
      \caption{SRE}\label{SREtime7}
  \end{subfigure}
  \begin{subfigure}{.32\linewidth}
    \includegraphics[height=2.8cm,width=\linewidth]{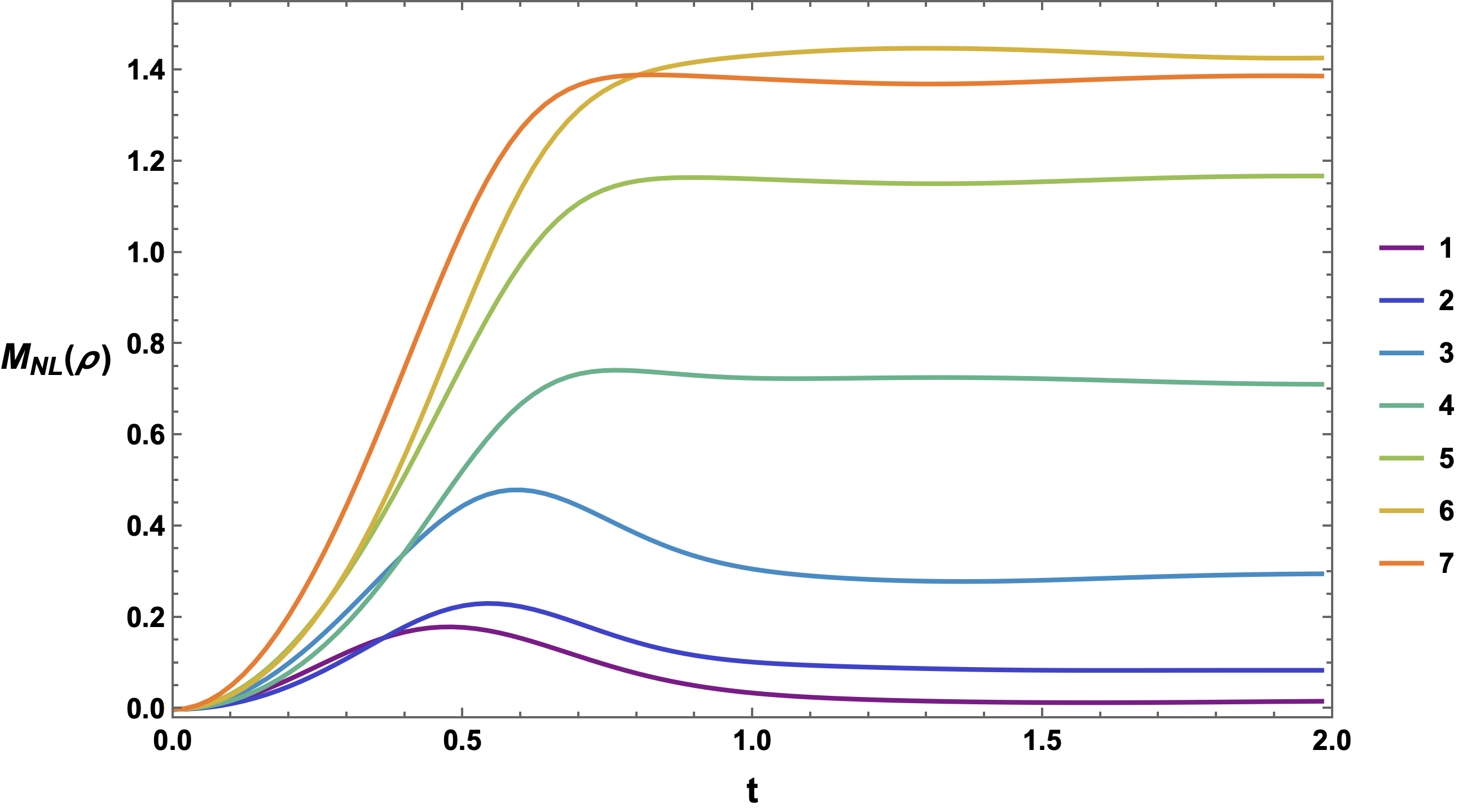}
      \caption{Multipartite non-local SRE}\label{NLSREtime7}
  \end{subfigure}
\caption{\footnotesize{Time evolution of the  SRE  for  $N=14$ in the SYK$_4$ model. }}\label{ResSRESYK4T}
\end{figure}

\subsection{Finite Temperature}
Having explored the time evolution of the restricted SRE, the full SRE, and the multipartite non-local SRE, we now turn to their behavior at finite temperature. In this context, we introduce temperature through two different constructions. The first is the thermal state,
\begin{equation}
    \rho_{\mathrm{th}} = e^{-\beta H_{\mathrm{SYK}_4}},
\end{equation}
which, in the present setting, is simply the Gibbs ensemble associated with the SYK$_4$ Hamiltonian. The second approach employs \emph{canonical thermal pure quantum} (TPQ) states, as introduced in Ref.~\cite{Sugiura:2013pla}. The unnormalized TPQ states are defined as
\begin{equation}
    \ket{\Psi(\beta)} = e^{-\frac{\beta}{2} H} \ket{\psi},
\end{equation}
where $\ket{\psi}$ is an unstructured reference state. A key property of TPQ states is that the ensemble-averaged expectation values of simple observables reproduce thermal behavior at inverse temperature $\beta$
\begin{equation}
    \frac{\overline{\langle \Psi_\beta | \mathcal{O} | \Psi_\beta \rangle}}
    {\overline{\langle \Psi_\beta | \Psi_\beta \rangle}}
    =
    \frac{1}{Z(\beta)} \operatorname{tr} \!\left( \mathcal{O} \, e^{-\beta H} \right).
\end{equation}
In practice, averaging over all possible reference states is computationally prohibitive. Therefore, in this work we take $\ket{\psi}$ to be a single random state, which is known to provide an accurate approximation of the TPQ state in typical many-body systems. We compare the restricted SRE defined in \cref{res SRE} for the TPQ state and the thermal state in \Cref{SRETPQthi} and \Cref{SRETPQthi2}. For small $\beta \ll 1$, corresponding to very high temperatures, we observe that the restricted SRE of the TPQ state is significantly larger than that of the thermal state. As the temperature is lowered, the difference between the SRE values of the two states decreases and eventually saturates to a constant value. We observe that the constant vanishes only when all the Majorana or Pauli strings are included.

At high temperatures, we observe a marked contrast: the TPQ state exhibits significantly larger SRE compared to the thermal state. This gap stems from their structural differences at infinite temperature. The thermal state approaches the maximally mixed state, which is computationally trivial and possesses zero magic. In contrast, the TPQ state at $\beta\to 0$ behaves as a random state, which is known to be highly entangled and magically complex. This implies that while the ensemble average is classical, individual typical microstates representing the ensemble are maximally hard to simulate. As temperature decreases, the thermal state acquires non-trivial structure from low-lying energy eigenstates, leading to a rapid growth in SRE. Simultaneously, the TPQ state projects onto the same low-energy subspace, causing the magic of both descriptions to converge toward that of the ground state. The persistence of high magic in TPQ states highlights their potential as robust candidates for demonstrating quantum advantage, as they encode thermal fluctuations into complex, non-stabilizer wavefunctions that evade efficient classical simulation.

This disparity admits an intriguing interpretation in the context of holography. If we regard the TPQ states as representative black hole microstates and the thermal state as the coarse-grained description of the black hole geometry, our results reveal a \textit{concealed complexity}: the immense computational hardness (magic) is carried by the specific microstructure of the black hole, while its thermodynamic description remains relatively simple and low-magic. This suggests that the \textit{quantumness} or non-stabilizerness of a black hole is washed out in the ensemble average, yet it is intrinsically present in the individual microstates responsible for the underlying unitary dynamics.

\begin{figure}[H]
  \centering
  \begin{subfigure}{.31\linewidth}
    \includegraphics[height=3.2cm,width=\linewidth]{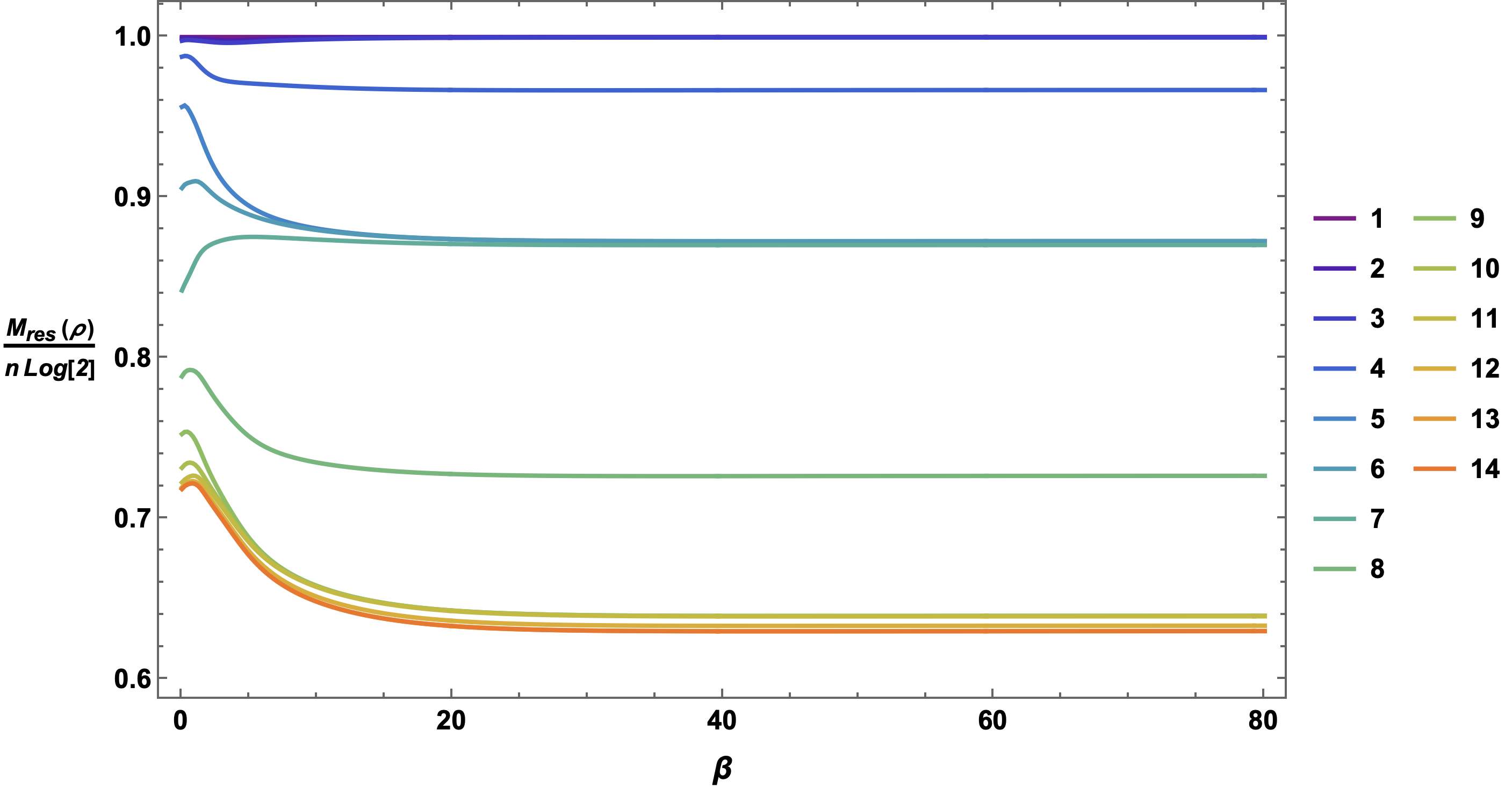}
    \caption{TPQ}\label{TPQvsThermal00}
  \end{subfigure}\quad
  \begin{subfigure}{.31\linewidth}
    \includegraphics[height=3.2cm,width=\linewidth]{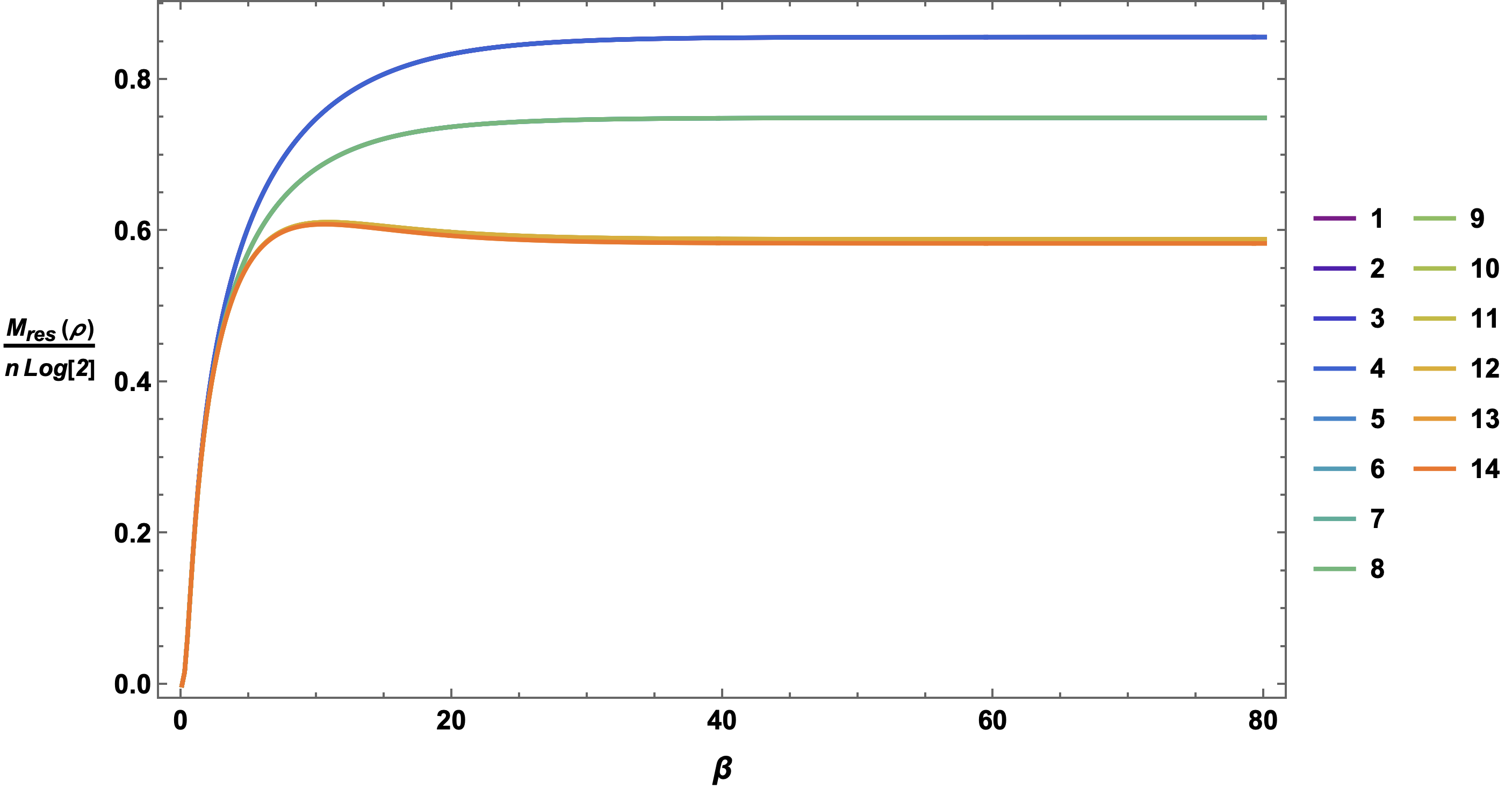}
    \caption{Thermal state}\label{TPQvsThermal01}
  \end{subfigure}\quad
  \begin{subfigure}{.31\linewidth}
    \includegraphics[height=3.2cm,width=\linewidth]{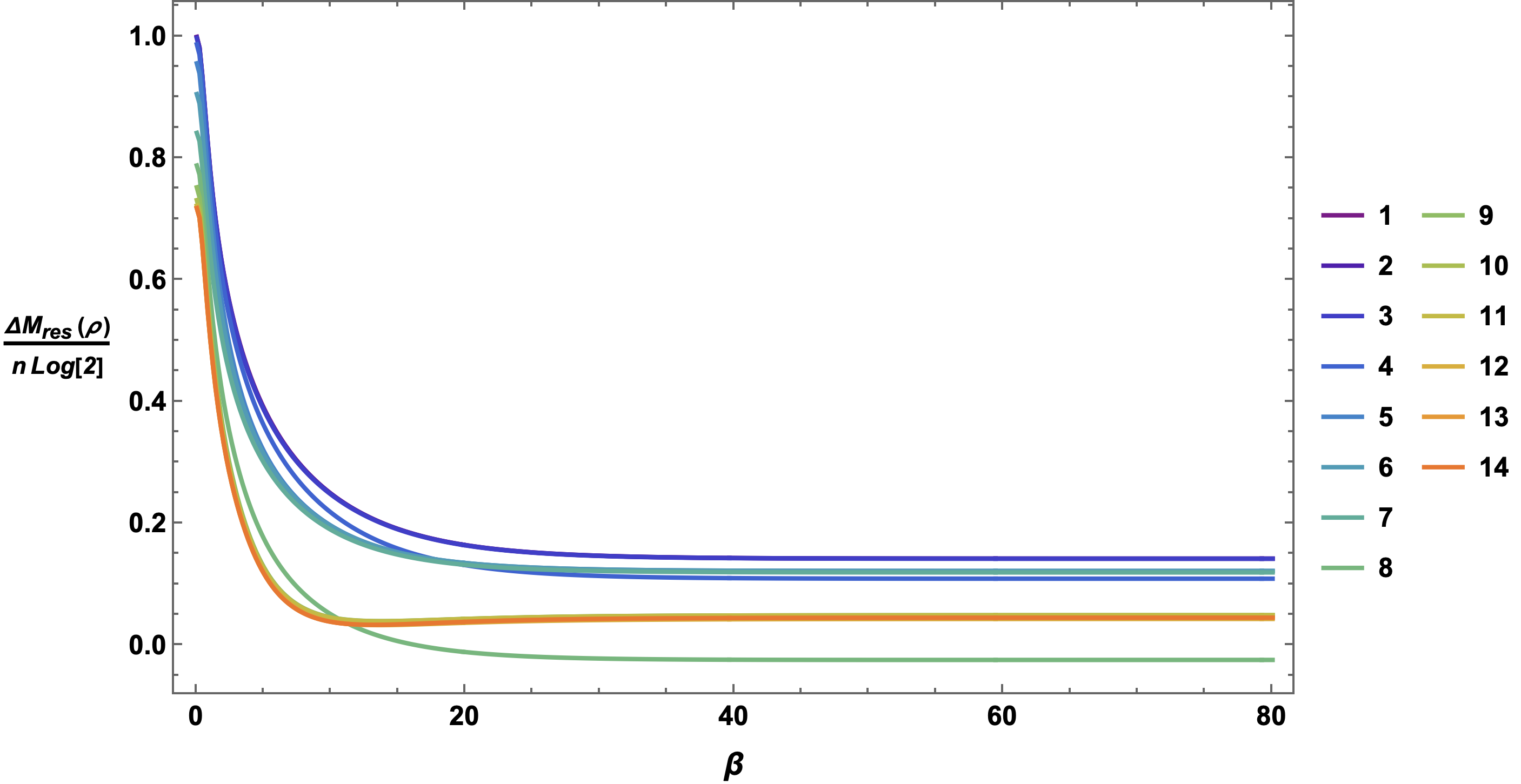}
    \caption{Diff b/w TPQ \& thermal}\label{TPQvsThermal02}
  \end{subfigure}
  \caption{\footnotesize{Stabilizer Rényi entropy as a function of temperature, restricted by the length of Majorana strings. The corresponding values of length of the string are indicated on the right.}}
  \label{SRETPQthi}
\end{figure}

\begin{figure}[H]
  \centering
  \begin{subfigure}{.32\linewidth}
\includegraphics[width=\linewidth]{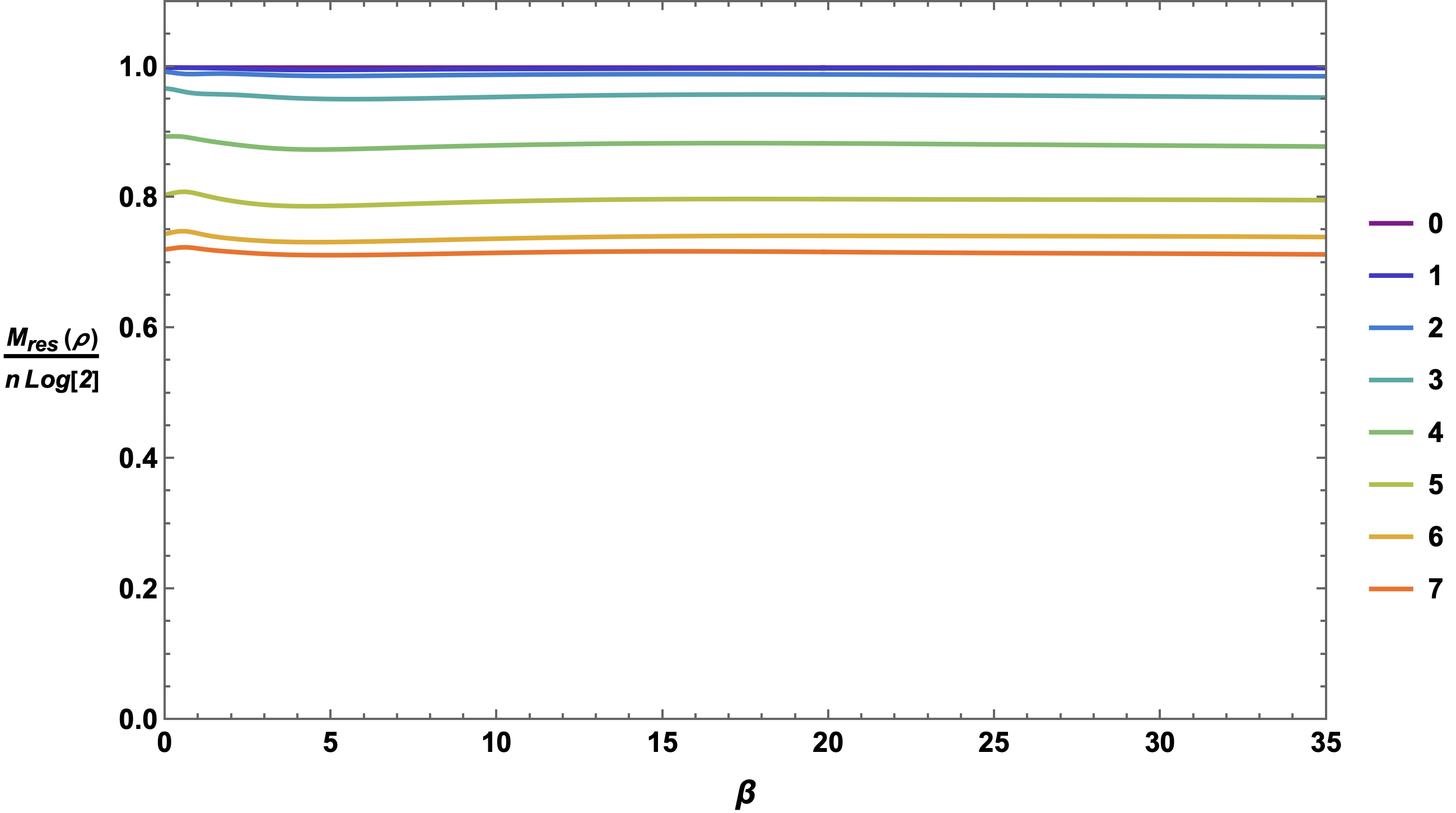}
\caption{TPQ}\label{TPQvsThermal5}
  \end{subfigure}
  \begin{subfigure}{.33\linewidth}
\includegraphics[width=\linewidth]{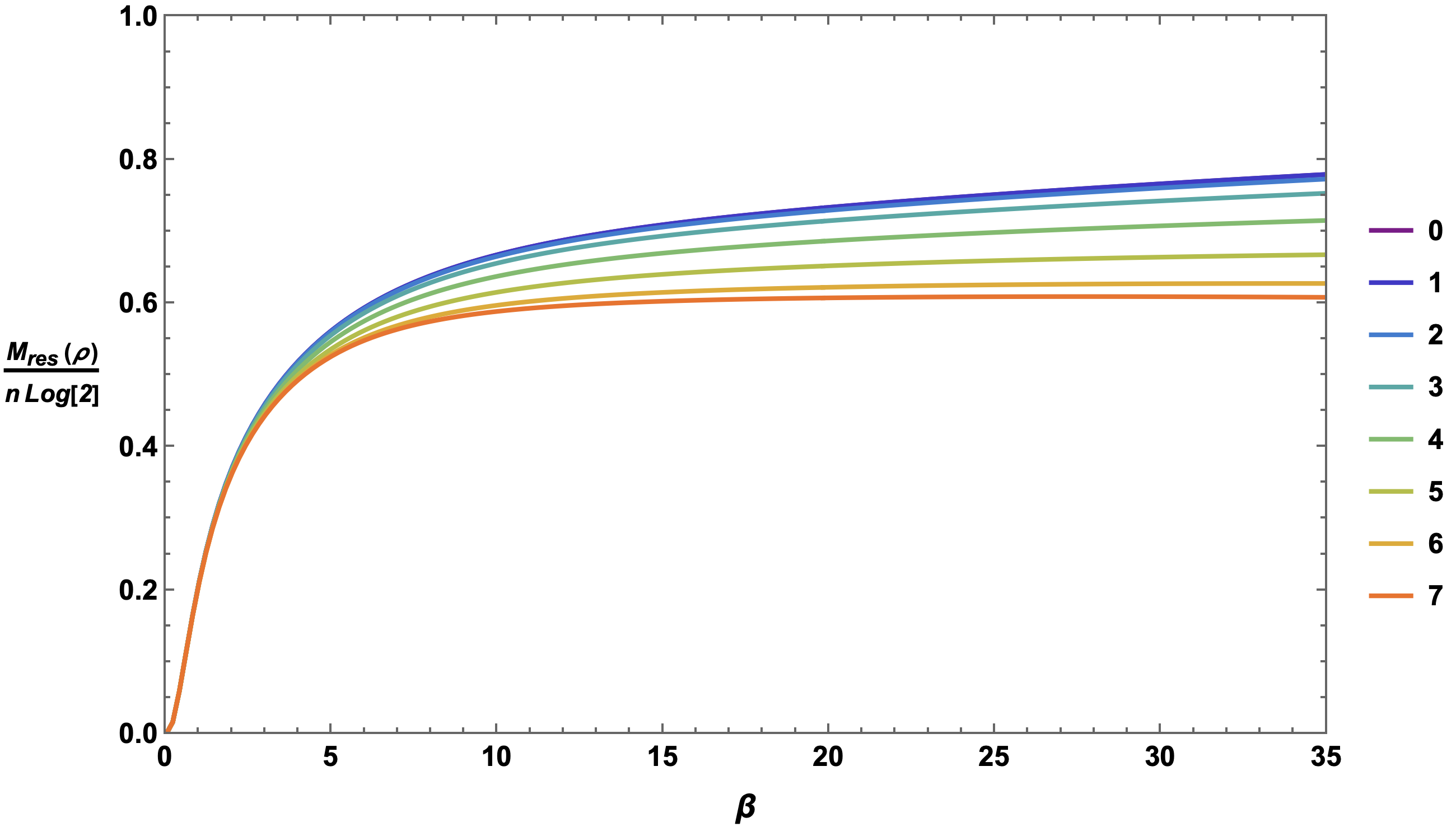}
    \caption{Thermal state}\label{TPQvsThermal6}
  \end{subfigure}
  \begin{subfigure}{.33\linewidth}
\includegraphics[width=\linewidth]{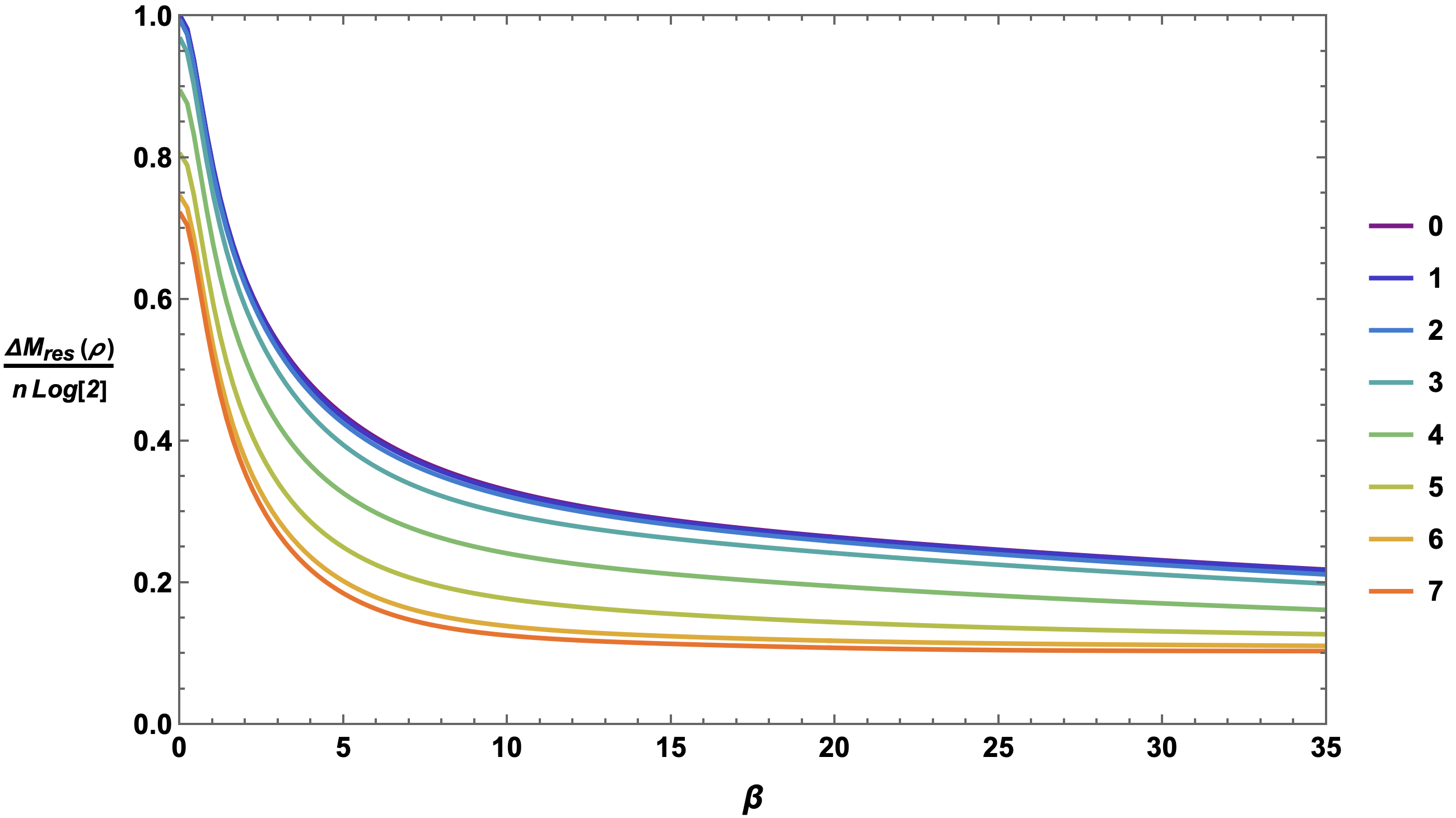}
    \caption{Diff b/w TPQ \& thermal}\label{TPQvsThermal7}
  \end{subfigure}
\caption{\footnotesize{Stabilizer Rényi entropy as a function of temperature, restricted by  the support of Pauli strings on the number of qubits indicated on the right.}}
\label{SRETPQthi2}
\end{figure}
Finally, in \Cref{SRETPQthi3} we present the multipartite non-local SRE for the TPQ and thermal states of the SYK$_4$ model as a function of inverse temperature, obtained by varying the number of qubits  up to the total system size. Once again, we find that the difference between the two rapidly decreases and saturates to a small value, which vanishes only for the case where all qubits are included.
\begin{figure}[H]
  \centering
  \begin{subfigure}{.32\linewidth}
\includegraphics[width=\linewidth]{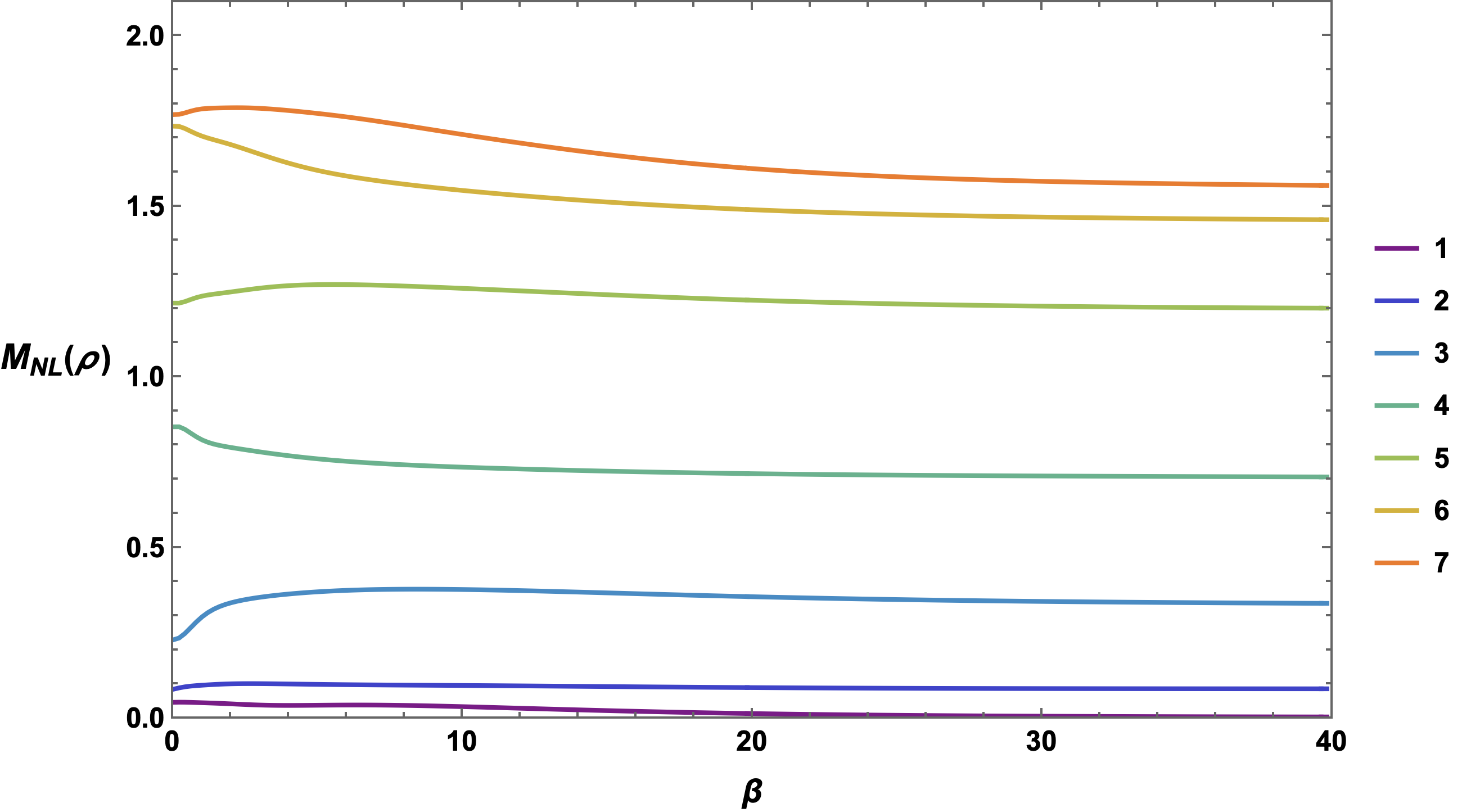}
\caption{TPQ}\label{NLTPQvsThermal00}
  \end{subfigure}
  \begin{subfigure}{.33\linewidth}
\includegraphics[width=\linewidth]{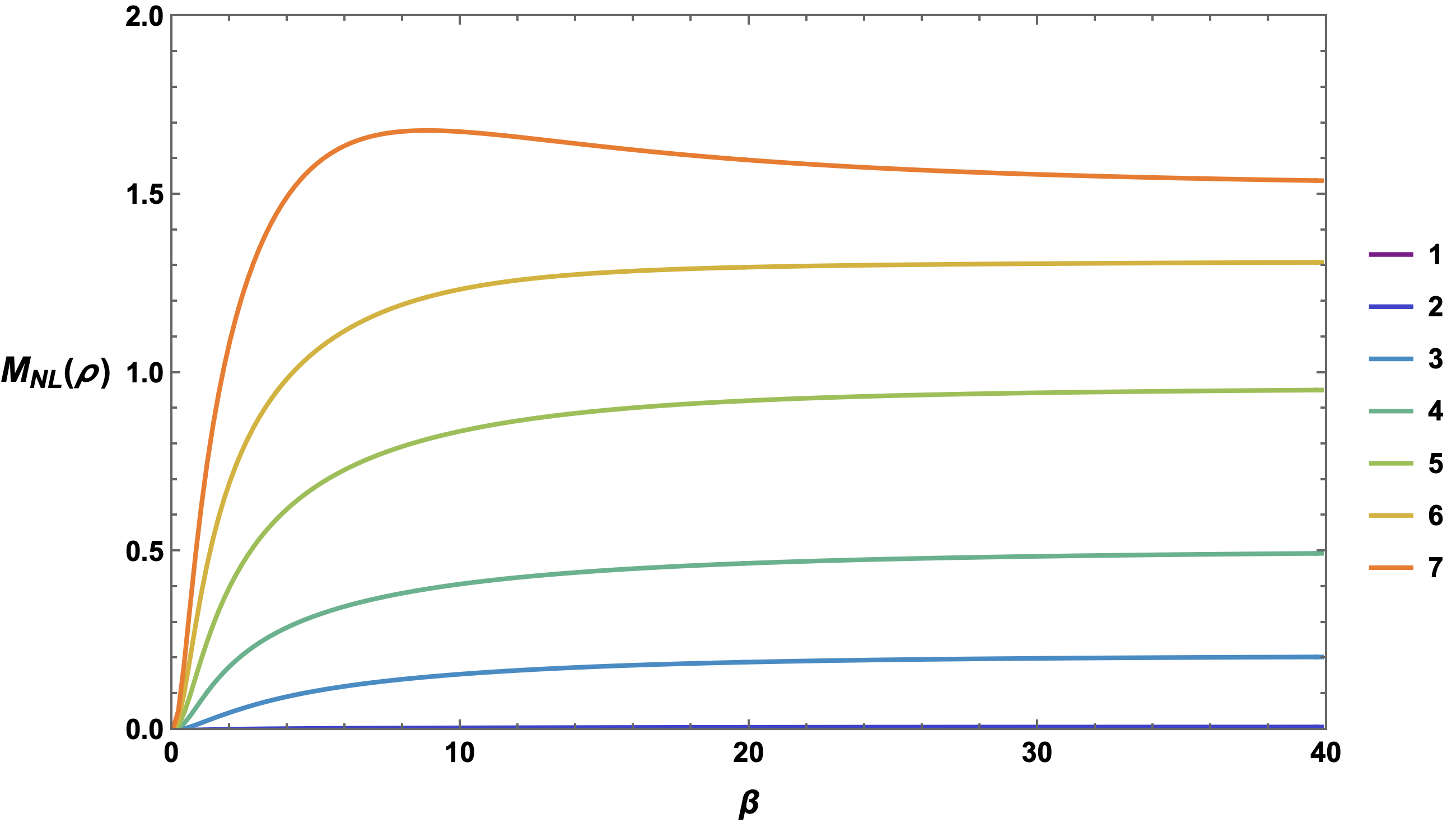}
    \caption{Thermal state}\label{NLTPQvsThermal01}
  \end{subfigure}
  \begin{subfigure}{.33\linewidth}
\includegraphics[width=\linewidth]{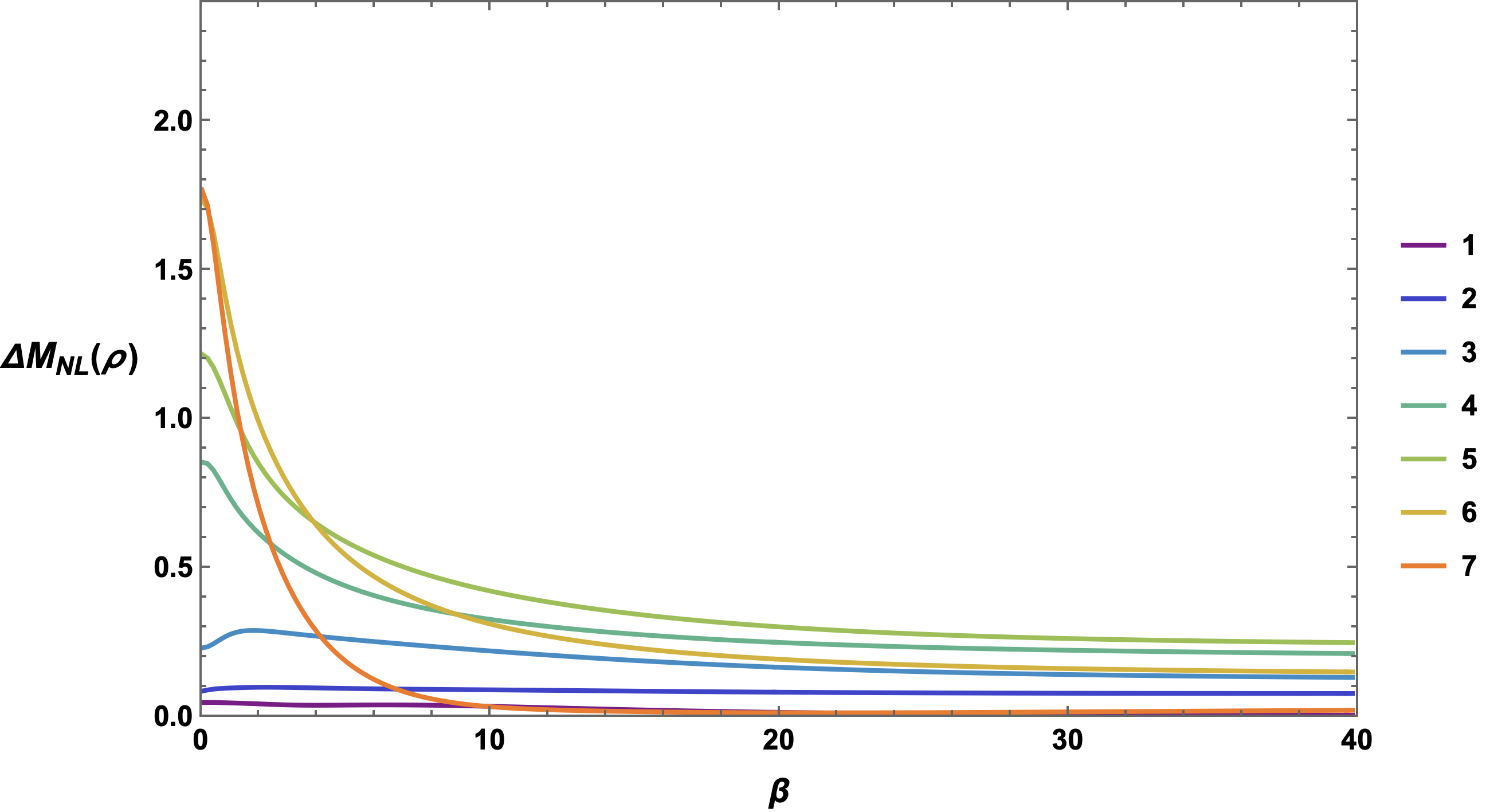}
    \caption{Diff b/w TPQ \& thermal}\label{NLTPQvsThermal02}
  \end{subfigure}
\caption{\footnotesize{For $N=14$ ($n=7$ qubits), Multipartite non-local Stabilizer Rényi entropy as a function of inverse temperature. The list on right indicates the number of qubits considered for the computation of multipartite SRE.}}
\label{SRETPQthi3}
\end{figure}
\subsubsection{Probability Distribution}
\subsubsection*{Based on the length of the Majorana string}
To further investigate the finite-temperature behavior of the SRE for both the TPQ and thermal states, we compute the average probability contribution to the SRE as a function of fixed Majorana string length. The results are shown in \Cref{ProbMN8}. At $\beta = 0$, the thermal state receives its entire contribution solely from the identity operator, whereas the TPQ state exhibits a distribution peaked at intermediate string lengths. As $\beta$ increases, the contribution of the identity operator to the SRE probability of the thermal state diminishes, and nontrivial strings restricted to lengths which are multiples of $4$ begin to contribute. The vanishing behavior of the probability for Majorana strings of lengths given by multiples of $4$ can be attributed to the structure of the thermal state governed by the SYK$_4$ Hamiltonian, which contains only quartic interactions. As a consequence, only operator strings whose lengths are multiples of four contribute appreciably. As we will show later, introducing a quadratic term through a mass deformation alters this pattern: once the deformation is added, probability contributions from Majorana strings of length equal to any even integer begin to appear, reflecting the modified operator content of the theory. Upon further increasing $\beta$, the probability distributions of the TPQ and thermal states gradually converge, and their profiles begin to closely resemble one another.

\begin{figure}[H]
  \centering
  \begin{subfigure}{.4\linewidth}
    \includegraphics[height=3.5cm,width=\linewidth]{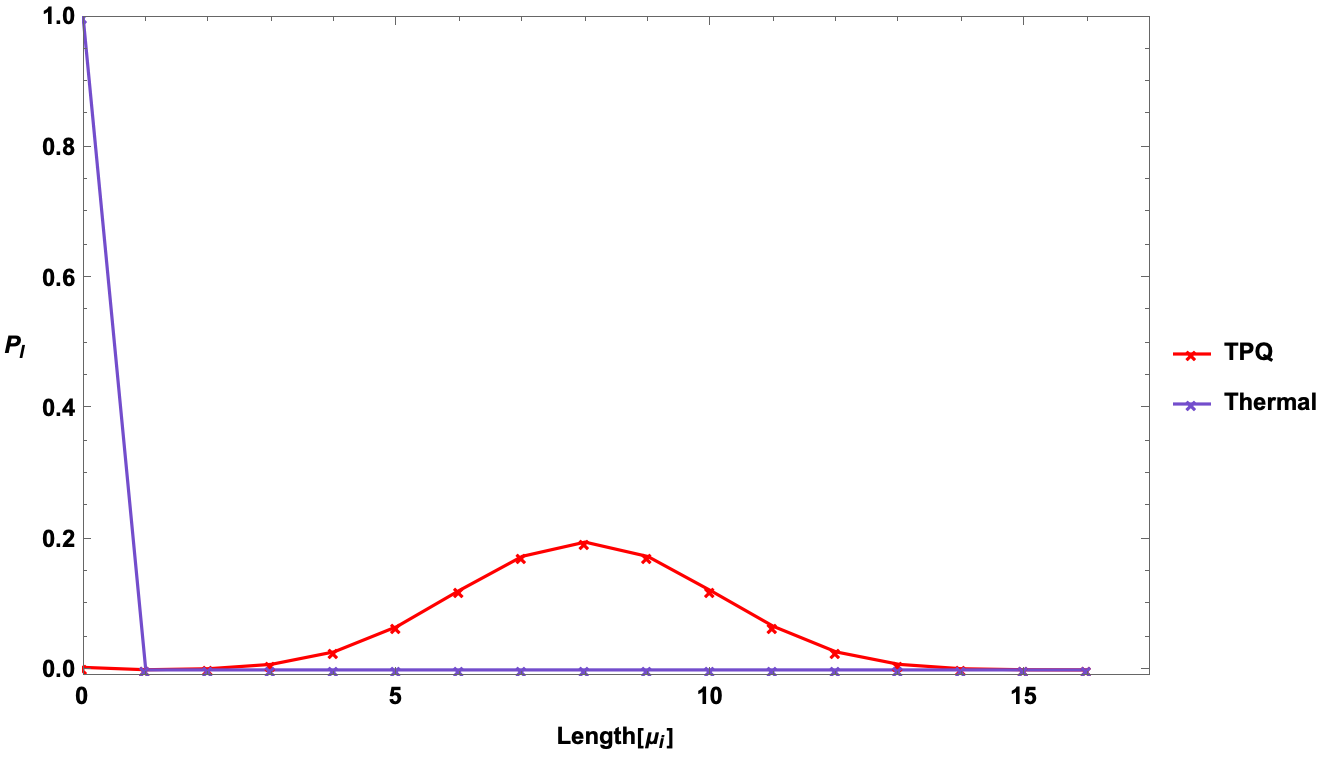}
    \caption{$\beta=0$}\label{}
  \end{subfigure}
  \begin{subfigure}{.4\linewidth}
    \includegraphics[height=3.5cm,width=\linewidth]{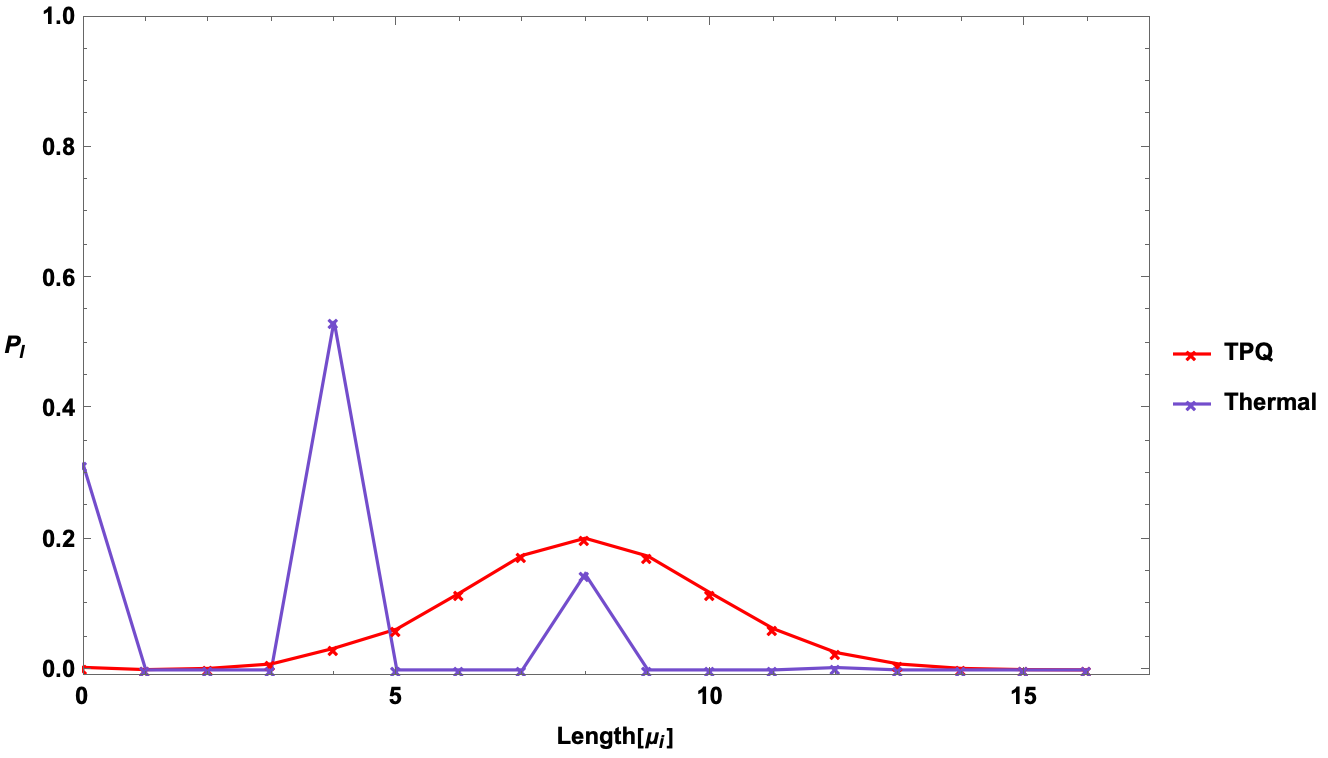}
    \caption{$\beta=1$}\label{}
  \end{subfigure}
  \begin{subfigure}{.4\linewidth}
\includegraphics[height=3.5cm,width=\linewidth]{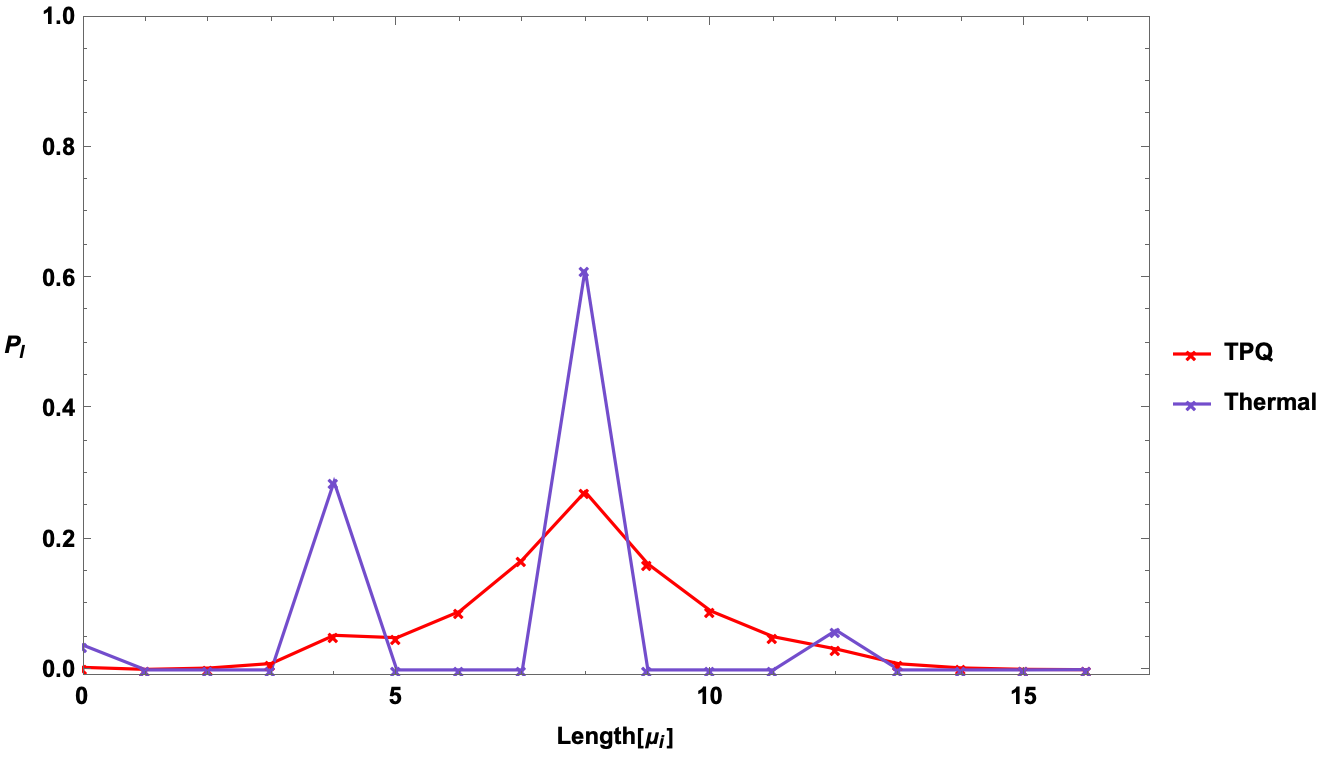}
    \caption{$\beta=5$}\label{}
  \end{subfigure}
   \begin{subfigure}{.4\linewidth}
\includegraphics[height=3.5cm,width=\linewidth]{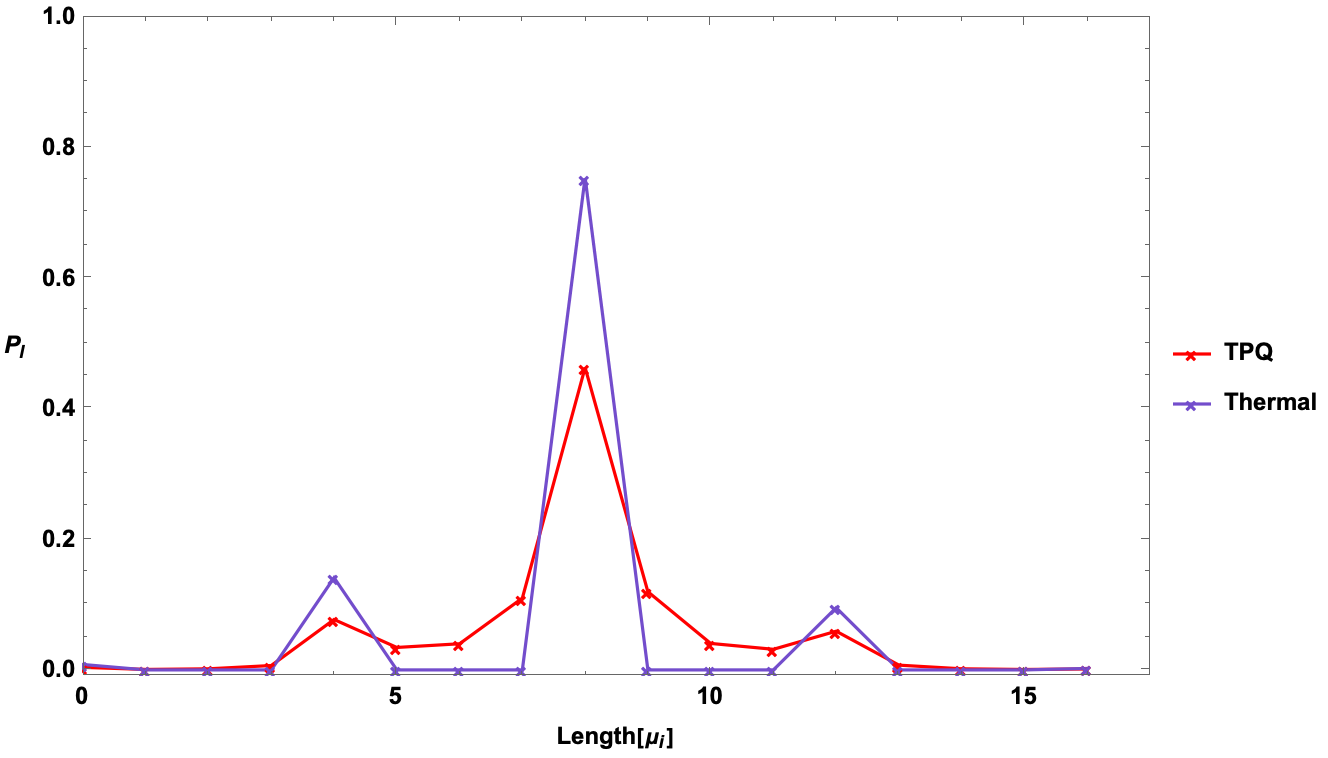}
    \caption{$\beta=20$}\label{}
  \end{subfigure}
  \caption{\footnotesize{For $N = 16$, $P_l = \sum_i \frac{\mathrm{Tr}(\rho \mu_i)^2}{2^{n}\mathrm{Tr}(\rho^2)}$ (summation is over all Majorana strings which have equal length $l$) quantifies the total probability distribution over Majorana strings of varying length for \textbf{SYK$_4$} model. In the thermal ensemble, nonzero contributions are observed only from the identity component and strings of lengths given by multiples of 4 whereas in the TPQ state this behavior emerges only at sufficiently large $\beta$.
  }}
  \label{ProbMN8}
\end{figure}

\subsubsection*{Based the support on the number of qubits (Pauli basis) }

After understanding how Majorana string length influences the probability contributions to the SRE, we now shift our focus to how these contributions are organized according to the number of qubits on which the strings have support.
Different Majorana strings act on varying numbers of qubits, allowing us to analyze the SRE contributions from strings that share the same support size. This behavior is illustrated in \Cref{ProbPauliSYK4} for both the TPQ and thermal states. As in the previous case, at infinite temperature the probabilities contributing to the SRE of the thermal state originate solely from the identity operator. In contrast, once $\beta$ becomes nonzero, the thermal state begins to receive contributions from Majorana strings with support on any number of qubits. On further increasing $\beta$ its probability distribution quickly approaches that of the TPQ state.

\begin{figure}[H]
  \centering
  \begin{subfigure}{.4\linewidth}
    \includegraphics[height=3.5cm,width=\linewidth]{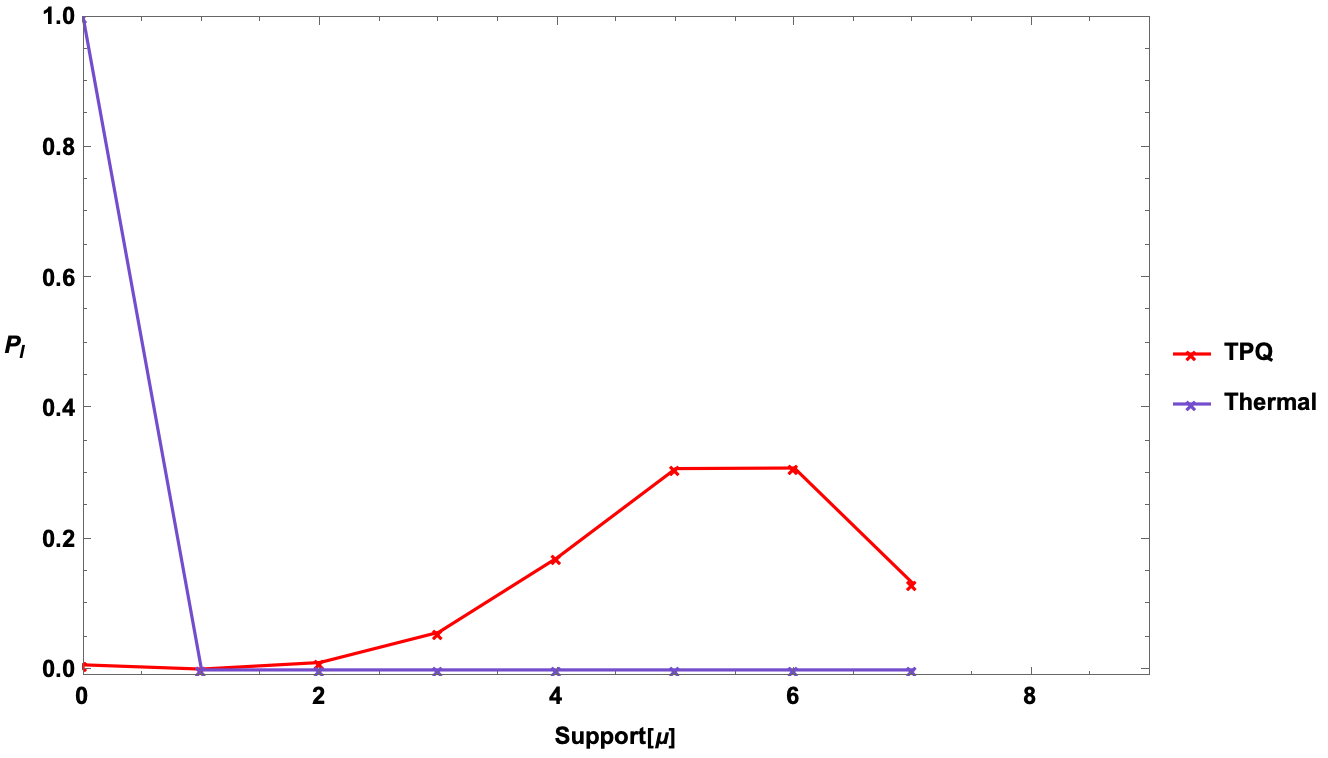}
    \caption{$\beta=0$}\label{}
  \end{subfigure}
  \begin{subfigure}{.4\linewidth}
    \includegraphics[height=3.5cm,width=\linewidth]{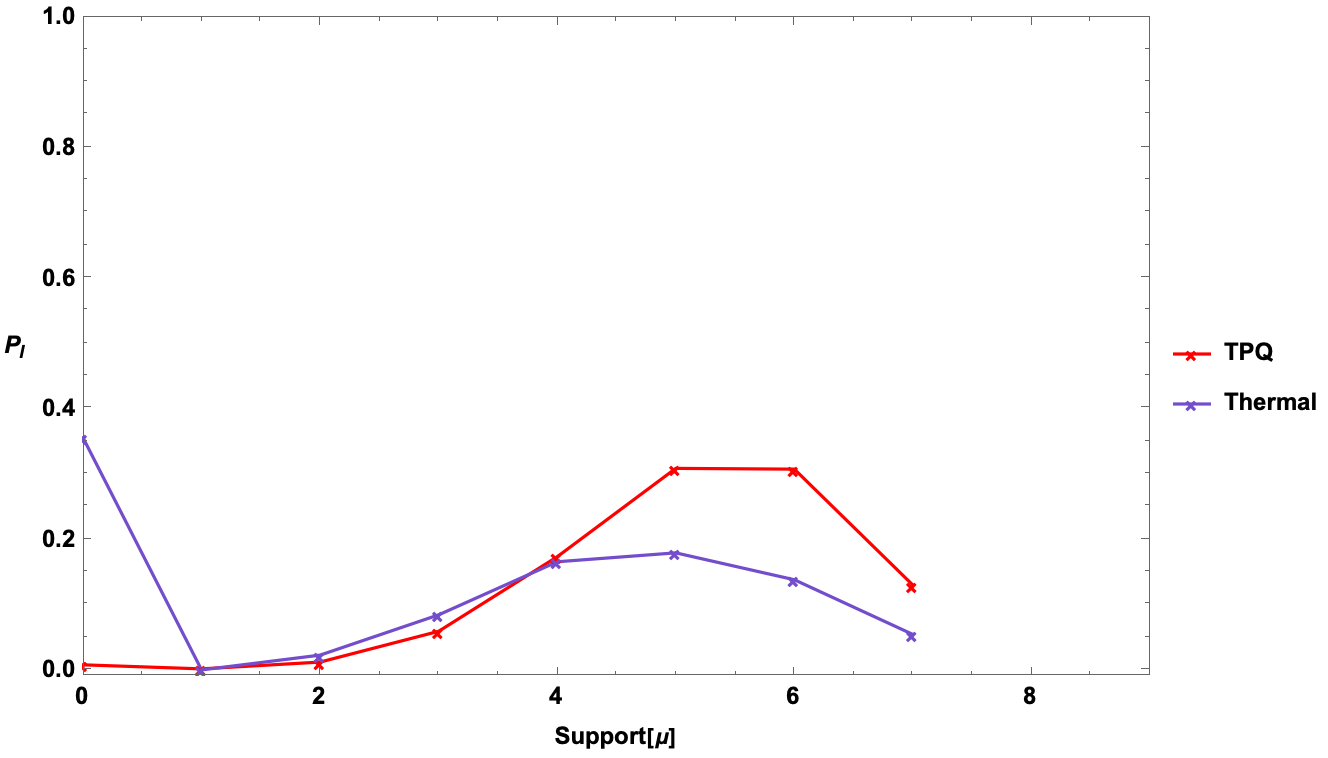}
    \caption{$\beta=1$}\label{}
  \end{subfigure}
  \begin{subfigure}{.4\linewidth}
\includegraphics[height=3.5cm,width=\linewidth]{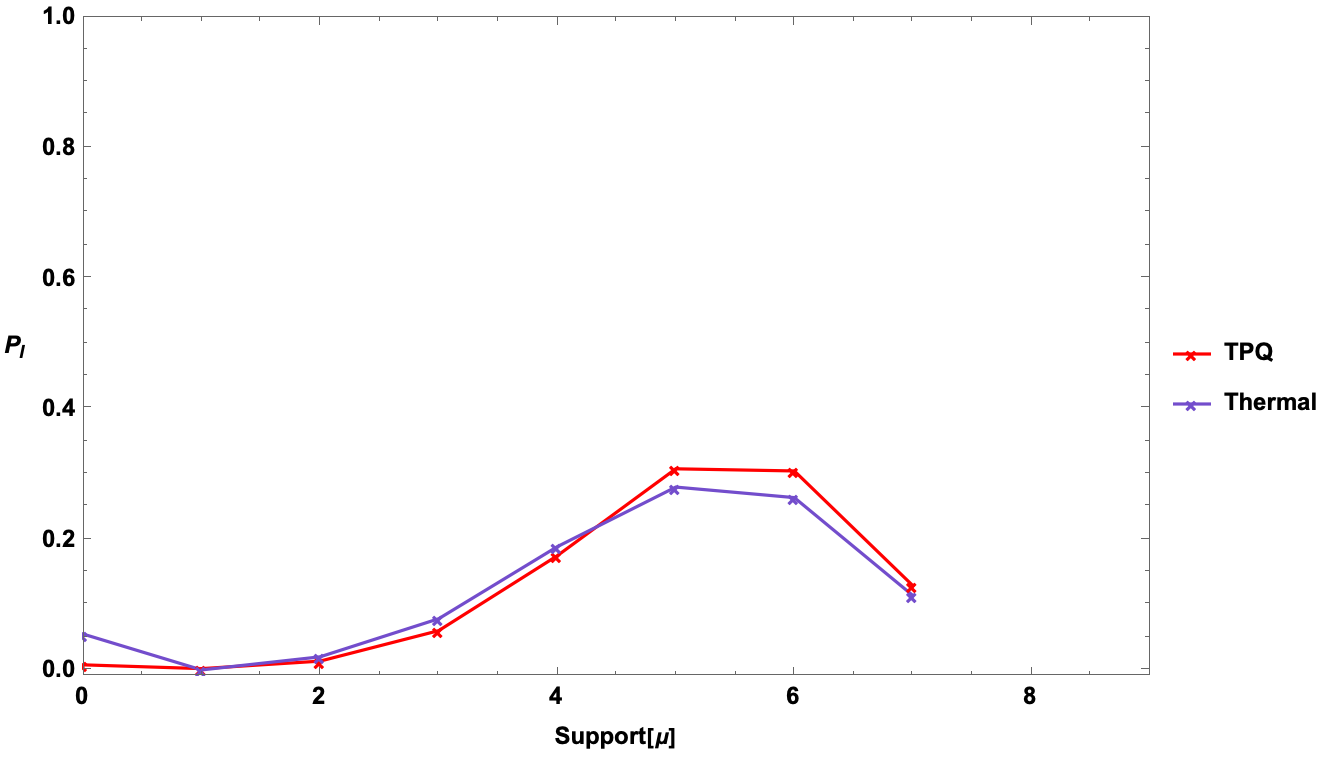}
    \caption{$\beta=5$}\label{}
  \end{subfigure}
   \begin{subfigure}{.4\linewidth}
\includegraphics[height=3.5cm,width=\linewidth]{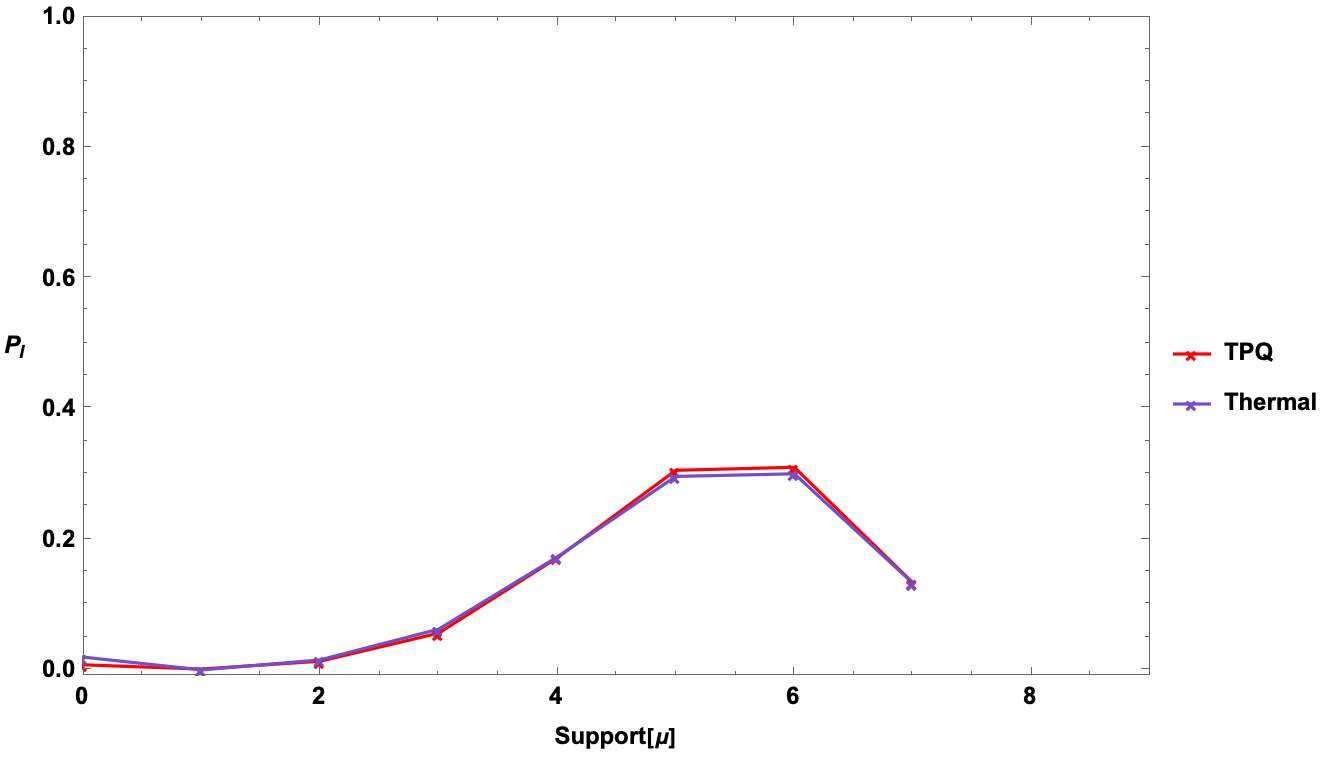}
    \caption{$\beta=20$}\label{}
  \end{subfigure}
  \caption{\footnotesize{For $N = 14$, $P_l = \sum_i \frac{\mathrm{Tr}(\rho \mu_i)^2}{\mathrm{Tr}(\rho^2)}$ (summation is over all Majorana strings which have support on equal number of the qubits) quantifies the total probability distribution over Majorana strings of varying support size for \textbf{SYK$_4$} model. 
  }}
  \label{ProbPauliSYK4}
\end{figure}

\section{Mass Deformed SYK}\label{sec 4}
It is well known that the SYK model with quartic interactions ($q=4$) exhibits maximal quantum chaos, whereas the quadratic model ($q=2$) is integrable. By taking a weighted linear combination of these two limits, one obtains a Hamiltonian that interpolates between integrable and chaotic behavior. This deformation is commonly referred to as the \emph{mass-deformed SYK model}  \cite{Banerjee:2016ncu,Garcia-Garcia:2017bkg,Nosaka:2018iat,Nandy:2022hcm}. The Hamiltonian takes the form
\begin{equation}
    H = (1-g)\, H_{\mathrm{SYK}_4} + g\, H_{\mathrm{SYK}_2},
\end{equation}
where the parameter $g \in [0,1]$ tunes the relative contribution of the chaotic quartic interactions and the integrable quadratic term. As $g$ is varied, the system exhibits a crossover, in the large-$N$ limit, a phase transition between the maximally chaotic SYK$_4$ regime and the integrable SYK$_2$ regime. 

\subsection{SRE}
In \Cref{SREMDSYK8}, we present the time evolution of the SRE in the mass-deformed SYK$_4$ model. We observe that, for all values of the deformation parameter $g$, the SRE quickly grows toward the saturation value characteristic of the SYK$_2$ model (i.e., $g=1$), consistent with the early-time behavior shown in \Cref{SYK7q2GHZ}. If we evolve the system for sufficiently long times, the SRE in every case except $g=1$ eventually reaches the saturation value of the SYK$_4$ model (corresponding to $g=0$). The noteworthy feature here is that the time required to reach this final saturation value increases monotonically as $g$ is varied from $0$ to $1$, as illustrated in \Cref{SYK7q1GHZ}.

\begin{figure}[H]
  \centering
  \begin{subfigure}{.3\linewidth}
\includegraphics[height=3.5cm,width=\linewidth]{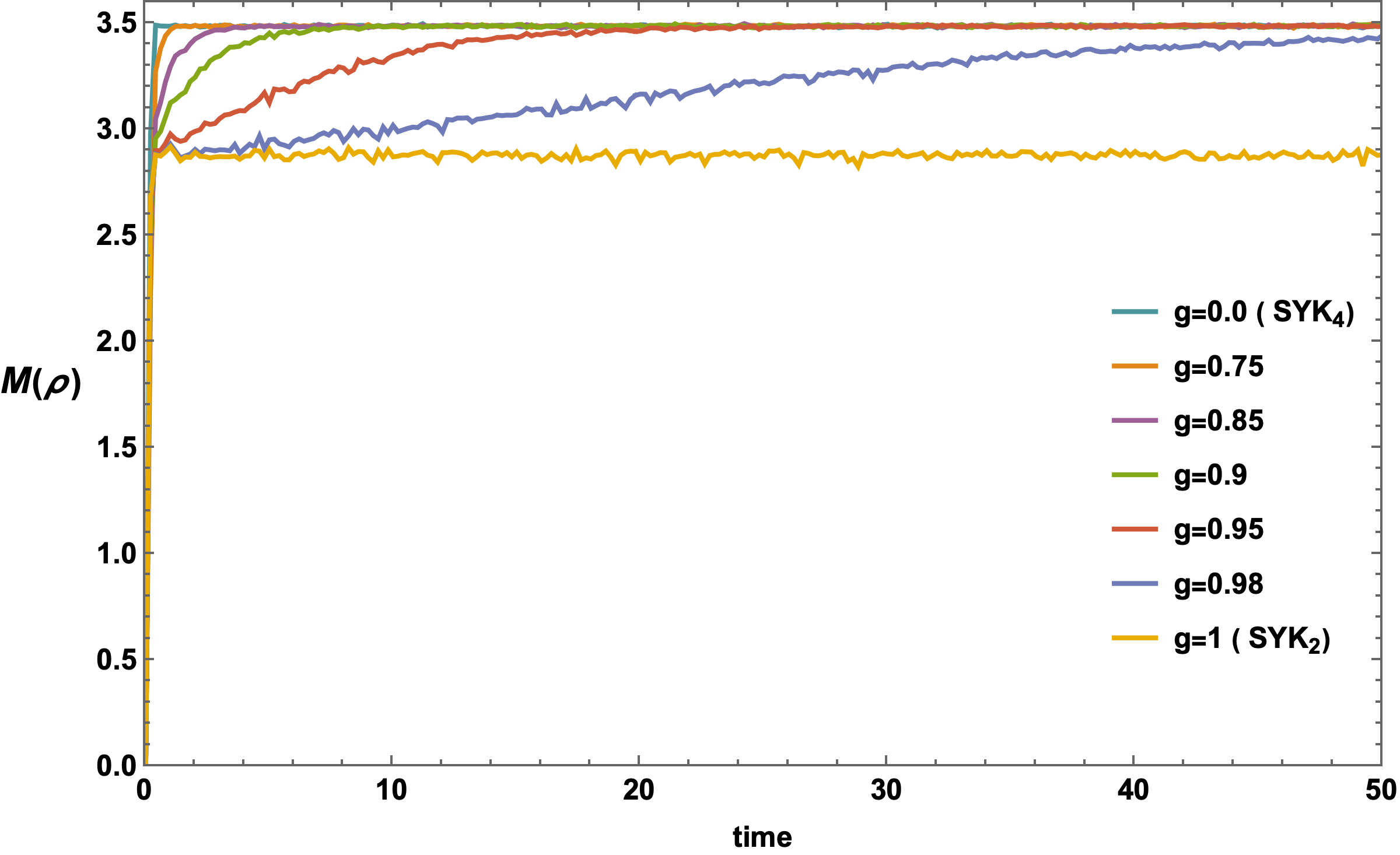}
    \caption{SRE}\label{SYK7q1GHZ}
  \end{subfigure}
  \begin{subfigure}{.3\linewidth}
\includegraphics[height=3.5cm,width=\linewidth]{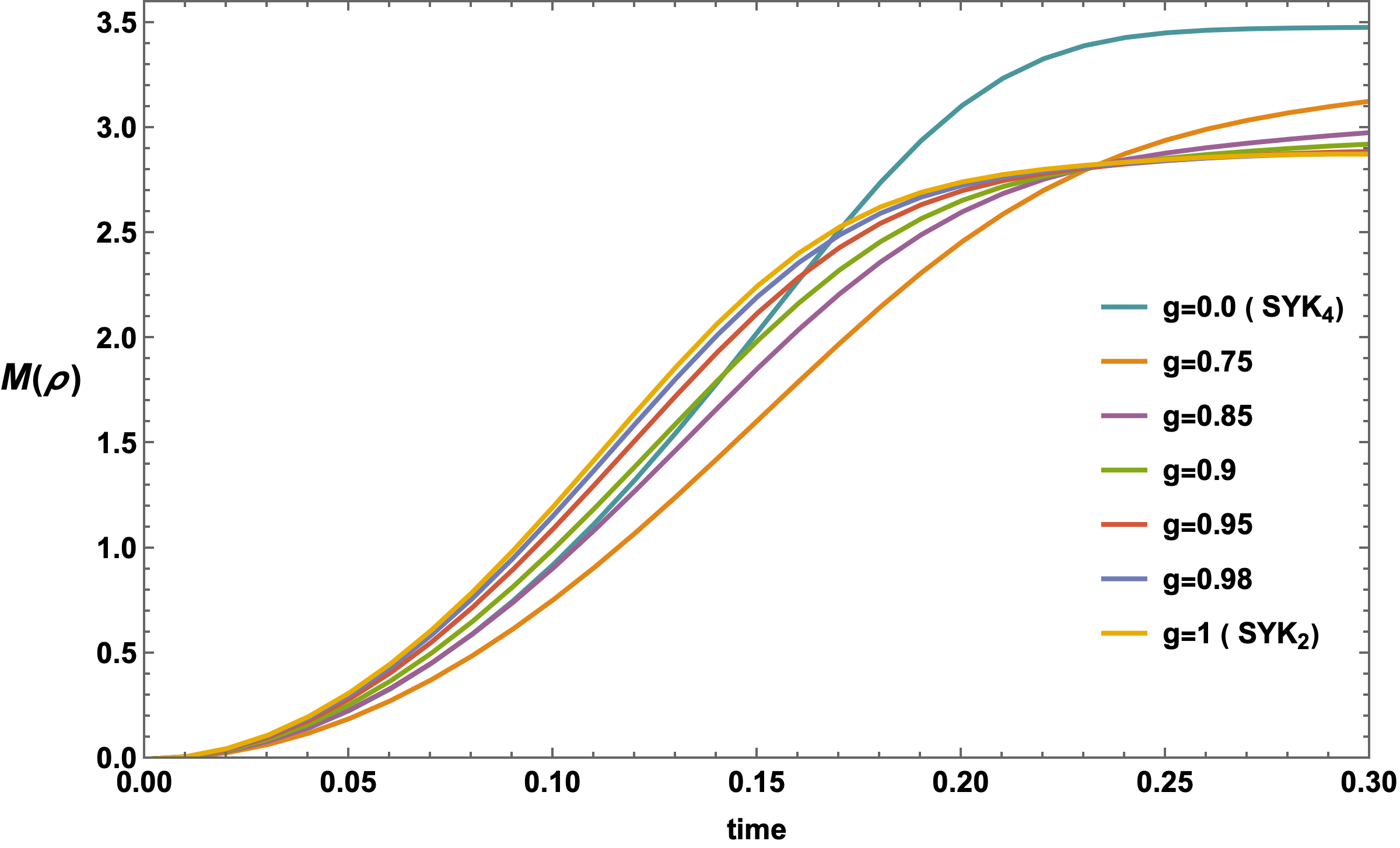}
    \caption{Early time growth of SRE}\label{SYK7q2GHZ}
  \end{subfigure}
  \begin{subfigure}{.3\linewidth}
\includegraphics[height=3.5cm,width=\linewidth]{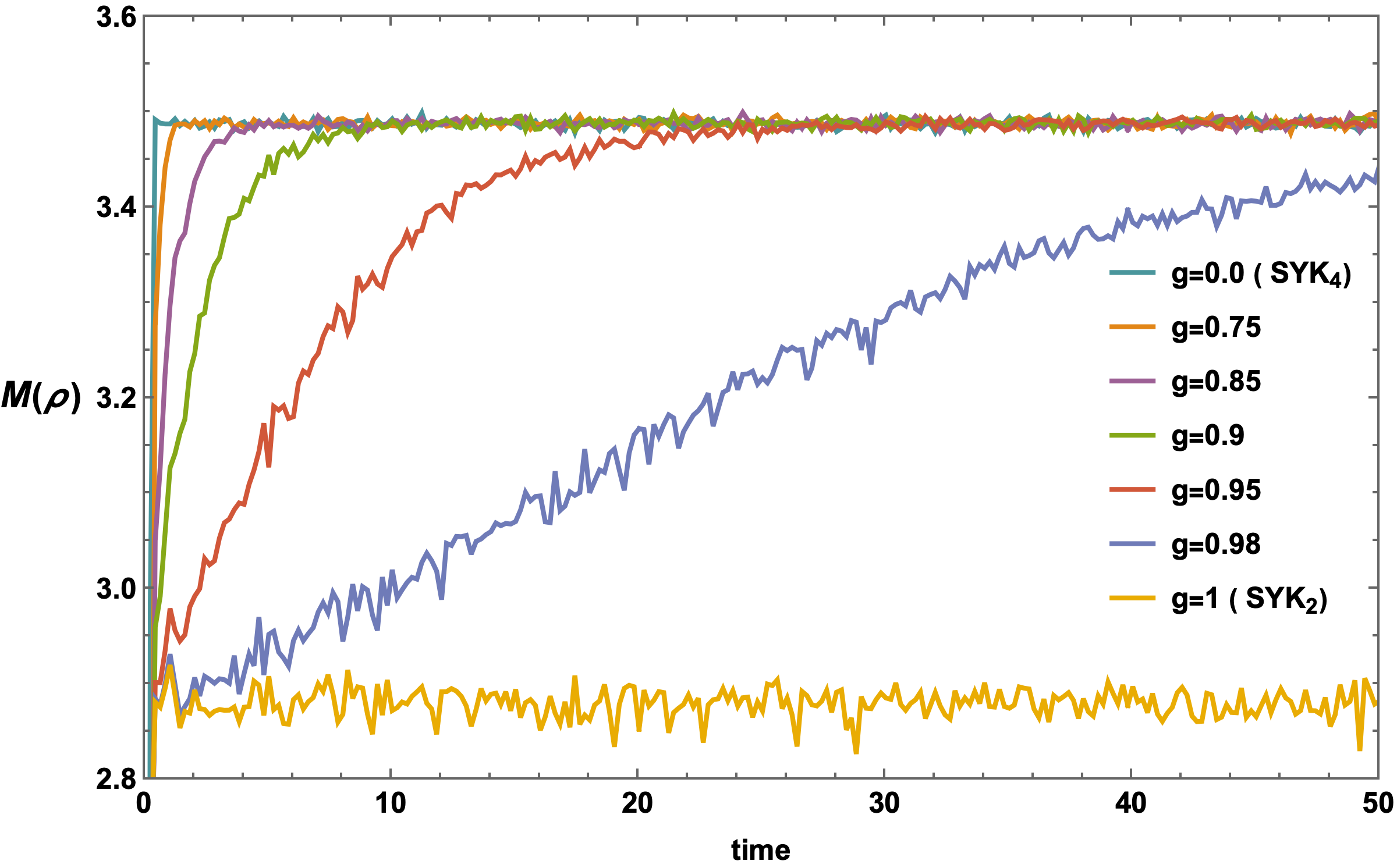}
    \caption{Late time growth of SRE}\label{SYK7q1GHZ}
  \end{subfigure}
  \caption{\footnotesize{For $N=16$ (8 qubits), panels (a) shows the time evolution of the stabilizer Rényi entropy (SRE) with $\ket{GHZ_7}$ as the initial state averaged over 25 samples. Panel (b) shows the initial growth where as panel (c) shows the late time growth. The plots highlight the transition of the late-time saturation value from the SYK$_2$ regime to the SYK$_4$ regime as the coupling $g$ deviates from unity. }}\label{SREMDSYK8}
\end{figure}
We next examine the effect of mass deformation on the probability distributions contributing to the SRE, both as a function of Majorana string length (\Cref{MDSYKProb1}) and as a function of the number of qubits on which the strings have support (\Cref{MDSYKprob2}). Quite interestingly, the Majorana–string–length–resolved plots reveal that introducing the mass deformation causes the SRE to receive contributions from strings of length equal to any multiple of $2$, in contrast to the pure SYK$_4$ case ($g=0$), where contributions arise only from lengths that are multiples of $4$. As described earlier, this behavior can be traced to the presence of the quadratic mass deformation term in the modified SYK$_4$ Hamiltonian.

\begin{figure}[H]
  \centering
  \begin{subfigure}{.3\linewidth}
    \includegraphics[height=3.5cm,width=\linewidth]{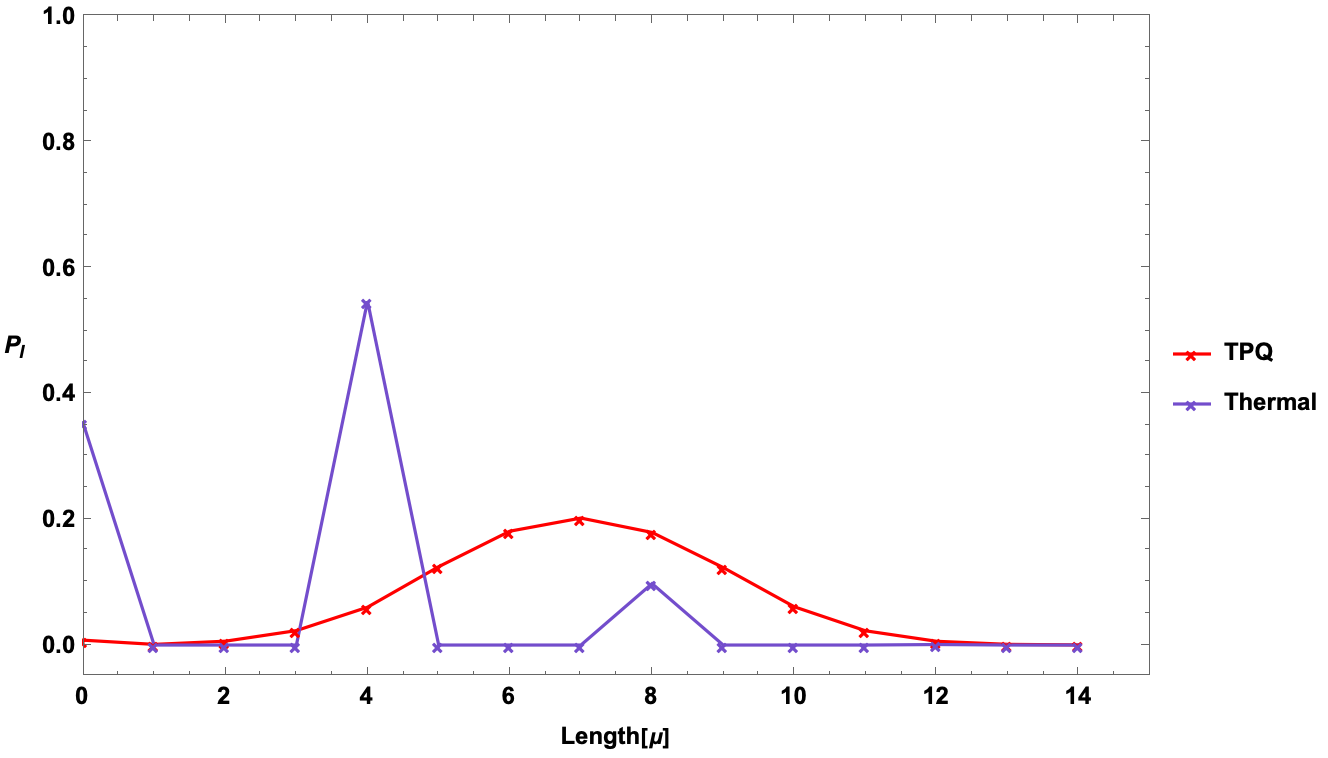}
    \caption{$g=0$}\label{}
  \end{subfigure}
  \begin{subfigure}{.3\linewidth}
    \includegraphics[height=3.5cm,width=\linewidth]{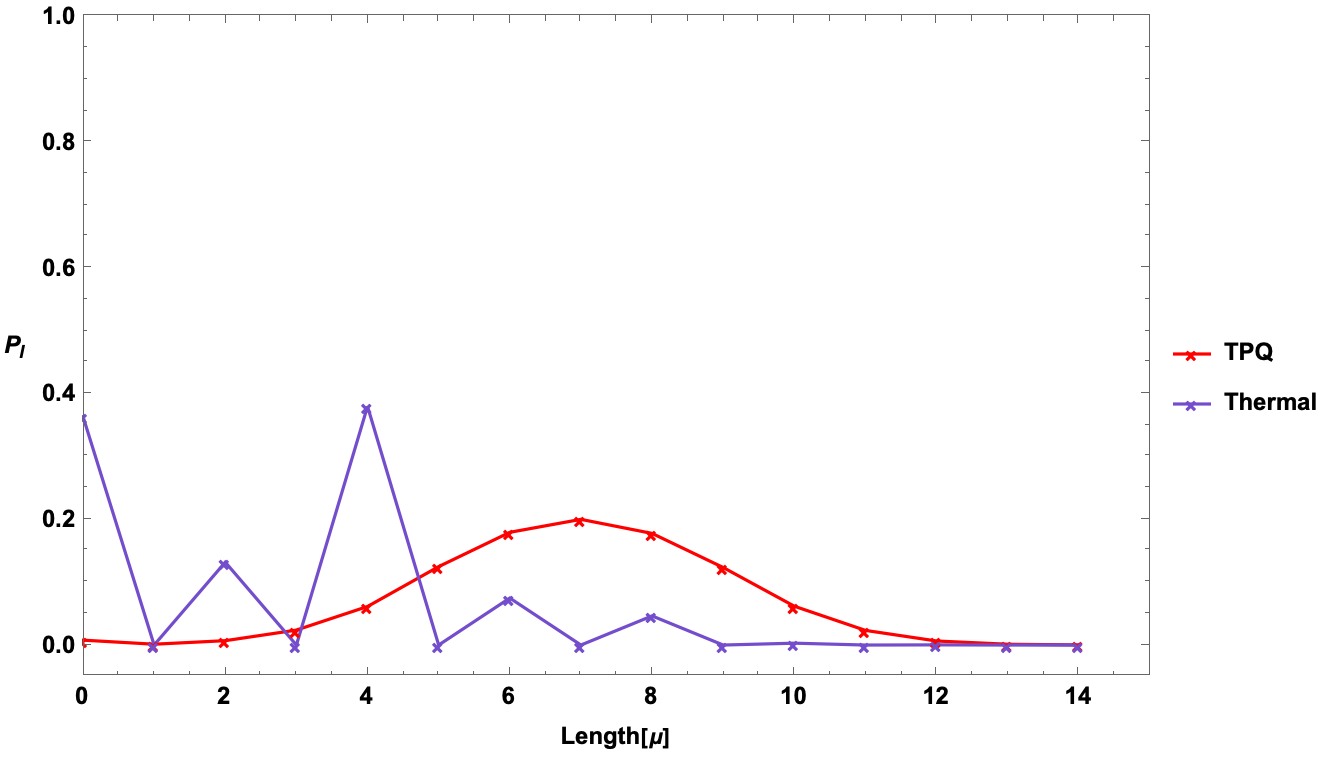}
    \caption{$g=0.2$}\label{}
  \end{subfigure}
  \begin{subfigure}{.3\linewidth}
\includegraphics[height=3.5cm,width=\linewidth]{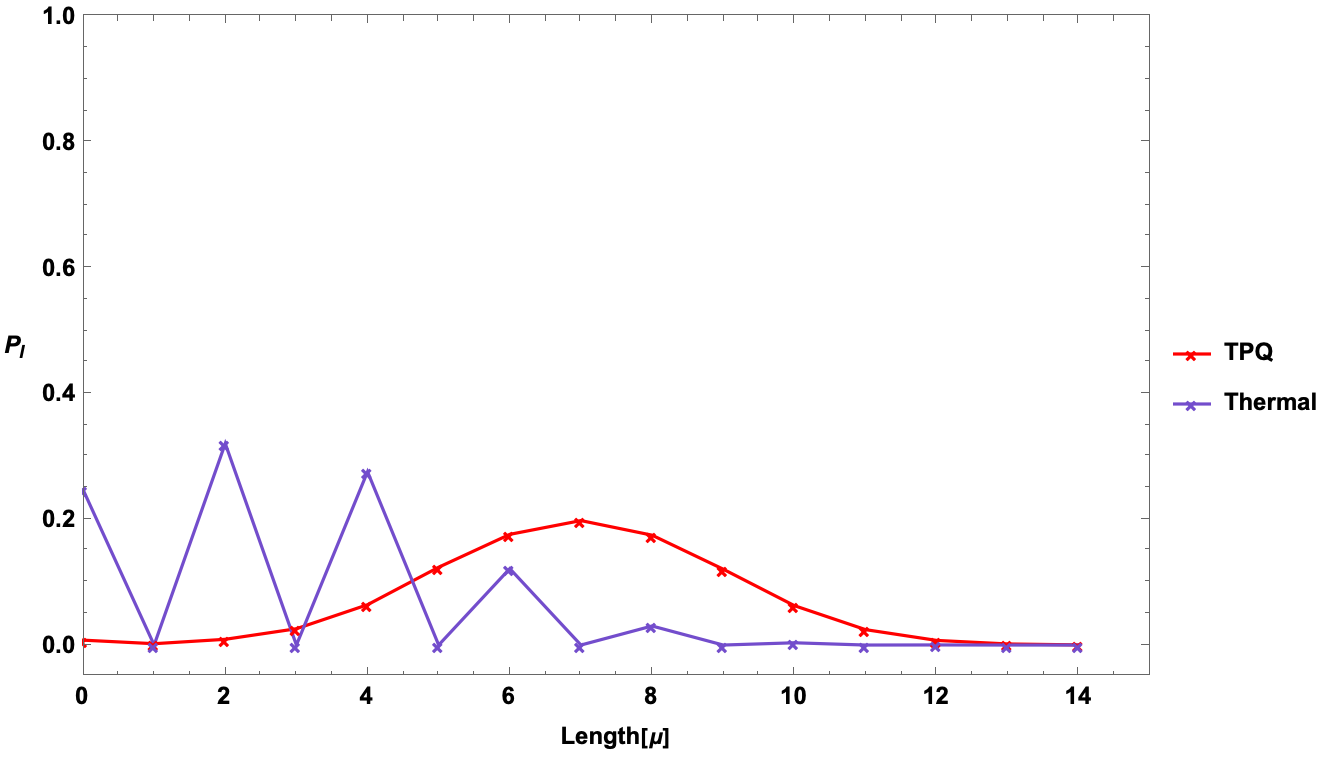}
    \caption{$g=0.4$}\label{}
  \end{subfigure}
  \begin{subfigure}{.3\linewidth}
\includegraphics[height=3.5cm,width=\linewidth]{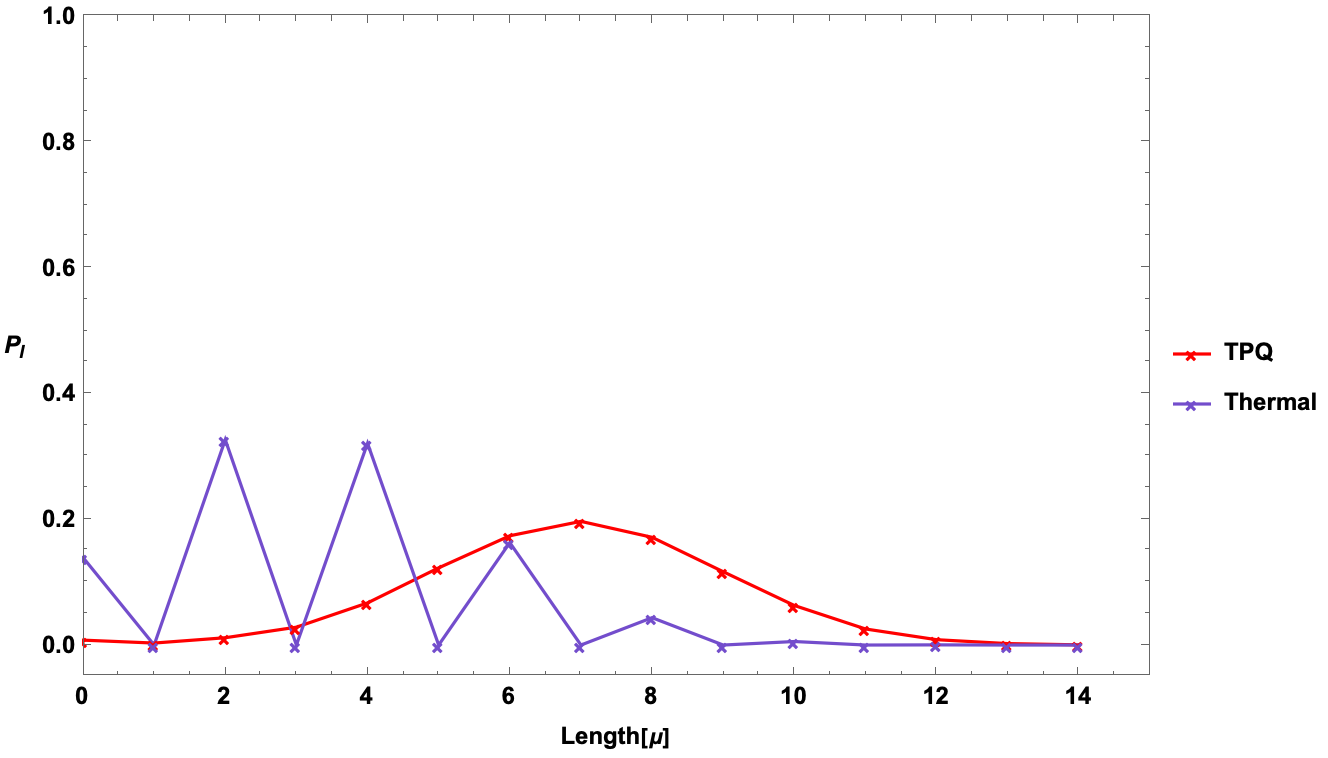}
    \caption{$g=0.6$}\label{}
  \end{subfigure}
  \begin{subfigure}{.3\linewidth}
\includegraphics[height=3.5cm,width=\linewidth]{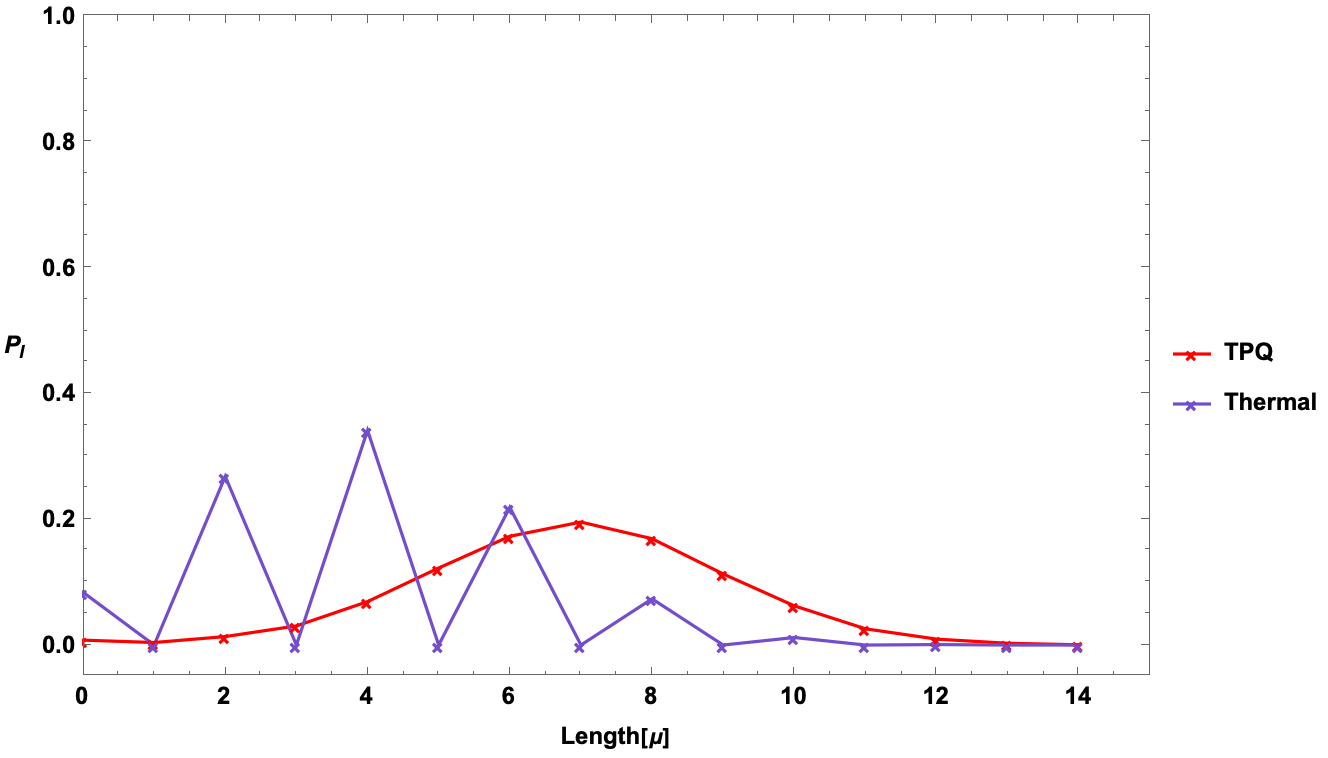}
    \caption{$g=0.8$}\label{}
  \end{subfigure}
   \begin{subfigure}{.3\linewidth}
\includegraphics[height=3.5cm,width=\linewidth]{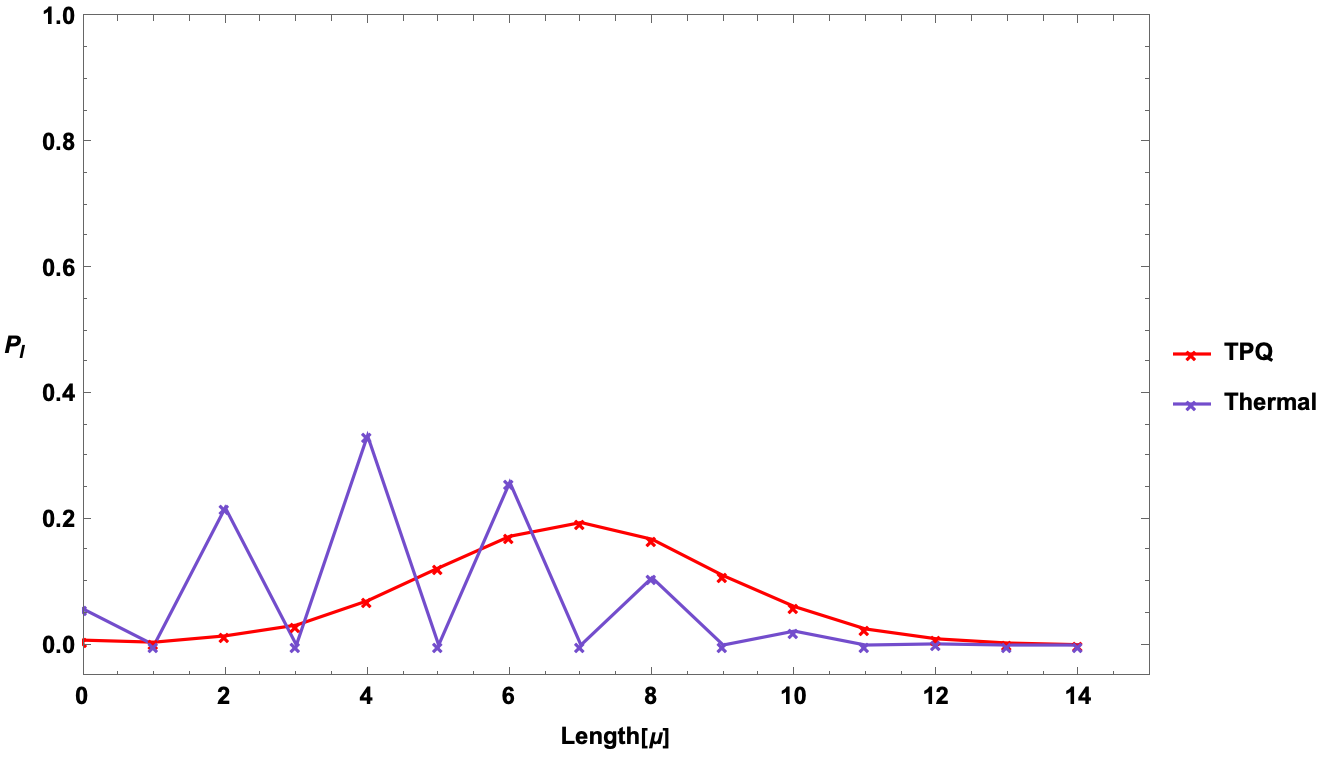}
    \caption{$g=1$}\label{}
  \end{subfigure}
  \caption{\footnotesize{For $N = 14$, $\beta=1$, $P_l = \sum_i \frac{\mathrm{Tr}(\rho \mu_i)^2}{2^{n}\mathrm{Tr}(\rho^2)}$ (summation is over all Majorana strings which have equal length) quantifies the total probability distribution over Majorana strings of varying length for the \textbf{Mass deformed SYK$_4$} model. 
  }}
  \label{MDSYKProb1}
\end{figure}

\begin{figure}[H]
  \centering
  \begin{subfigure}{.3\linewidth}
    \includegraphics[height=3.5cm,width=\linewidth]{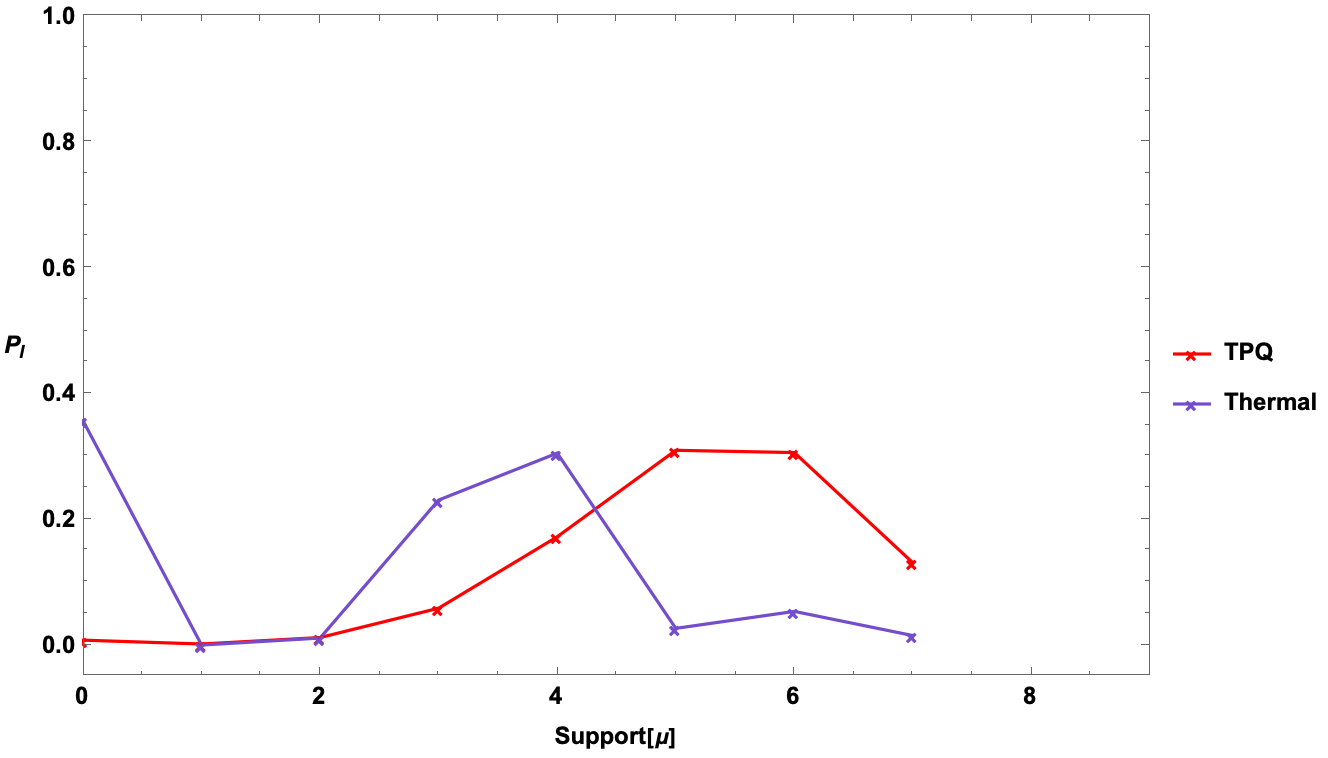}
    \caption{$g=0$}\label{}
  \end{subfigure}
  \begin{subfigure}{.3\linewidth}
    \includegraphics[height=3.5cm,width=\linewidth]{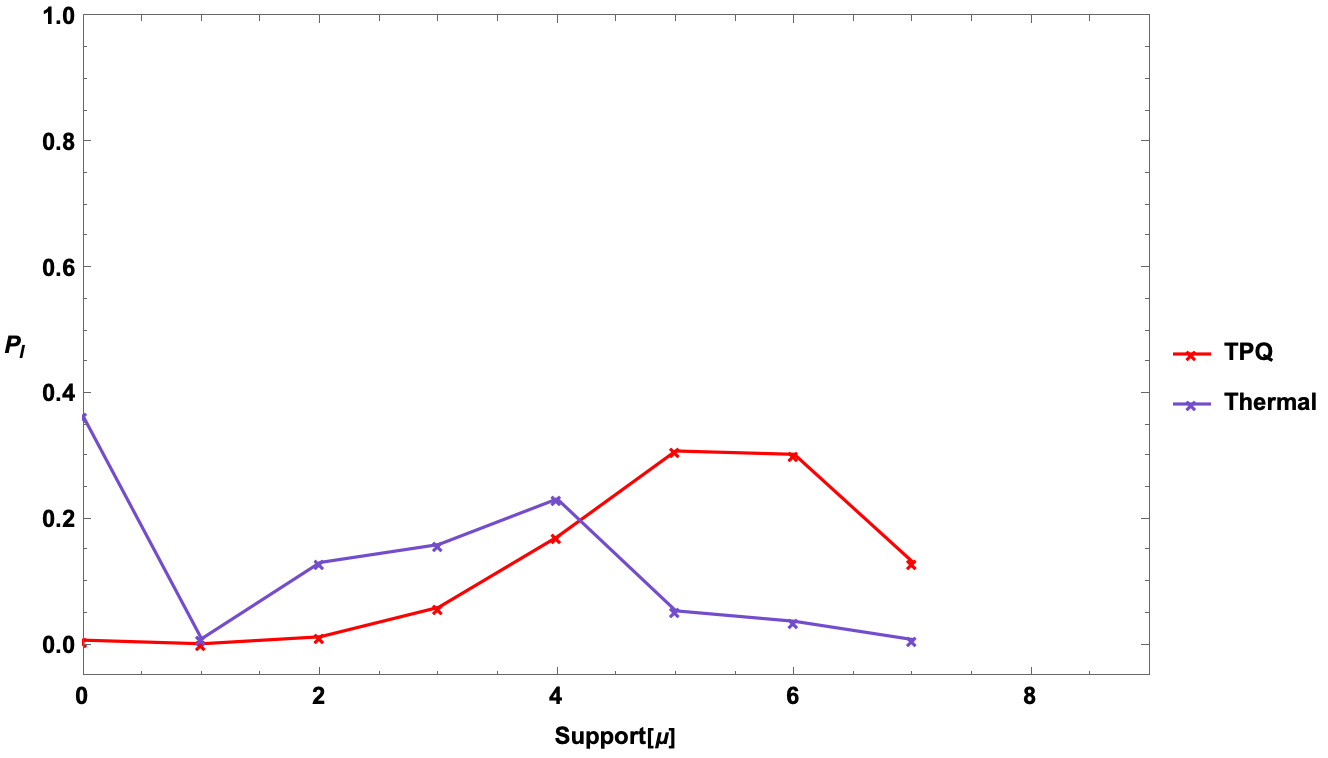}
    \caption{$g=0.2$}\label{}
  \end{subfigure}
  \begin{subfigure}{.3\linewidth}
\includegraphics[height=3.5cm,width=\linewidth]{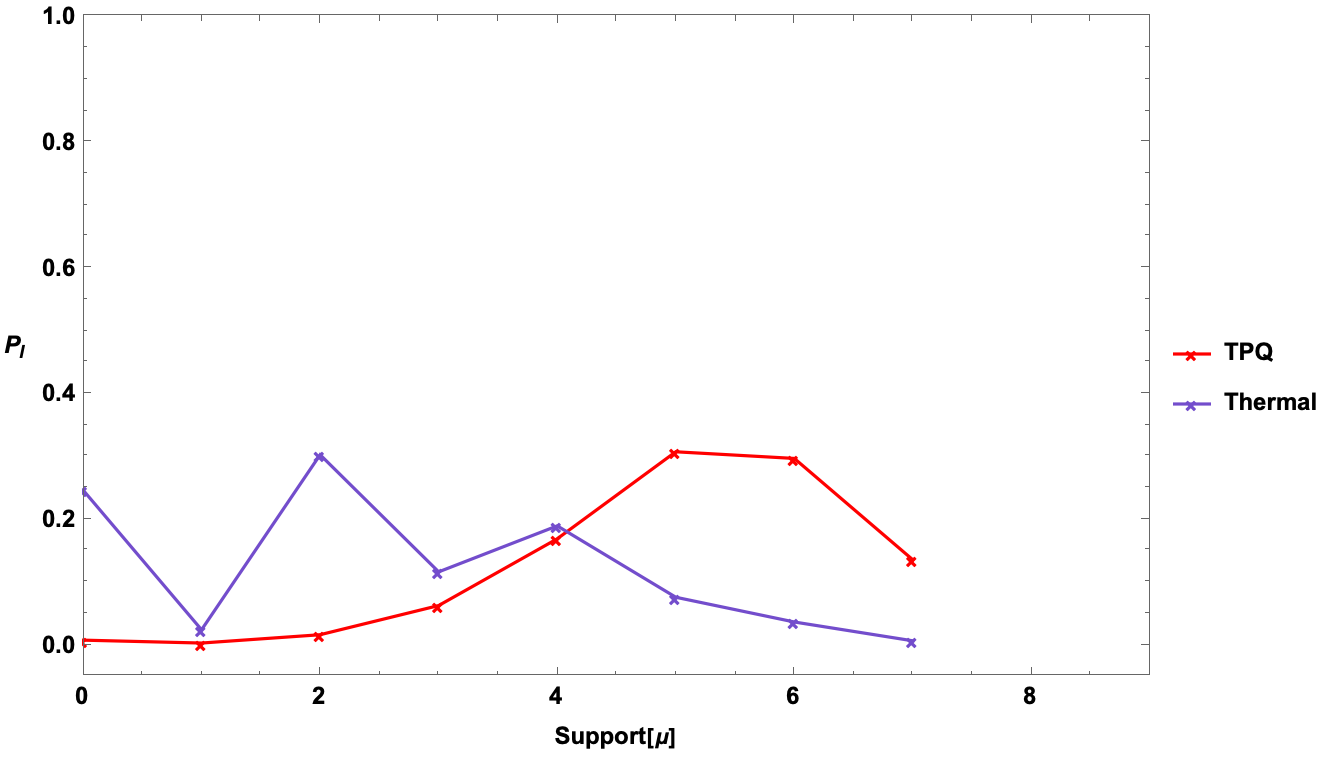}
    \caption{$g=0.4$}\label{}
  \end{subfigure}
  \begin{subfigure}{.3\linewidth}
\includegraphics[height=3.5cm,width=\linewidth]{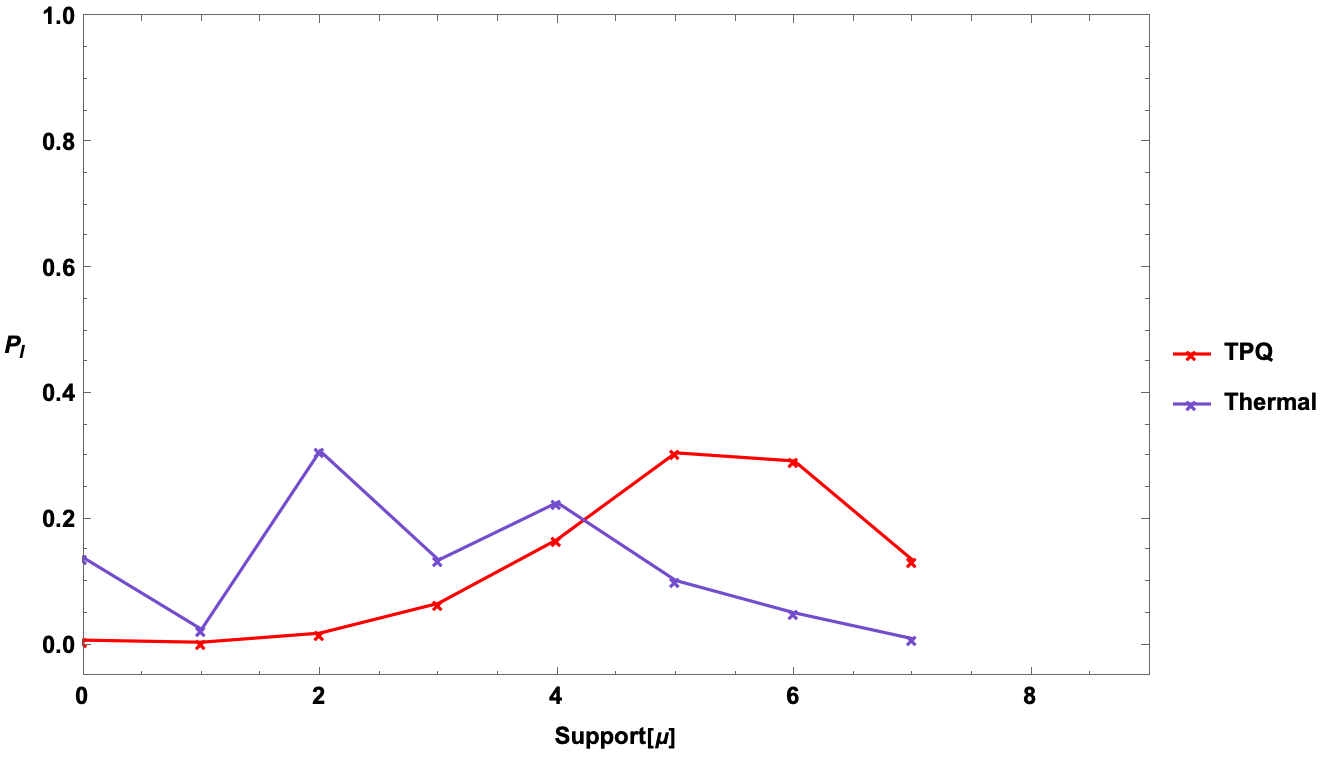}
    \caption{$g=0.6$}\label{}
  \end{subfigure}
  \begin{subfigure}{.3\linewidth}
\includegraphics[height=3.5cm,width=\linewidth]{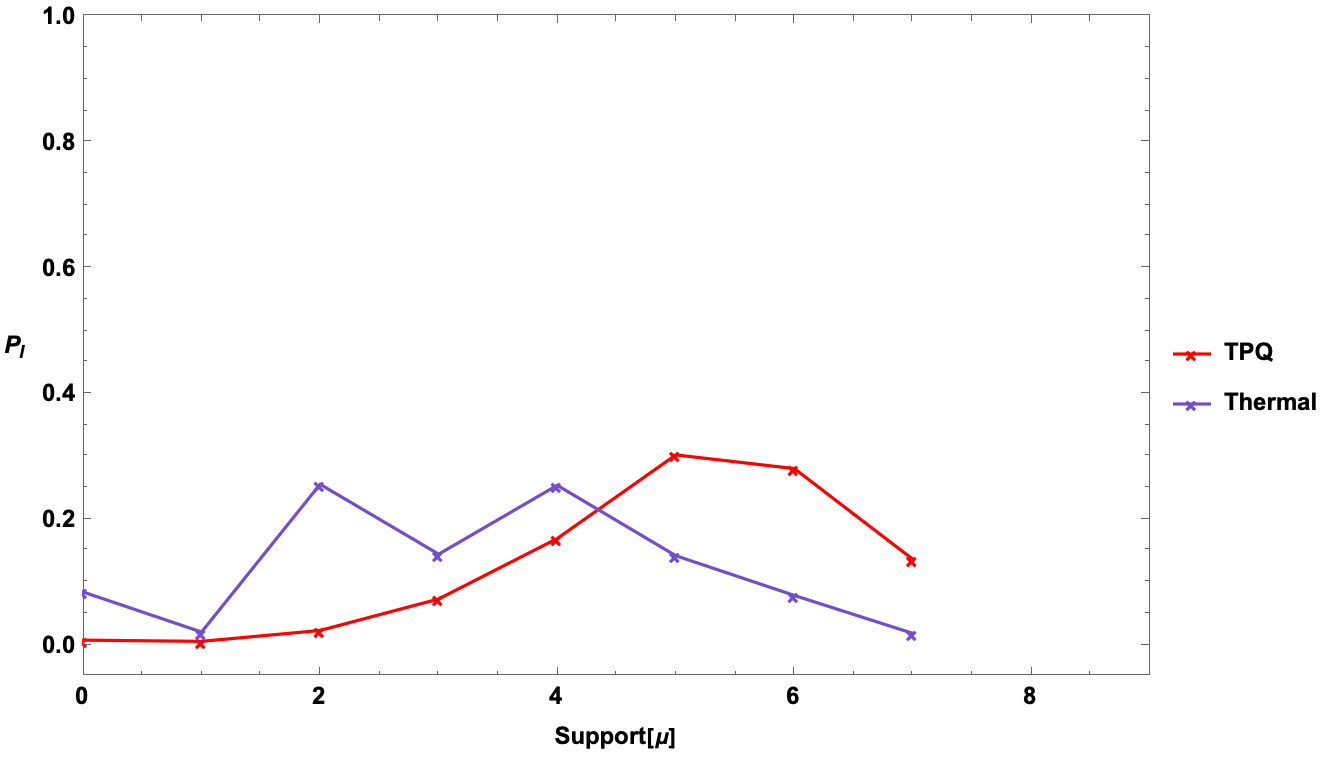}
    \caption{$g=0.8$}\label{}
  \end{subfigure}
   \begin{subfigure}{.3\linewidth}
\includegraphics[height=3.5cm,width=\linewidth]{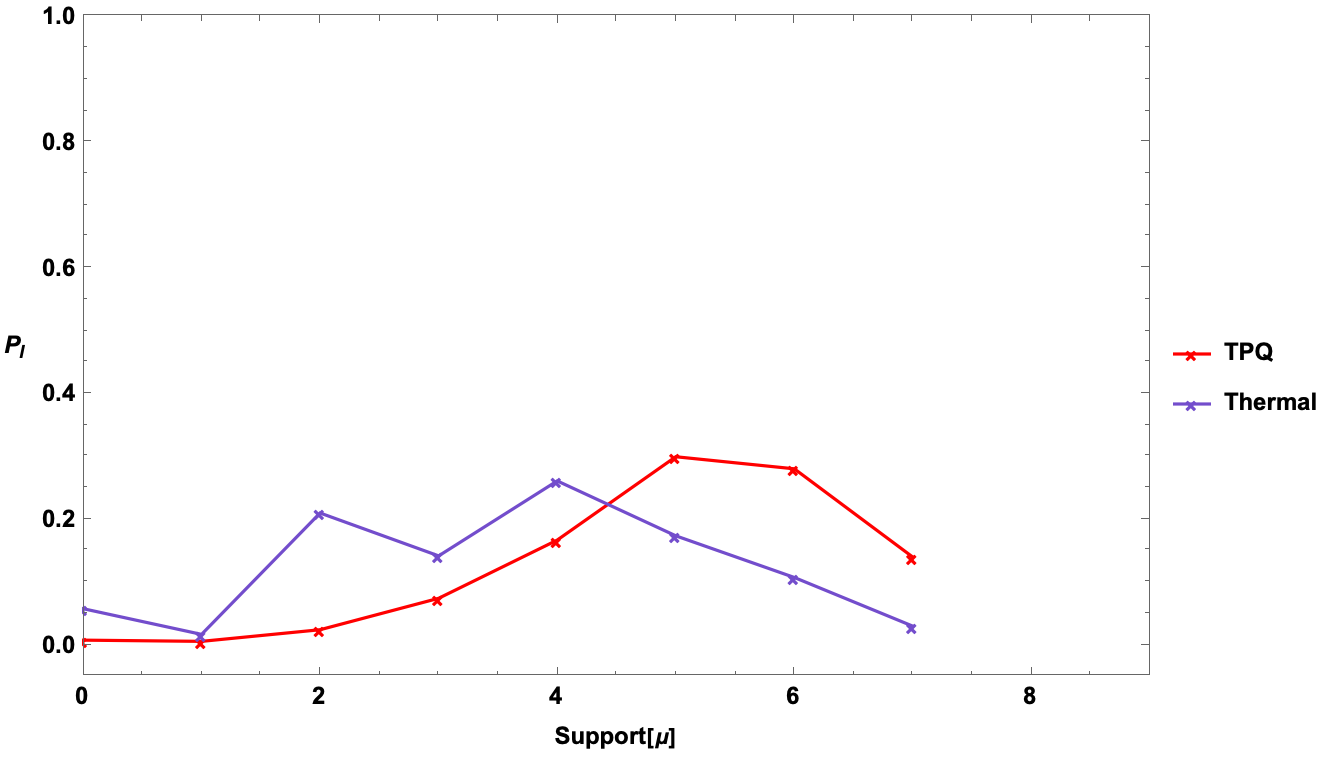}
    \caption{$g=1$}\label{}
  \end{subfigure}
  \caption{\footnotesize{For $N = 14$, $\beta=1$, $P_l = \sum_i \frac{\mathrm{Tr}(\rho \mu_i)^2}{2^{n}\mathrm{Tr}(\rho^2)}$ (summation is over all Majorana strings which have support on same number of qubits) quantifies the total probability distribution over Majorana strings of varying support for the \textbf{Mass deformed SYK$_4$} model. 
  }}
  \label{MDSYKprob2}
\end{figure}

\subsection{Multipartite non-local SRE}

Having obtained the SRE, we now proceed to study the time evolution of the multipartite non-local SRE defined in \cref{MNLDefin} for the mass-deformed SYK model . We observe that, for all values of $g \neq 1$, the  late time saturation value coincides with that of the SYK$_4$ model, while the time required to reach saturation exhibits a clear dependence on $g$ as depicted in \Cref{NLSREMDlongtime}. Interestingly, for a range of intermediate values of $g$, the multipartite non-local SRE shows an initial dip before rising and eventually saturating. Such behavior appears to be a non-trivial feature unique to the multipartite non-local SRE.

A notable feature in this mass-deformed setting is the behavior of the multipartite non-local SRE, whose sign and magnitude diagnose how magic is shared rather than merely how much magic exists.  When the non-local contribution becomes appreciable, it indicates that the system’s non-stabilizerness cannot be accounted for by a sum of few-body contributions: magic is stored in correlations that genuinely require many parties simultaneously. Moreover, negative values have a natural interpretation: the magic present in large subsystems is not additive under overlaps, and the dynamics generates a kind of frustration between different partitions. In the present interpolation, this provides a sharp way to see the crossover from chaotic to Gaussian behavior. In the strongly chaotic regime, global scrambling is strong,  allowing for efficient multipartite magic generation. As $g$ increases towards the integrable SYK$_2$ limit, the non-local component should weaken as the state becomes more describable in terms of local structure, ensuring a better additivity of the magic of the subsystems, therefore explaining the shift of multipartite non-local SRE towards $0$ as $g$ increases to $1$. Thus, the multipartite non-local SRE can be viewed as a quantitative scrambling proxy that is sensitive to the structure of magic, not just its amount.

\begin{figure}[H]
  \centering
  \begin{subfigure}{.38\linewidth}
    \includegraphics[width=\linewidth]{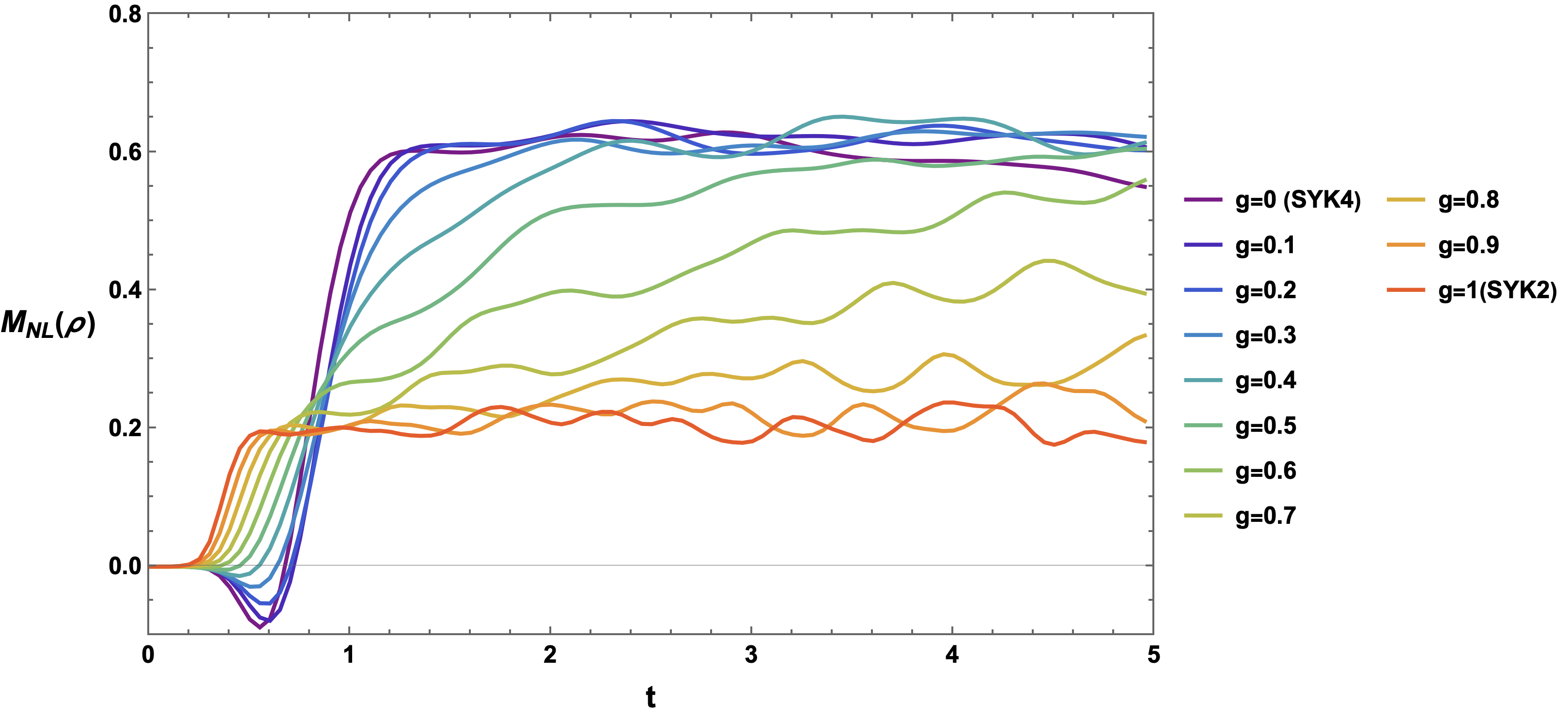}
    \caption{$N=12$ (6 qubits)}\label{SYKMD6q1GHZ00}
  \end{subfigure}
  \begin{subfigure}{.38\linewidth}
    \includegraphics[width=\linewidth]{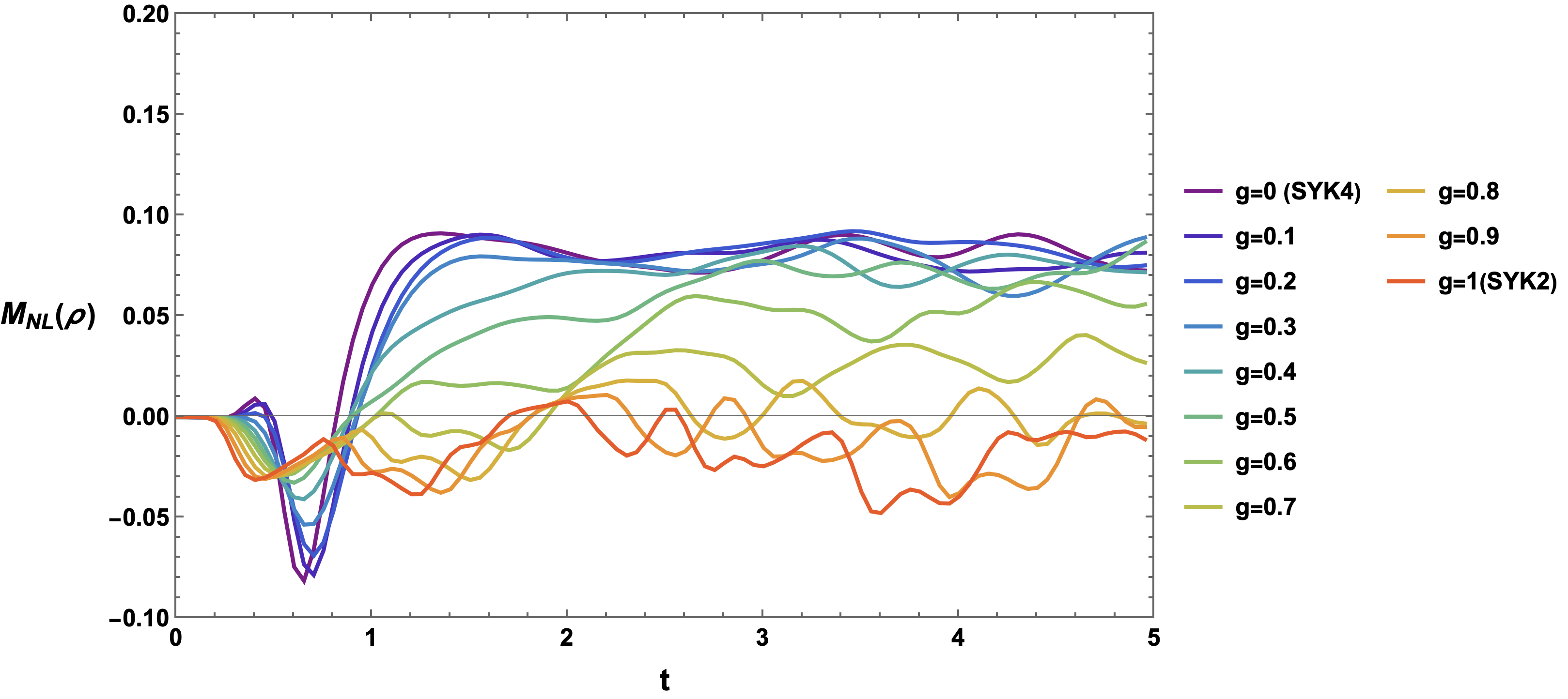}
    \caption{$N=14$ (7 qubits)}\label{SYKMD7q1GHZ00}
  \end{subfigure}

  \caption{\footnotesize{For (a) $N=12$ (b) $N=14$ and (c) $N=16$ early time behaviour of multipartite non-local Stabilizer Renyi entropy for mass deformed SYK with product state $\ket{000\cdots}$ as the initial state. }}
\end{figure}

\begin{figure}[H]
  \centering
   \begin{subfigure}{.38\linewidth}
    \includegraphics[width=\linewidth]{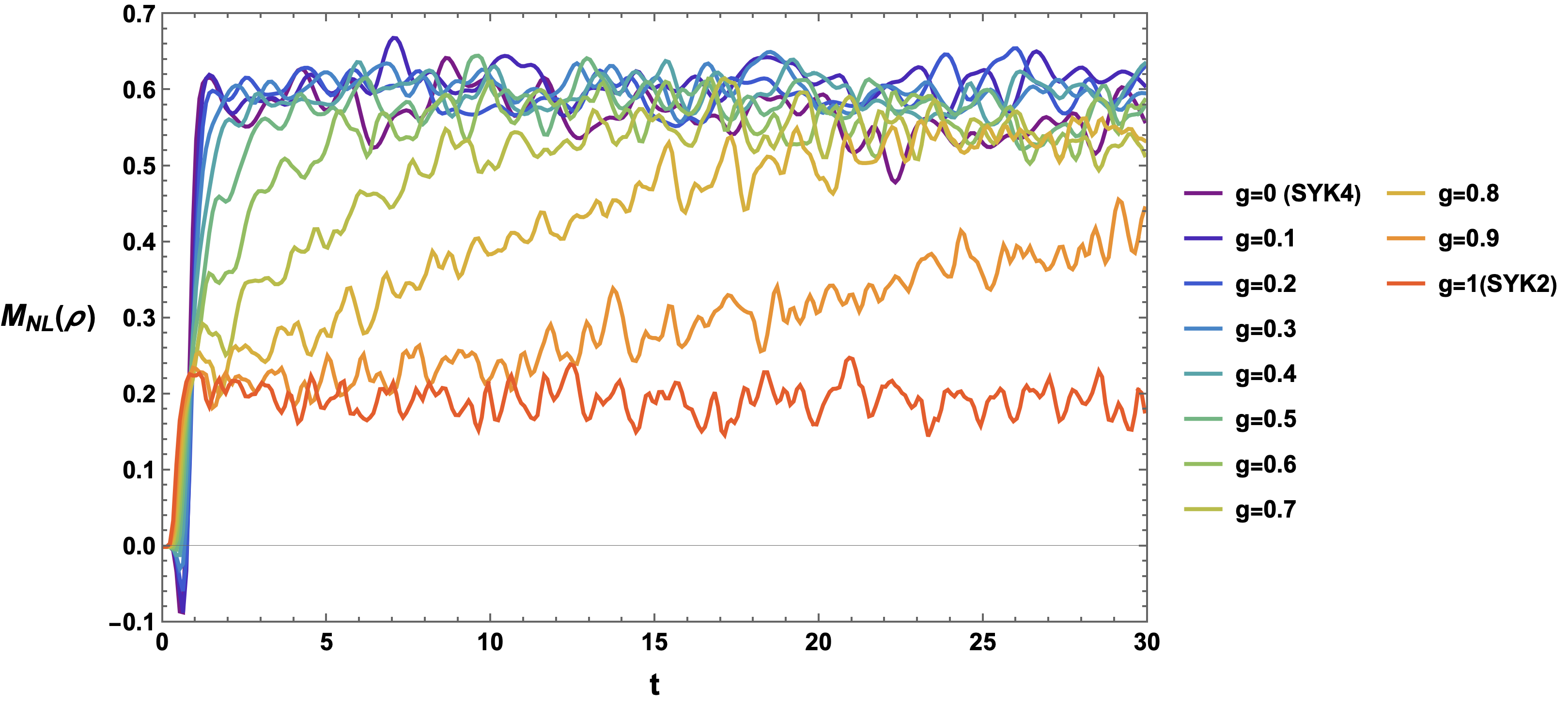}
    \caption{$N=12$ (6 qubits) }\label{SYK6q1Z}
  \end{subfigure}
  \begin{subfigure}{.38\linewidth}
\includegraphics[width=\linewidth]{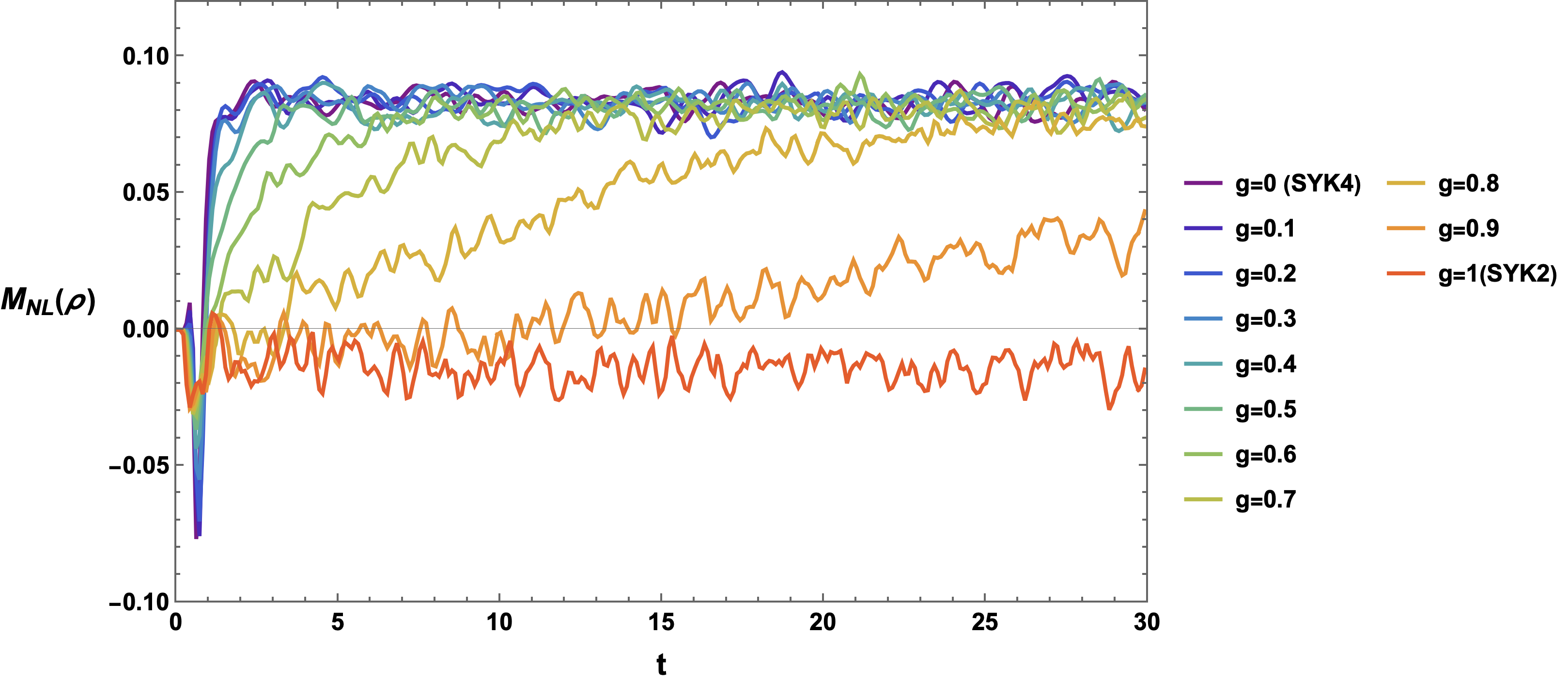}
    \caption{$N=14$ (7 qubits)}\label{SYK7q1Z}
  \end{subfigure}
  \caption{\footnotesize{For (a) $N=12$ (b)$N=14$  Long time behaviour of multipartite non-local Stabilizer Renyi entropy for mass deformed SYK with product state $\ket{000\cdots}$ as the initial state averaged over 100 samples.}}\label{NLSREMDlongtime}
\end{figure}
\subsection{SRE distribution of energy eigen states}
Having examined the effect of mass deformation on the time evolution of both the SRE and the multipartite SRE, we now turn to the distribution of SRE values across the energy eigenstates. We observe that the behavior is largely overlapping for different system sizes $N$ in the two limiting cases: the pure SYK$_4$ model and the purely quadratic SYK$_2$ (mass-term-only) model. For small but nonzero $g$, however, the distributions corresponding to different eigenstates become clearly separated, with each set taking values within a relatively flat band that bends only near the spectral edges. The degree of bending becomes more pronounced as $g$ increases.

\begin{figure}[H]
  \centering
  \begin{subfigure}{.3\linewidth}
    \includegraphics[height=3.5cm,width=\linewidth]{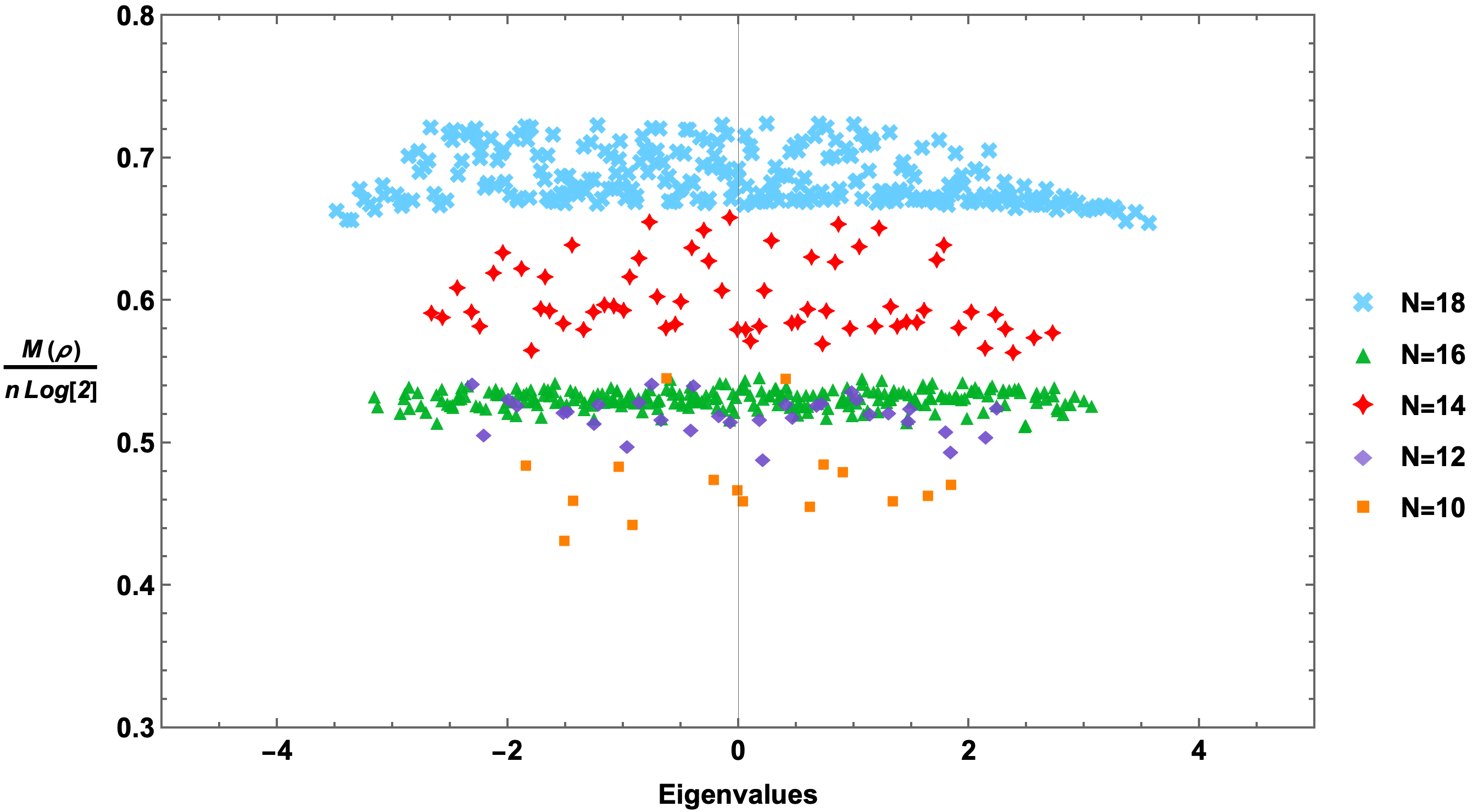}
    \caption{$g=0$}\label{}
  \end{subfigure}
  \begin{subfigure}{.3\linewidth}
    \includegraphics[height=3.5cm,width=\linewidth]{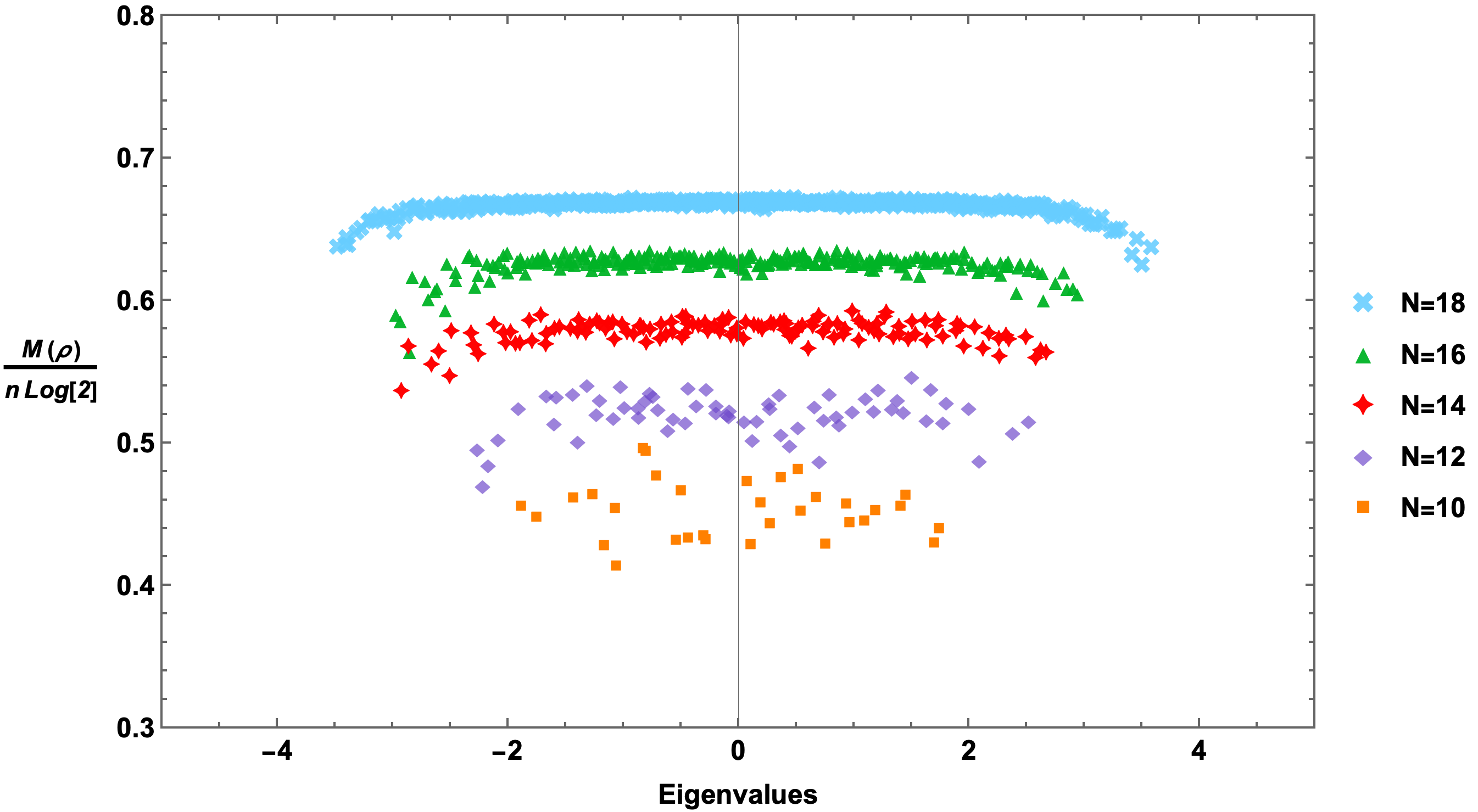}
    \caption{$g=0.25$}\label{}
  \end{subfigure}
  \begin{subfigure}{.3\linewidth}
\includegraphics[height=3.5cm,width=\linewidth]{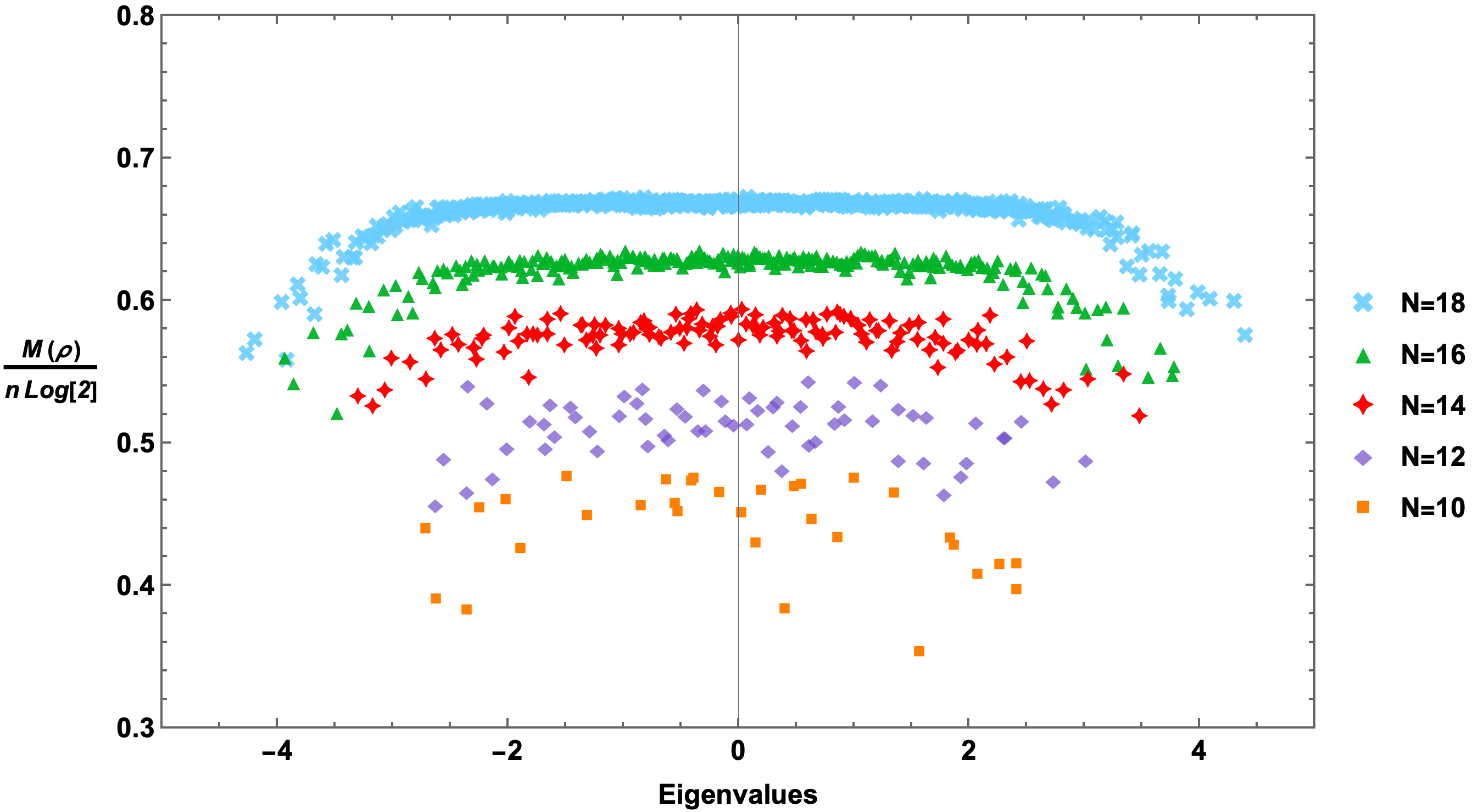}
    \caption{$g=0.5$}\label{}
  \end{subfigure}
  \begin{subfigure}{.3\linewidth}
\includegraphics[height=3.5cm,width=\linewidth]{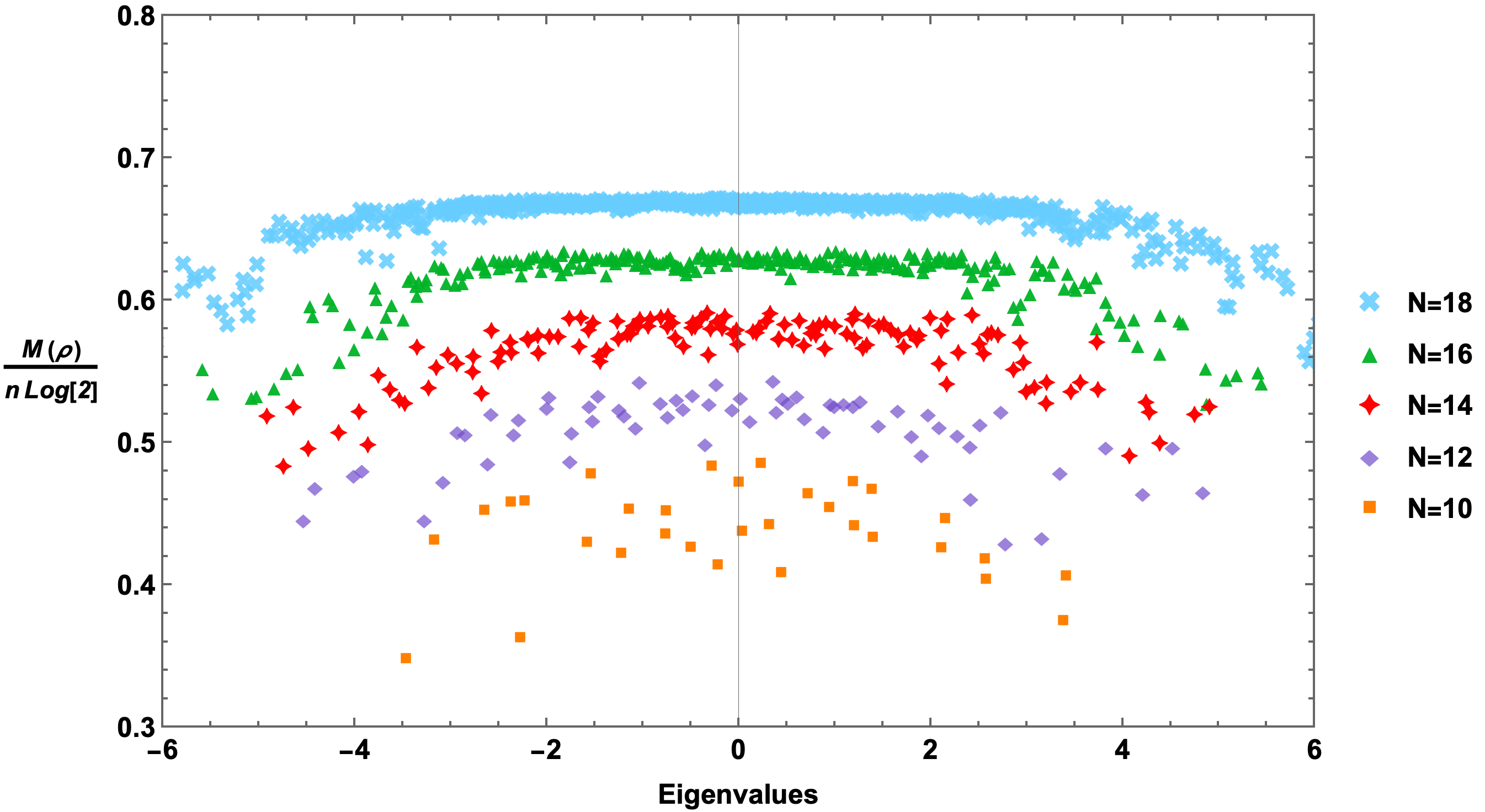}
    \caption{$g=0.75$}\label{}
  \end{subfigure}
   \begin{subfigure}{.3\linewidth}
\includegraphics[height=3.5cm,width=\linewidth]{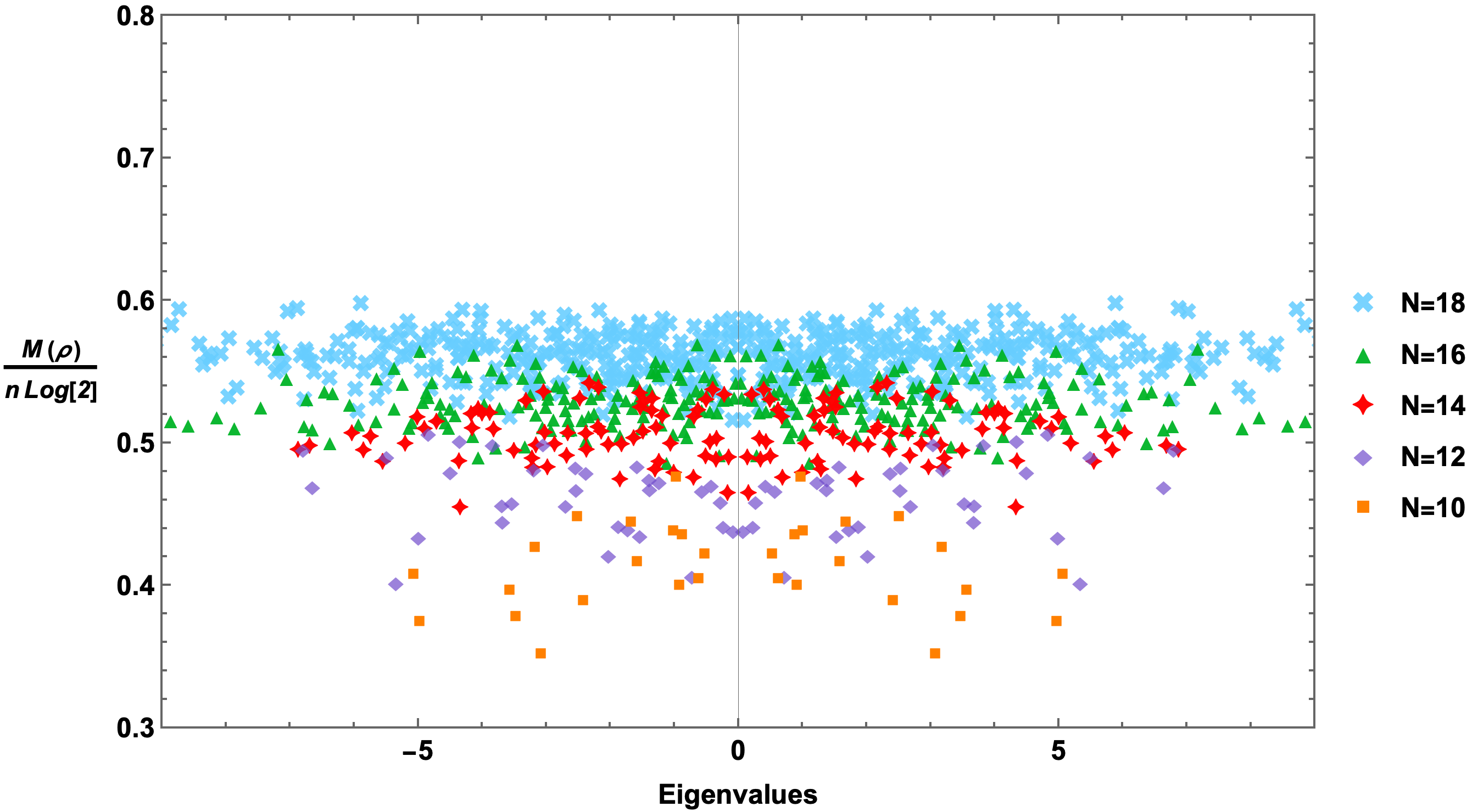}
    \caption{$g=1$}\label{}
  \end{subfigure}
  \caption{\footnotesize{For $N =10$-$18$, SRE distribution of energy eigenstates in the \textbf{Mass deformed SYK$_4$} model. 
  }}
  \label{SREeigen}
\end{figure}

For small and intermediate deformation $g$, the SRE of individual energy eigenstates exhibits a neat bulk plateau. Indeed, the eigenstates cluster around a value that is mostly independent of the energy, and whose fluctuations decrease with $N$. This is the expected ETH-like behavior for a chaotic many-body system for which eigenstates in the bulk of the spectrum behave as ``typical'' vectors, so basis/frame-dependent measures such as the SRE become self-averaging and nearly universal. The decrease near the edges of the spectrum can be roughly understood as follows. Low energy and near extremal states are more structured (more constrained by approximate quasi-particle descriptions), and therefore explore a smaller portion of Hilbert space in terms of our stabilizer basis, therefore leading to reduced non-stabilizerness.
Let us now turn to the effect of the mass deformation. As $g$ increases toward the SYK$_2$ limit, the SRE plateau lowers and the separation between different $N$ becomes less important, culminating at $g=1$ where the scatter is larger and the overall level is largely reduced. This is consistent with the interpretation that the mass term introduces an increasingly integrable, quadratic contribution, and making the eigenstates become closer to a fermionic Gaussian structure, which is intuitively less efficient at generating non-stabilizerness than pure strongly interacting SYK$_4$ eigenstates. The mass deformation interpolates between a regime where eigenstates are highly spread in the stabilizer basis (high SRE) and a regime where integrability  dominates (lower SRE).

\subsection{Multipartite non-local SRE distribution of energy eigen states}
We now investigate the distribution of the multipartite non-local SRE across the energy eigenstates. For $g = 1$, we find a strong overlap of the distributions for different system sizes $N$, similar to the behavior observed earlier. However, unlike the previous case, a noticeable amount of overlap persists even for other values of $g$. In addition, we observe that the absolute value of the multipartite non-local SRE is consistently higher near the center of the spectral distribution for any intermediate value of $g$, i.e., for $g$ away from both $g = 0$ and $g = 1$.

\begin{figure}[H]
  \centering
  \begin{subfigure}{.3\linewidth}
    \includegraphics[height=3.5cm,width=\linewidth]{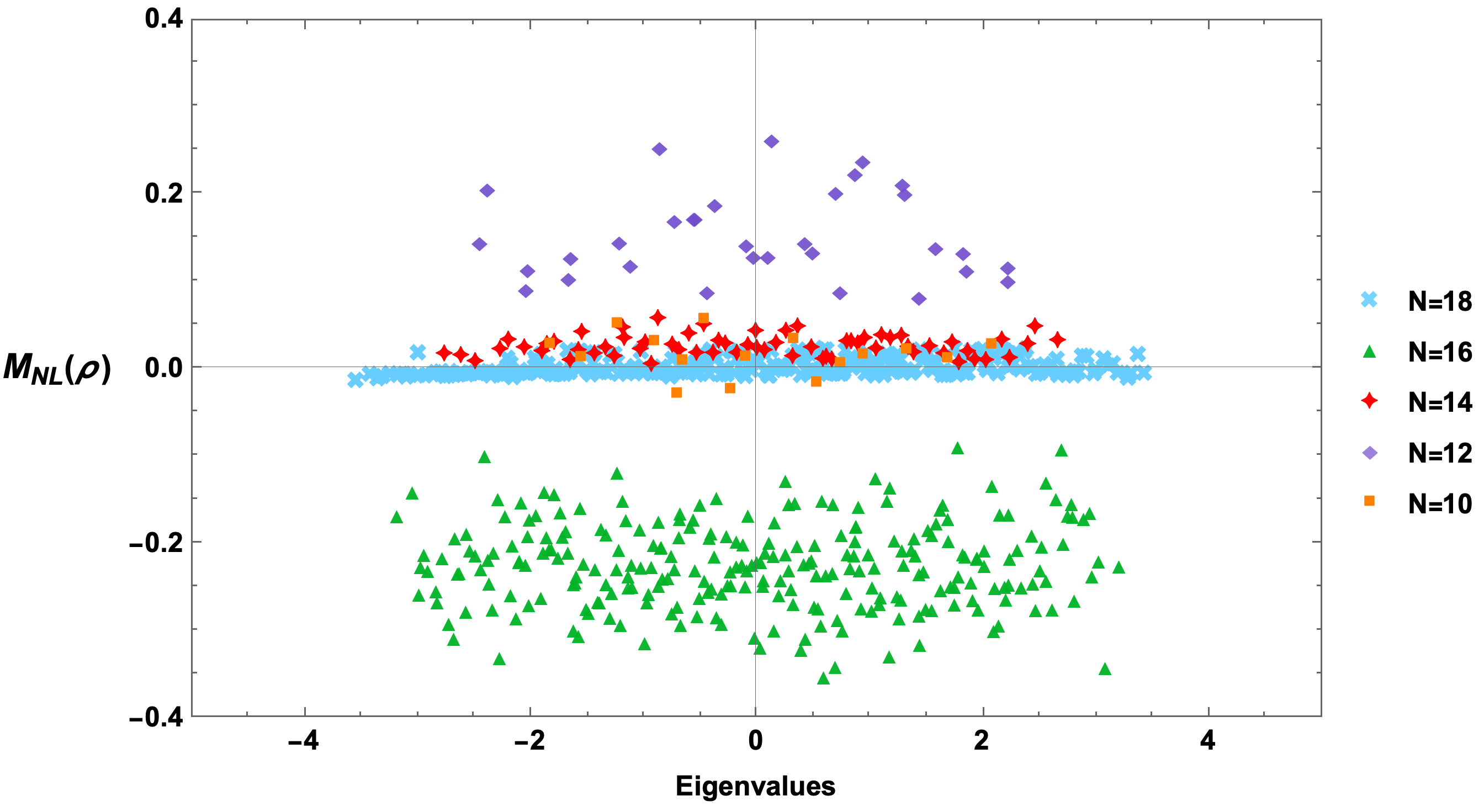}
    \caption{$g=0$}\label{}
  \end{subfigure}
  \begin{subfigure}{.3\linewidth}
    \includegraphics[height=3.5cm,width=\linewidth]{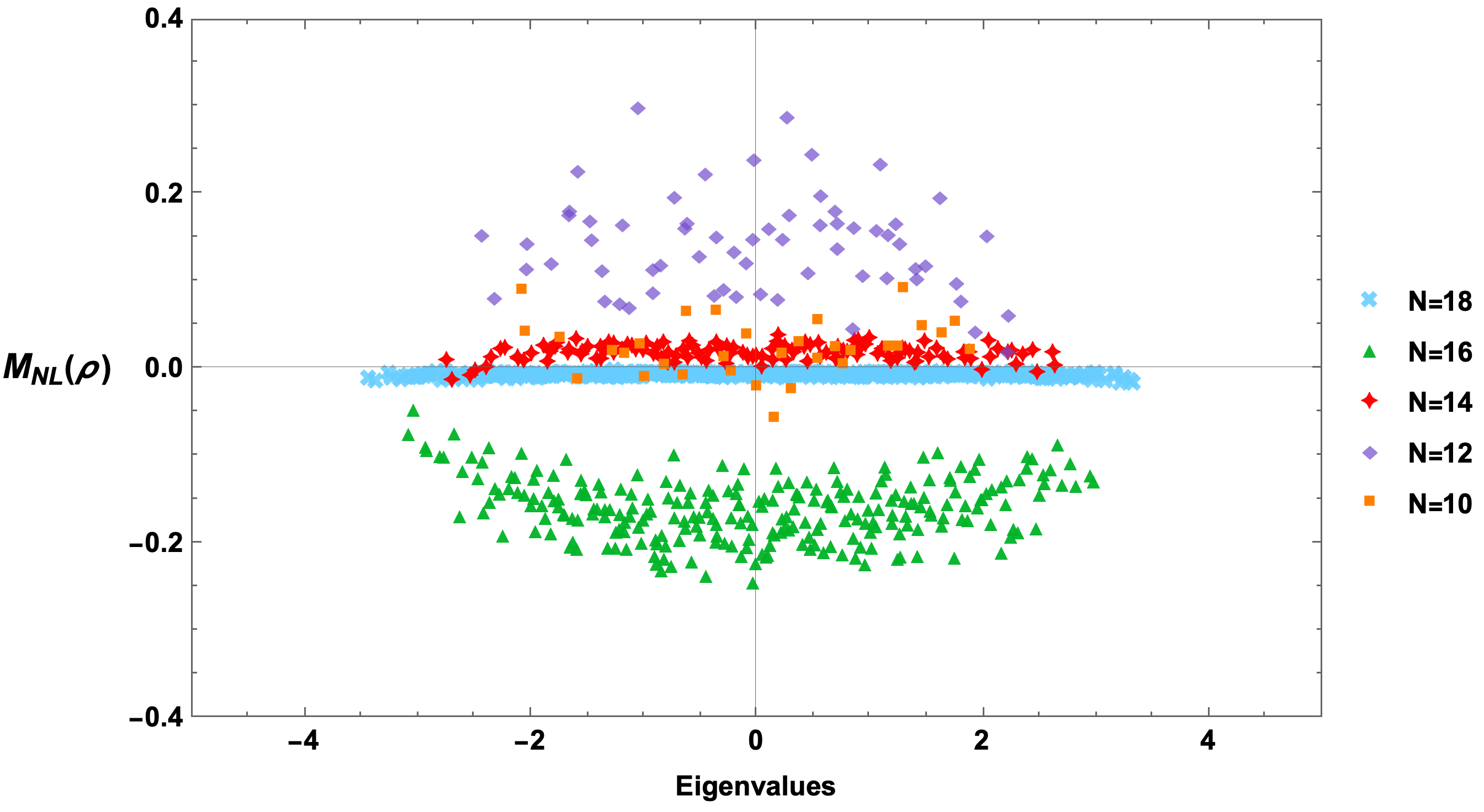}
    \caption{$g=0.25$}\label{}
  \end{subfigure}
  \begin{subfigure}{.3\linewidth}
\includegraphics[height=3.5cm,width=\linewidth]{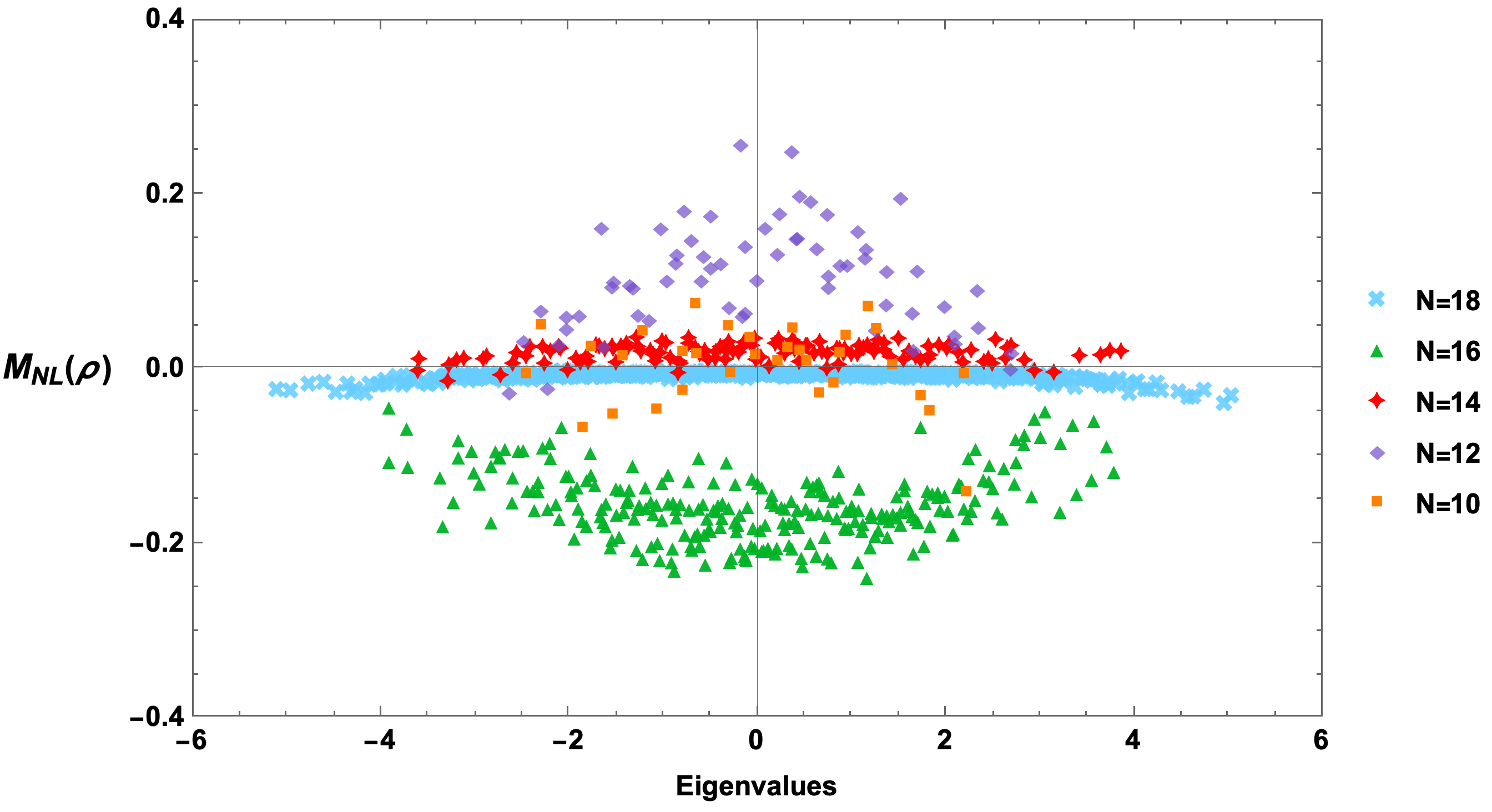}
    \caption{$g=0.5$}\label{}
  \end{subfigure}
  \begin{subfigure}{.3\linewidth}
\includegraphics[height=3.5cm,width=\linewidth]{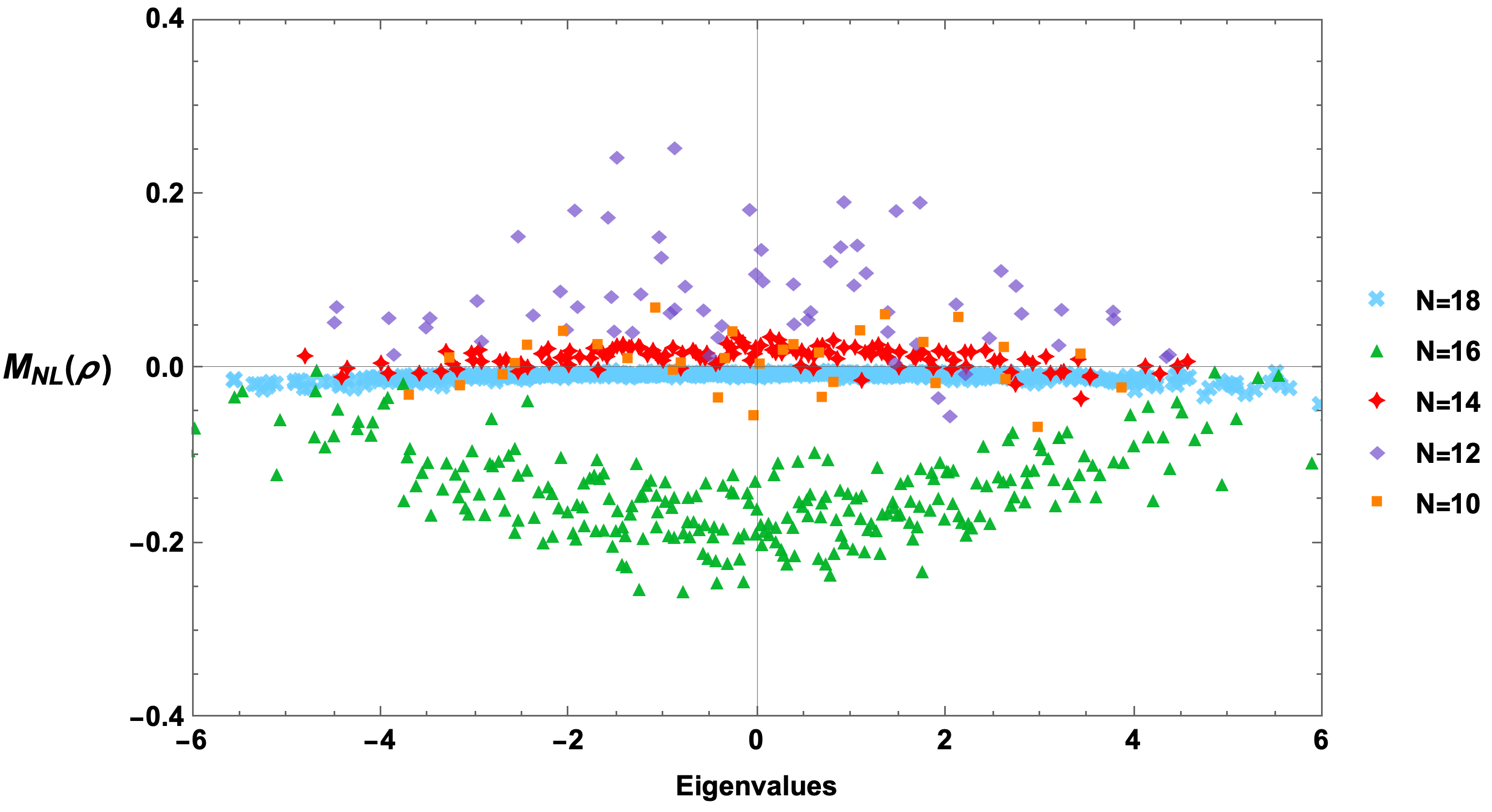}
    \caption{$g=0.75$}\label{}
  \end{subfigure}
   \begin{subfigure}{.3\linewidth}
\includegraphics[height=3.5cm,width=\linewidth]{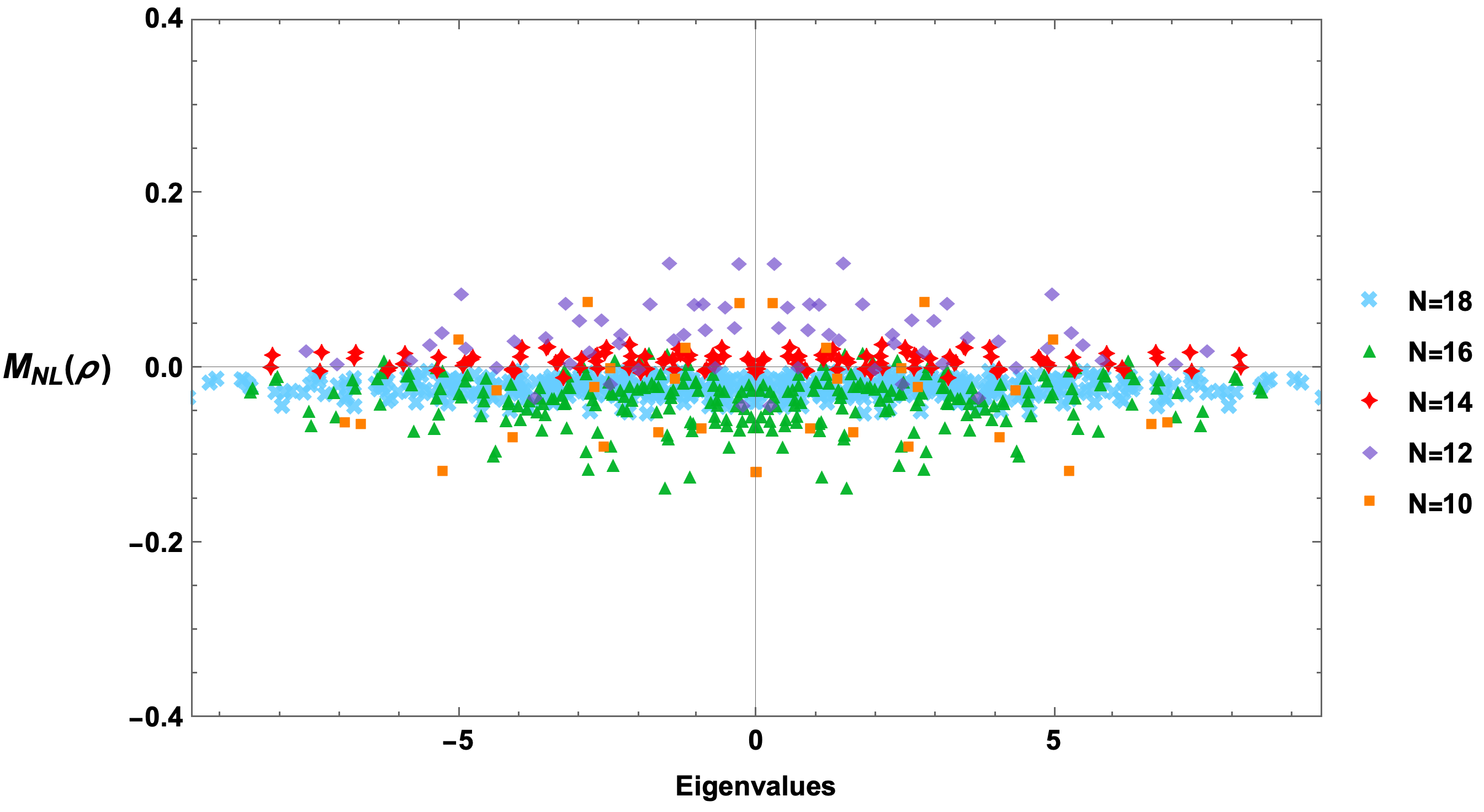}
    \caption{$g=1$}\label{}
  \end{subfigure}
  \caption{\footnotesize{For $N =10$-$18$, multipartite non-local SRE distribution of energy eigenstates in the \textbf{Mass deformed SYK$_4$} model. 
  }}
  \label{NLSREeigen}
\end{figure}

Fig.~\ref{NLSREeigen} does not show a simple monotone shift of $M_{\mathrm{NL}}(\rho)$ toward $0$ as $g\to 1$ for all $N$.  The clearest trend is instead a flattening of the spread and a weakening of the energy dependence as one approaches the quadratic limit. For the largest size shown ($N=18$), the points already lie very close to $M_{\mathrm{NL}}=0$ for every $g$, suggesting that this  quantity becomes strongly self-averaging at large $N$.  At smaller $N$ there is a much more visible residue, reflecting finite-size sensitivity in how magic is distributed across subsystems. As $g\to 1$ these finite-size effects are reduced and the data for different $N$ converge near $0$, consistent again with the quadratic (more structured) limit producing a more nearly additive sharing pattern across subsystems.

\section{Sparse SYK}\label{sec 5}
Another interesting variant of the SYK model where the degree of chaos can be tuned is the \emph{sparse} SYK$_4$ model, in which the random couplings $J_{ijkl}$ are taken to be nonzero only with probability $p$ \cite{Xu:2020shn,Orman:2024mpw}. In this construction the model reduces to the fully connected, maximally chaotic SYK$_4$ at $p=1$, while decreasing $p$ progressively suppresses interactions and weakens chaotic behavior. Another way to introduce sparseness is by randomly choosing a fixed number of non-zero couplings among the $J_{ijkl}$ coupling which are chosen from Gaussian distribution and setting the rest of the coupling to zero.  We examine SRE and its multipartite extension in such sparse $SYK$ models below.

\subsection{SRE}
Below, in \Cref{Sparse 1} and \Cref{Sparse 2} we depict the plots for  the stabilizer R\'enyi entropy  for states obtained through the time evolution of such a sparse SYK$_4$ model. The plots highlight the transition of the late-time saturation value as the parameter controlling the sparseness in the SYK$_4$ model is varied. More specifically, we observe that the greater the sparseness (corresponding to smaller values of $p$ or $n_s$), the lower the resulting saturation value.

\begin{figure}[H]
  \centering
  \begin{subfigure}{.3\linewidth}
    \includegraphics[height=3.5cm,width=\linewidth]{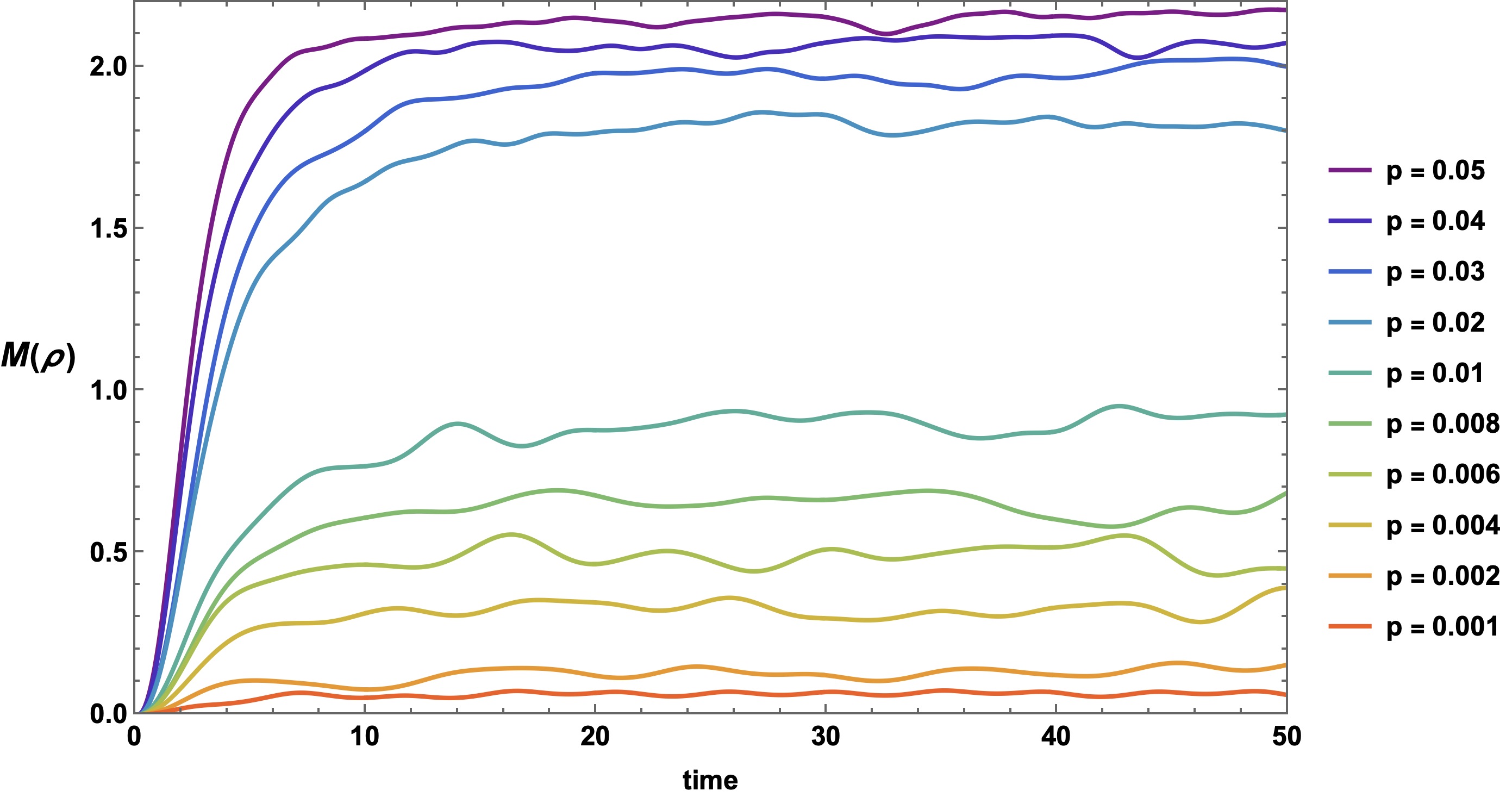}
    \caption{6 Qubits}\label{SP2SYK6q1GHZ}
  \end{subfigure}
  \begin{subfigure}{.3\linewidth}
    \includegraphics[height=3.5cm,width=\linewidth]{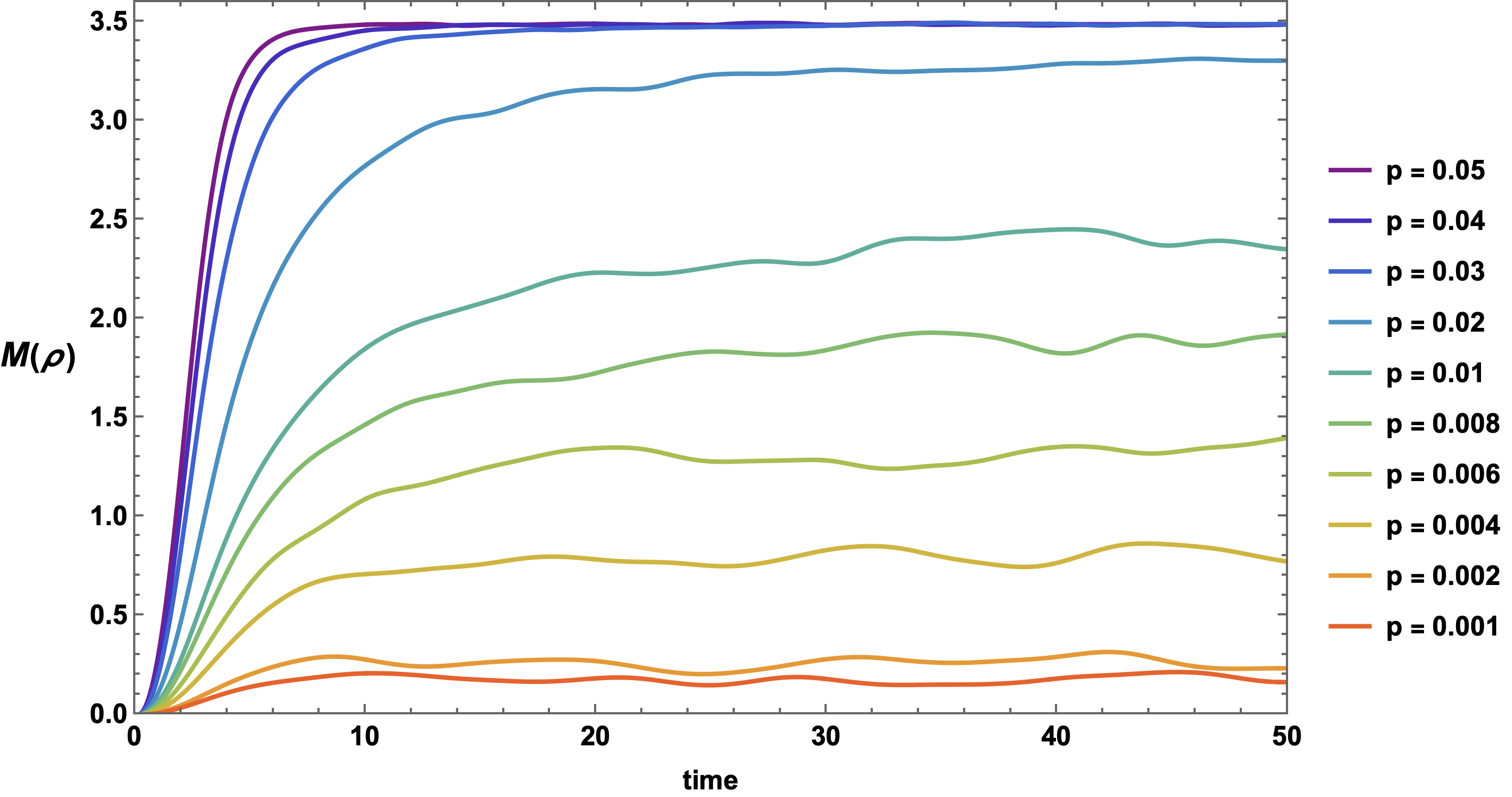}
    \caption{7 Qubits}\label{SP2SYK7q1GHZ}
  \end{subfigure}
  \begin{subfigure}{.3\linewidth}
    \includegraphics[height=3.5cm,width=\linewidth]{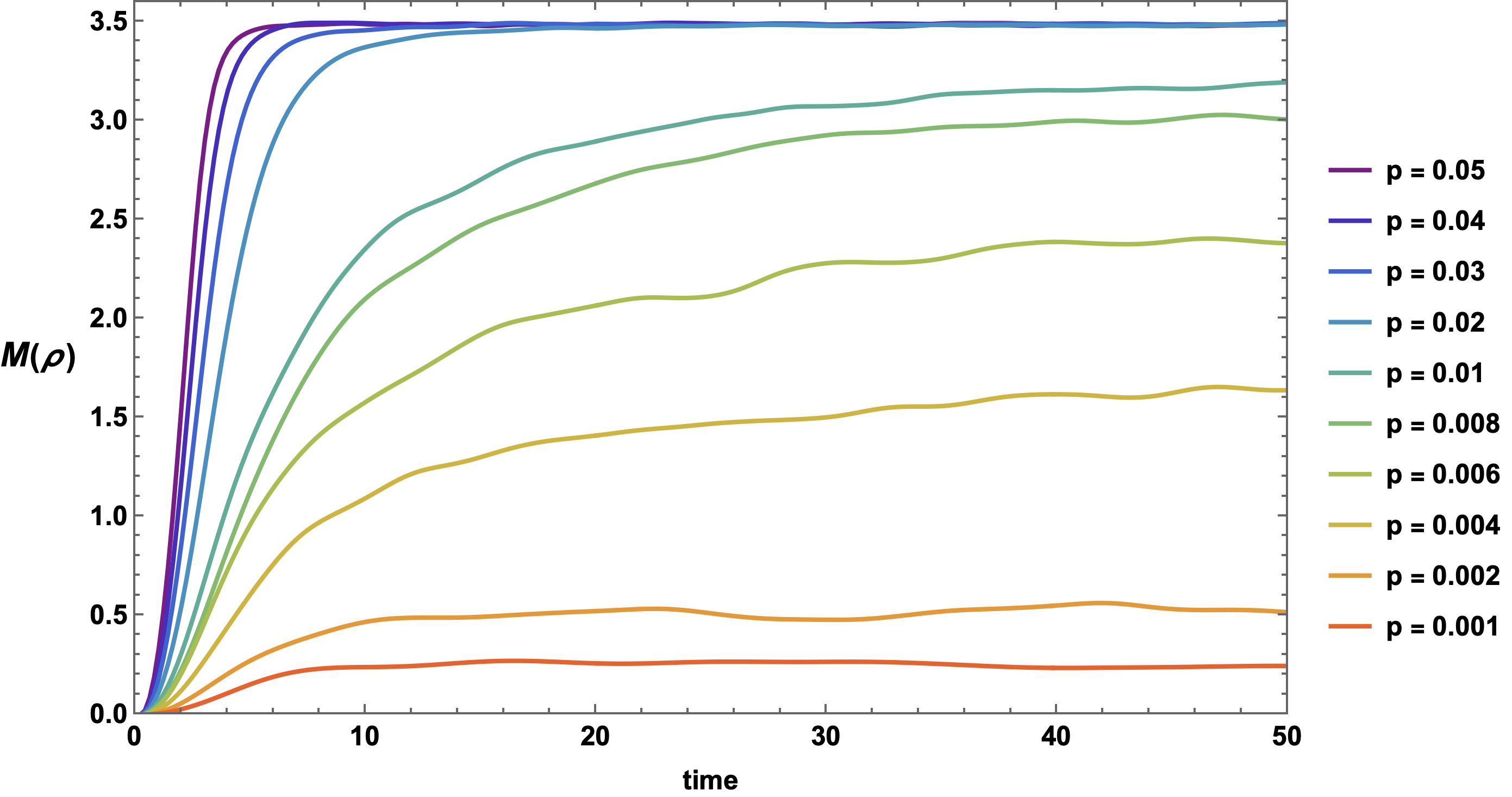}
    \caption{8 Qubits}\label{SP2SYK8q1GHZ}
  \end{subfigure}
  \caption{\footnotesize{For panels (a) $N=12$ (6 qubits), (b) $N=14$ (7 qubits) (C) $N=16$ (8 qubits) show the time evolution of the stabilizer Rényi entropy (SRE) with $\ket{GHZ_N}$ as the initial state averaged over 25 samples. }}\label{Sparse 1}
\end{figure}

\begin{figure}[H]
  \centering
  \begin{subfigure}{.3\linewidth}
    \includegraphics[height=3.5cm,width=\linewidth]{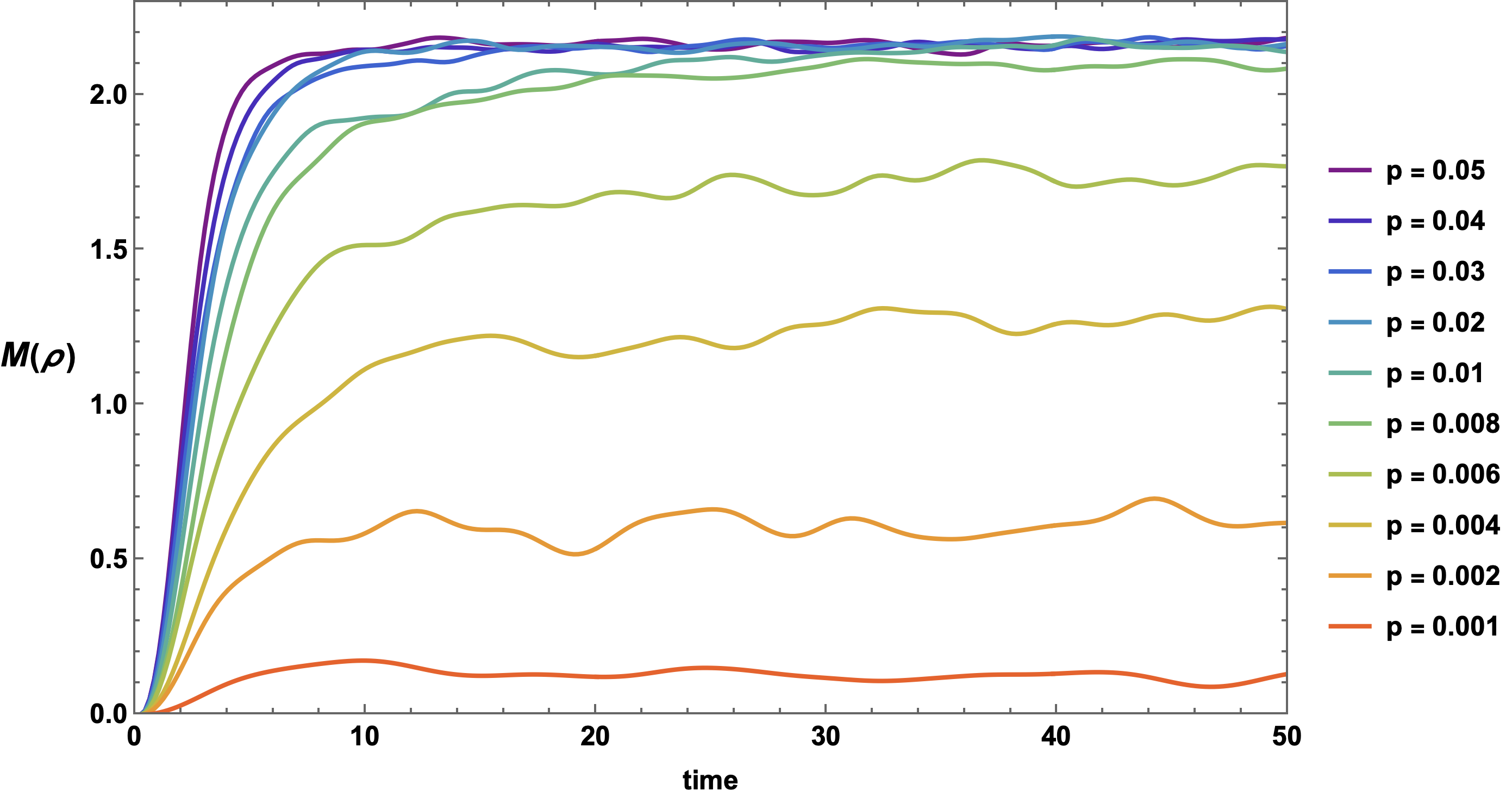}
    \caption{6 Qubits}\label{SP3SYK6q1GHZ}
  \end{subfigure}
  \begin{subfigure}{.3\linewidth}
    \includegraphics[height=3.5cm,width=\linewidth]{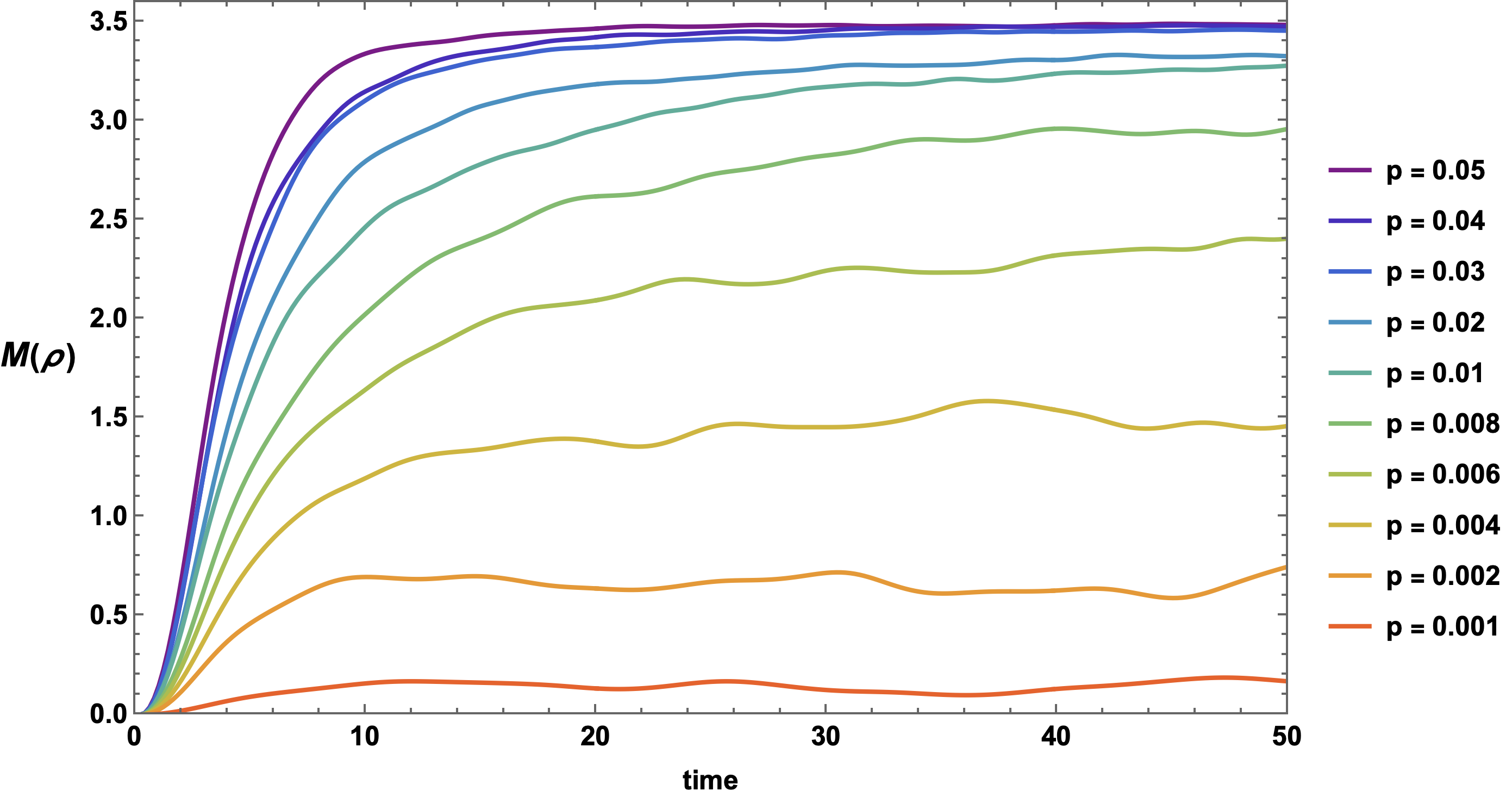}
    \caption{7 Qubits}\label{SP3SYK7q1GHZ}
  \end{subfigure}
  \begin{subfigure}{.3\linewidth}
    \includegraphics[height=3.5cm,width=\linewidth]{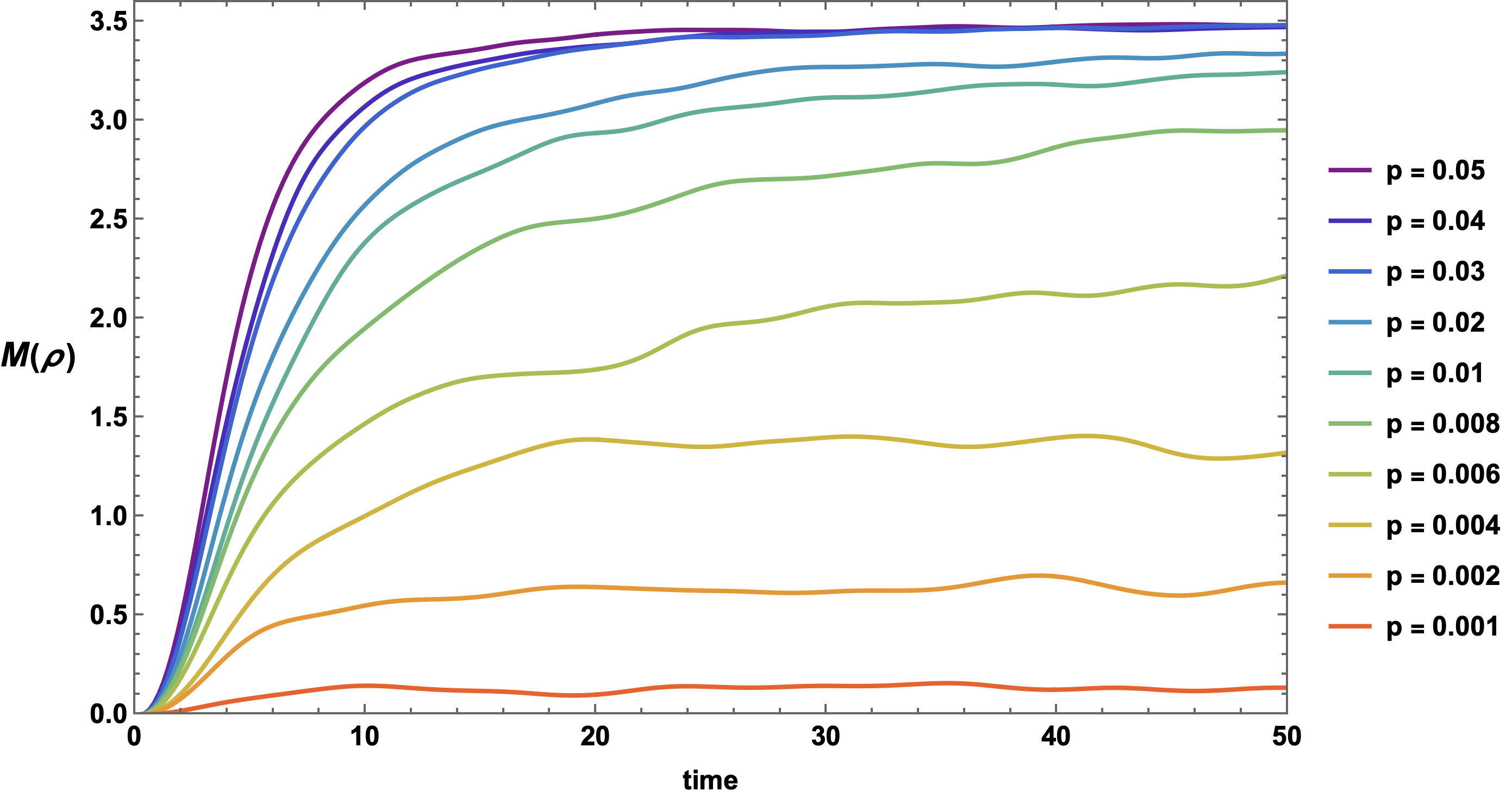}
    \caption{8 Qubits}\label{SP3SYK8q1GHZ}
  \end{subfigure}
  \caption{\footnotesize{For panels (a) $N=12$ (6 qubits), (b) $N=14$ (7 qubits) (C) $N=16$ (8 qubits) show the time evolution of the stabilizer Rényi entropy (SRE) with $\ket{GHZ_N}$ as the initial state averaged over 25 samples. Here the sparseness is controlled by  $n_s$ whose list is displayed on the right. These numbers denote the nonzero couplings that are randomly sampled among the $J_{ijkl}$ couplings obtained from the Gaussian distribution. }}\label{Sparse 2}
\end{figure}

A useful way to interpret this trend is in terms of operator growth. Full SYK corresponds to a complete all-to-all graph of $4$-body couplings, which triggers rapid mixing between many distinct Majorana strings and drives the state toward a highly typical, high-magic steady regime. Sparsification reduces the number of available scattering options. Operator growth becomes constrained to propagate along a much more restricted graph, which reduces the effective portion of Hilbert space explored at finite time scales. 
therefore, sparsity suppresses magic for the same reason it suppresses fast scrambling, namely we observe a less thorough delocalization of the wavefunction in the fermionic Pauli/Majorana basis.

Quite interestingly, the long-time behavior of the SRE differs markedly between the mass-deformed and sparse SYK models. In the mass-deformed case, we observed that for any $g\neq 1$ the SRE eventually saturates to the same value as in the pure SYK$_4$ model, with only the saturation time depending on the deformation parameter. This may be understood from the fact that the quartic SYK$_4$ interaction remains present for all $g\neq 1$, so although the quadratic term modifies the early-time dynamics, the late-time operator content is still governed by the full SYK$_4$ Hamiltonian.
In contrast, sparsity alters the Hamiltonian in a fundamentally different manner: by reducing the number of interaction terms, it limits the available operator growth channels. Consequently, the saturation value of the SRE decreases with increasing sparseness and no longer approaches the SYK$_4$ value. Thus, mass deformation changes primarily the timescale of magic generation, whereas sparsity changes the capacity of the model to generate magic.

\subsection{Multipartite non-local SRE}

We also evaluate the multipartite SRE, with the corresponding plots shown in \Cref{NLSRESparse}. As in the mass-deformed $SYK$ model, we once again observe an initial dip followed by a rise and eventual saturation. Interestingly, both the onset time of the dip and the saturation time exhibit a clear monotonic dependence on the parameter $p$. In particular, as $p$ increases (i.e., as the ensemble becomes less sparse), the dip occurs later and the system takes longer to reach its saturation value. This behavior highlights the sensitivity of multipartite non-local correlations to the degree of sparseness and shows that increasing $p$ effectively slows down the approach to the steady-state regime.

\begin{figure}[H]
  \centering
  \begin{subfigure}{.3\linewidth}
    \includegraphics[height=3.5cm,width=\linewidth]{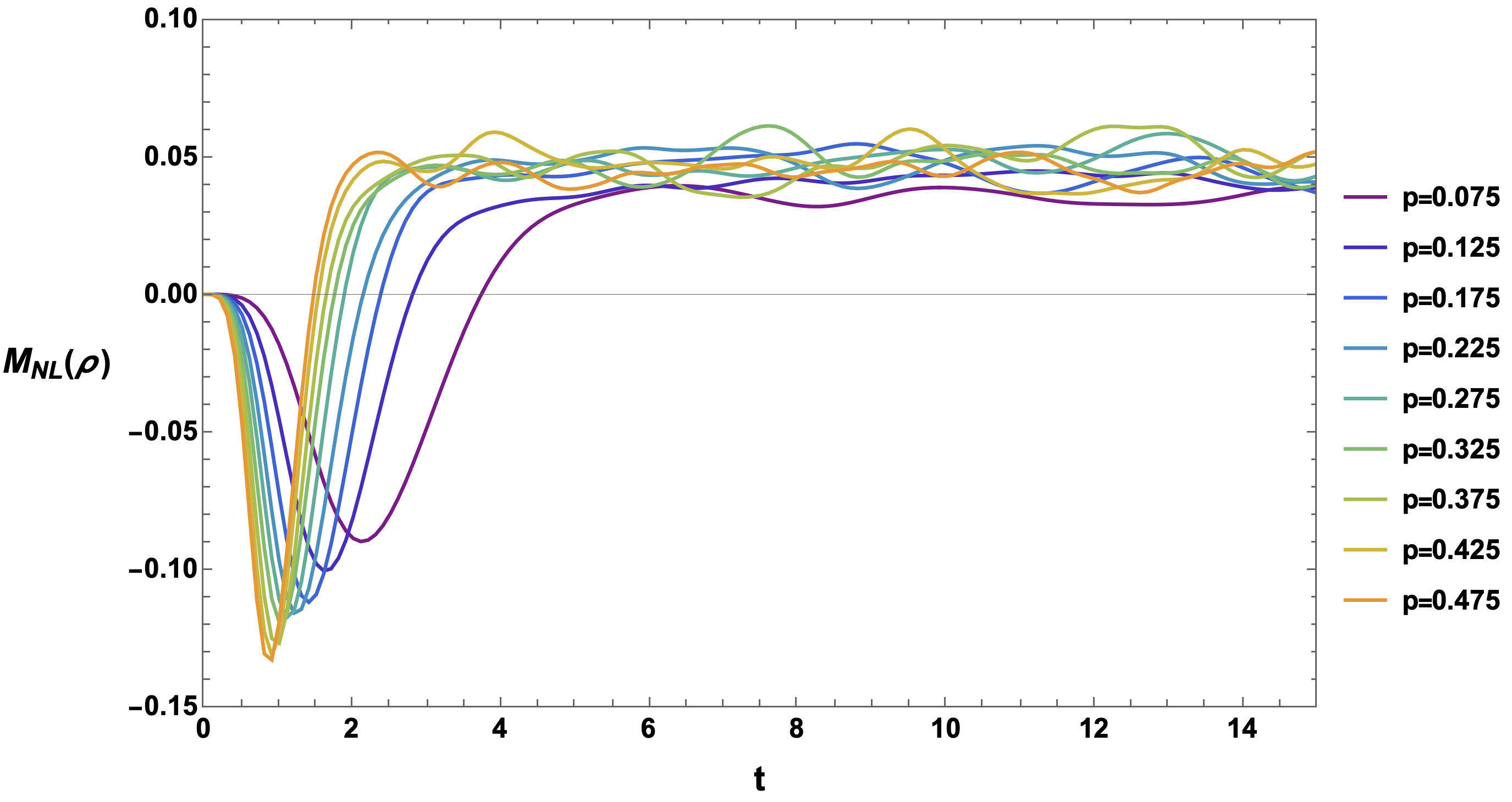}
    \caption{5 Qubits}\label{SPSYK8qNLSREZ}
  \end{subfigure}
  \begin{subfigure}{.3\linewidth}
    \includegraphics[height=3.5cm,width=\linewidth]{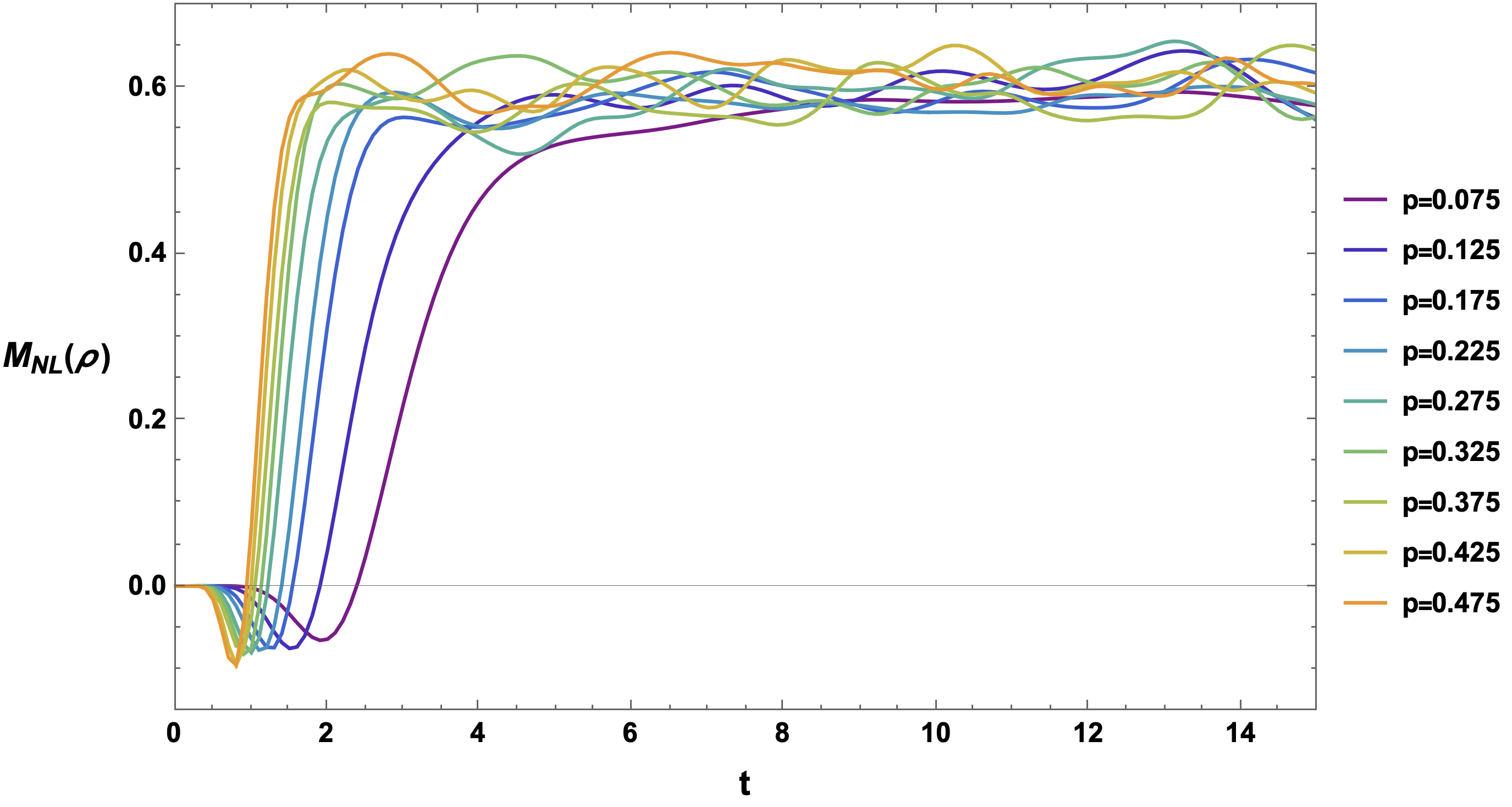}
    \caption{6 Qubits}\label{SPSYK6qNLSREZ}
  \end{subfigure}
  \begin{subfigure}{.3\linewidth}
    \includegraphics[height=3.5cm,width=\linewidth]{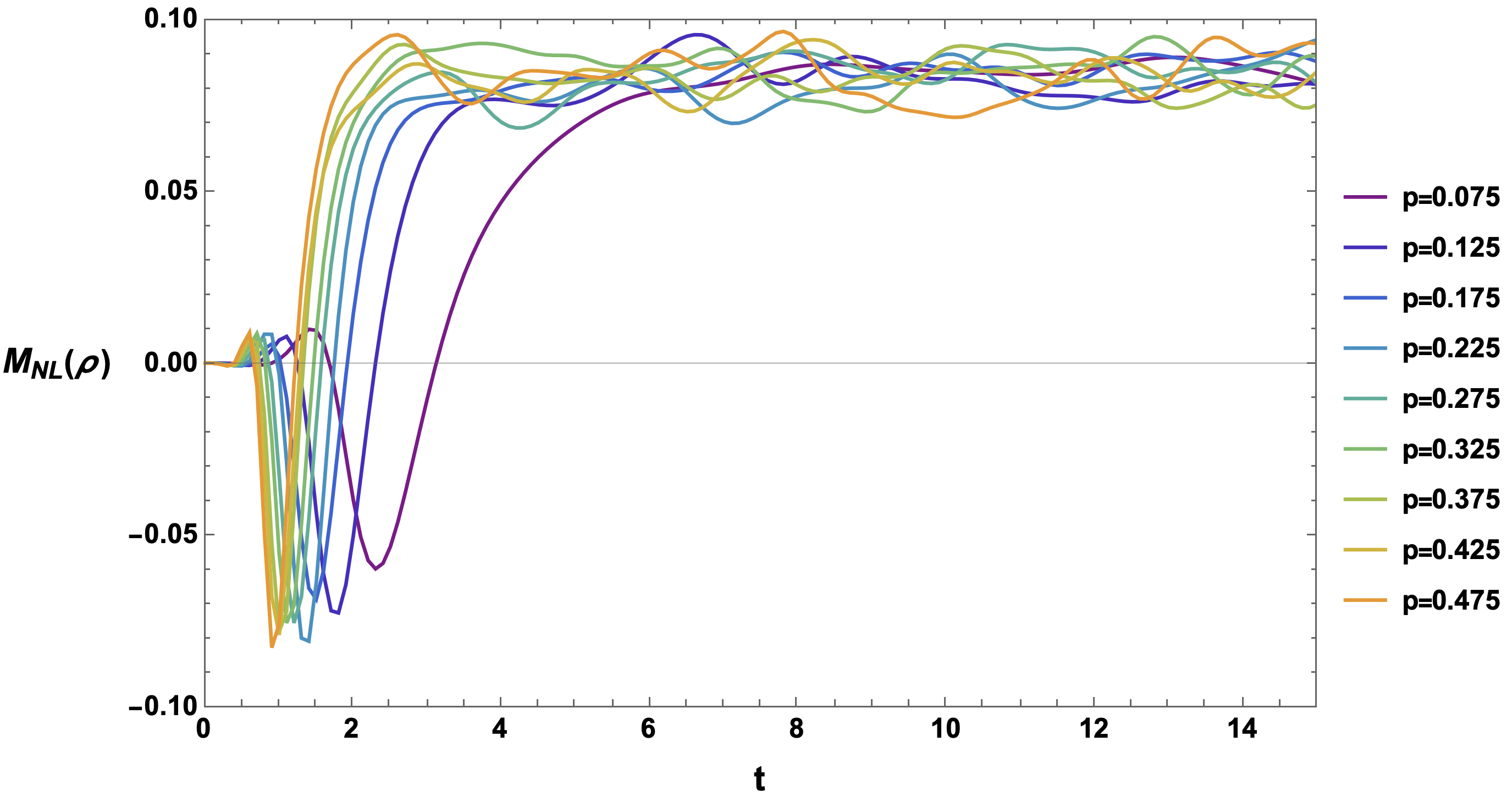}
    \caption{7 Qubits}\label{SPSYK7qNLSREZ}
  \end{subfigure}
  
  \caption{\footnotesize{For panels (a) $N=12$ (6 qubits), (b) $N=14$ (7 qubits) (C) $N=16$ (8 qubits) Non-local SRE for Sparse SYK. }}\label{NLSRESparse}
\end{figure}

\section{ ${\cal N}=2$ Super-symmetric SYK and BPS states}\label{sec 6}
In this section, we explore how the SRE and the multipartite non-local SRE behave in the ${\cal N}=2$ supersymmetric SYK model \cite{Fu:2016vas}. Recent interest in this system stems from the intriguing nature of its BPS sector: all BPS states in this model turn out to be \emph{fortuitous}, meaning that their cohomological properties rely crucially on finite-$N$ relations and do not extend smoothly to the large-$N$ limit  \cite{Chang:2024lxt}. Such states have been proposed as candidates for typical black-hole microstates in holographic settings and are believed to exhibit comparatively strong signatures of chaos. This stands in contrast to many other supersymmetric theories, where one also encounters \emph{monotonous} BPS states that persist at large $N$ and are often associated with smooth, horizonless configurations in the dual gravitational description.

The model itself is defined in terms of $N$ complex fermions, $\psi^i$ and $\bar{\psi}_i$ ($i=1,\dots,N$), from which a pair of conjugate supercharges, $Q$ and $Q^\dagger$, are constructed. These supercharges take the form
\begin{align}
    Q &= i \sum_{1 \leq i < j < k \leq N} C_{ijk}\, \psi^i \psi^j \psi^k , \\
    Q^\dagger &= i \sum_{1 \leq i < j < k \leq N} \overline{C}^{\,ijk}\, \bar{\psi}_i \bar{\psi}_j \bar{\psi}_k ,
\end{align}
where the coefficients $C_{ijk}$ are random complex variables sampled from a Gaussian ensemble with zero mean. Their variance is chosen to satisfy
\begin{align}
    \langle C_{ijk}\, \overline{C}^{\,ijk} \rangle = \frac{2J}{N^2},
\end{align}
ensuring a well-defined large-$N$ scaling of the interaction terms. We choose  operators $\bar{\psi}_i$  as fermionic creation operators, while $\psi^i$ act as annihilation operators. The vacuum $|0\rangle$ satisfies $\psi^i |0\rangle = 0$ for all $i$, and the full Fock space is generated by applying creation operators to it.  
A convenient occupation-number basis is
\begin{align}
|\lambda_1 \cdots \lambda_N\rangle
    = (\bar{\psi}_1)^{\lambda_1} (\bar{\psi}_2)^{\lambda_2} \cdots (\bar{\psi}_N)^{\lambda_N} |0\rangle,
    \qquad \lambda_i \in \{0,1\},
\end{align}
which spans a Hilbert space of dimension $2^N$.  

The system exhibits a $U(1)_R$ symmetry, under which the fermion number plays the role of the $R$-charge. Thus each basis vector has charge
\begin{align}
q_R = \sum_{i=1}^{N} \nu_i,
\end{align}
and the full Hilbert space partitions itself into charge sectors $\mathcal{H}_{q_R}$.

We will investigate the SRE and multipartite SRE behave for three different classes of states belonging to the same charge sector: BPS or Fortuitous States,$Q$-exact States and Typical States within a Charge Sector.

\textbf{\textit{(i) BPS or Fortuitous States.}  }
The supersymmetric ground states of $H = \{Q, Q^\dagger\}$ form the BPS sector. They correspond to nontrivial cohomology classes of the supercharge, namely elements of $\ker(Q)/\mathrm{im}(Q)$. Such states exist only in a restricted window of charge sectors~\cite{Fu:2016vas,Chang:2024lxt}.  
For even $N$ they occur at $q_R = \frac{N}{2}, \frac{N\pm 2}{2}$, while for odd $n$ they lie at $q_R = \frac{N\pm 1}{2}, \frac{N\pm 3}{2}$.  
To probe typical SRE properties of the BPS subspace, we construct random superpositions of cohomology representatives:
\begin{align}
|\Psi_{\mathrm{BPS}}\rangle
    = \frac{1}{M} \sum_{i} a_i\, |\phi_i\rangle,
    \qquad |\phi_i\rangle \in \ker(Q)/\mathrm{im}(Q),
\end{align}
with $a_i$ drawn from a Gaussian distribution and $M$ ensuring normalization.

\textbf{\textit{(ii) $Q$-exact States.} } 
Cohomologically trivial states belong to $\mathrm{im}(Q)$ and provide a natural comparison set.  
We generate representative $Q$-exact states of charge $q_R$ by acting with $Q$ on random states drawn from the sector with charge $q_R+3$:
\begin{align}
|\Psi_{\mathrm{Q\!-\!exact}}\rangle
    = \frac{1}{M} \sum_{\{\lambda_i\}} b_{\{\lambda_i\}}\, 
      Q |\lambda_1 \cdots \lambda_N\rangle,
      \qquad \sum_{i=1}^{N} \lambda_i = q_R + 3,
\end{align}
with $b_{\{\lambda_i\}}$ chosen randomly.

\textbf{\textit{(iii) Typical States within a Charge Sector.}  }
As a reference for generic behavior of SRE and multipartite SRE, we also consider Haar-random states sampled directly from the full charge sector $\mathcal{H}_{q_R}$:
\[
|\Psi_{\mathrm{typical}}\rangle \in \mathcal{H}_{q_R}.
\]
This ensemble provides a baseline against which the structured BPS and $Q$-exact states may be compared.

\subsection{Stabilizer Renyi Entropy}

\begin{figure}[H]
  \centering
  \begin{subfigure}{.3\linewidth}
\includegraphics[height=3.5cm,width=\linewidth]{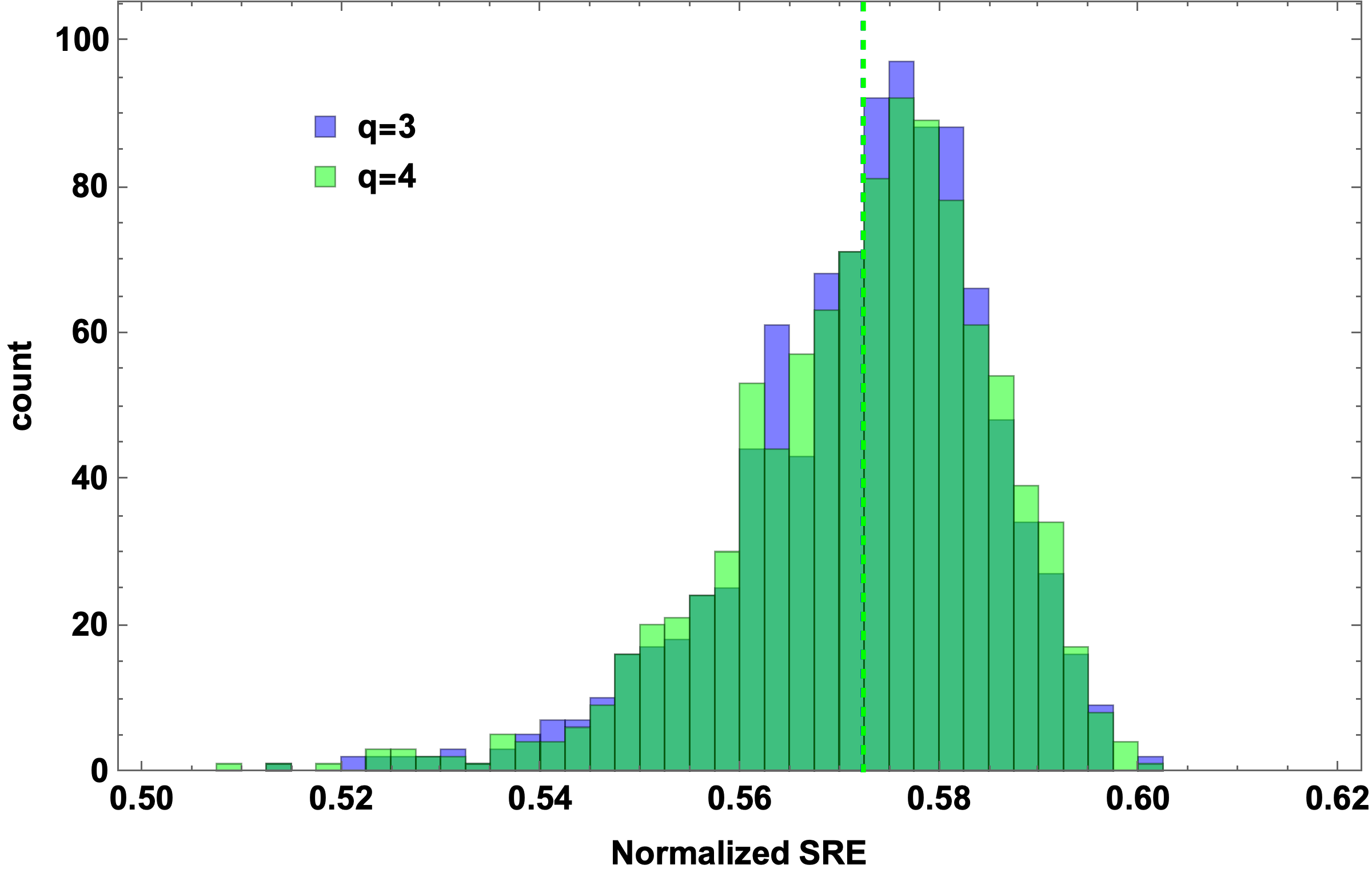}
    \caption{$N=7$ }\label{SUSYBPSO1}
  \end{subfigure}\quad
  \begin{subfigure}{.3\linewidth}
    \includegraphics[height=3.5cm,width=\linewidth]{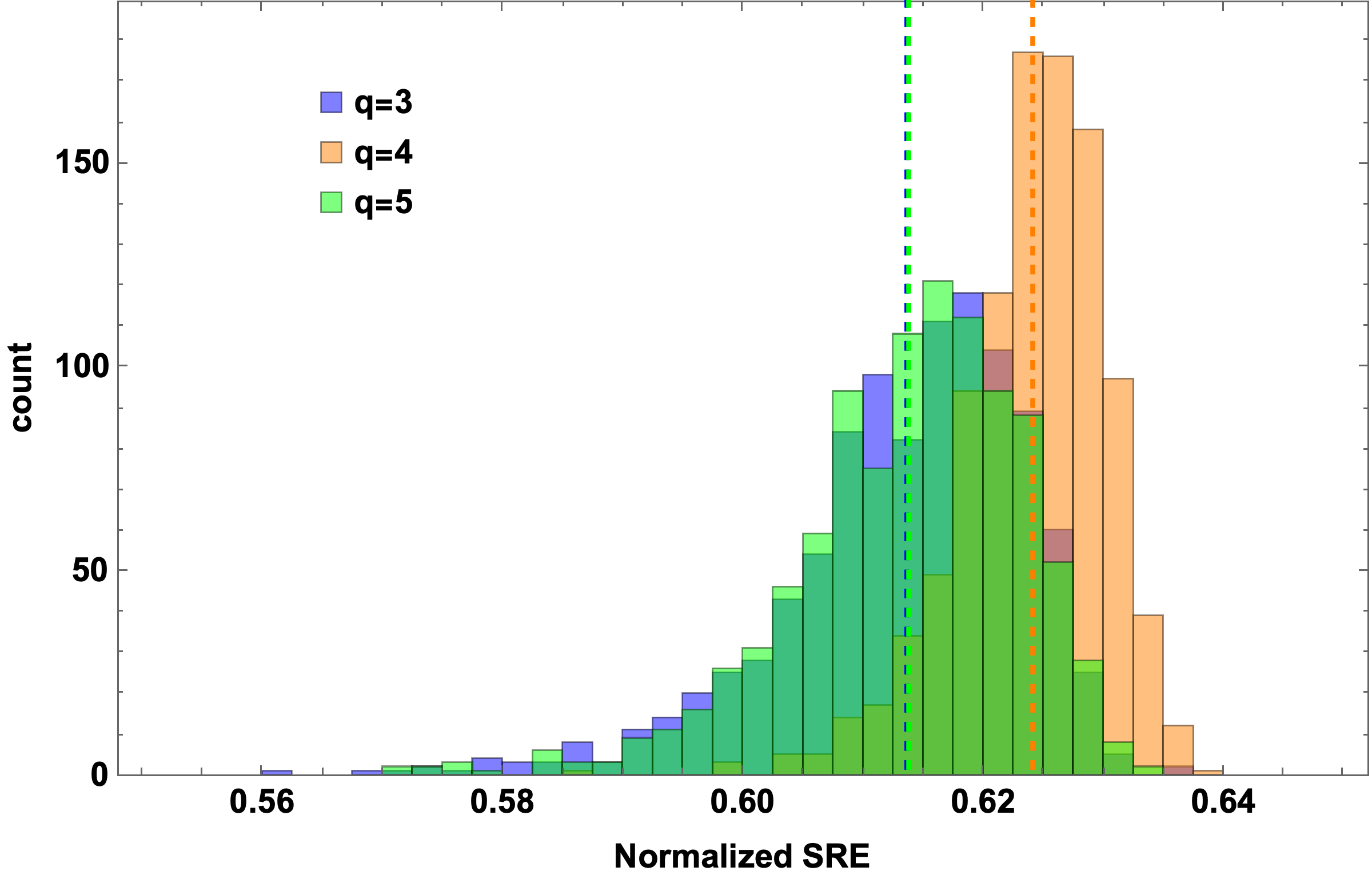}
    \subcaption{$N=8$}\label{SUSYBPSO2}
  \end{subfigure}\quad
  \begin{subfigure}{.3\linewidth}
    \includegraphics[height=3.5cm,width=\linewidth]{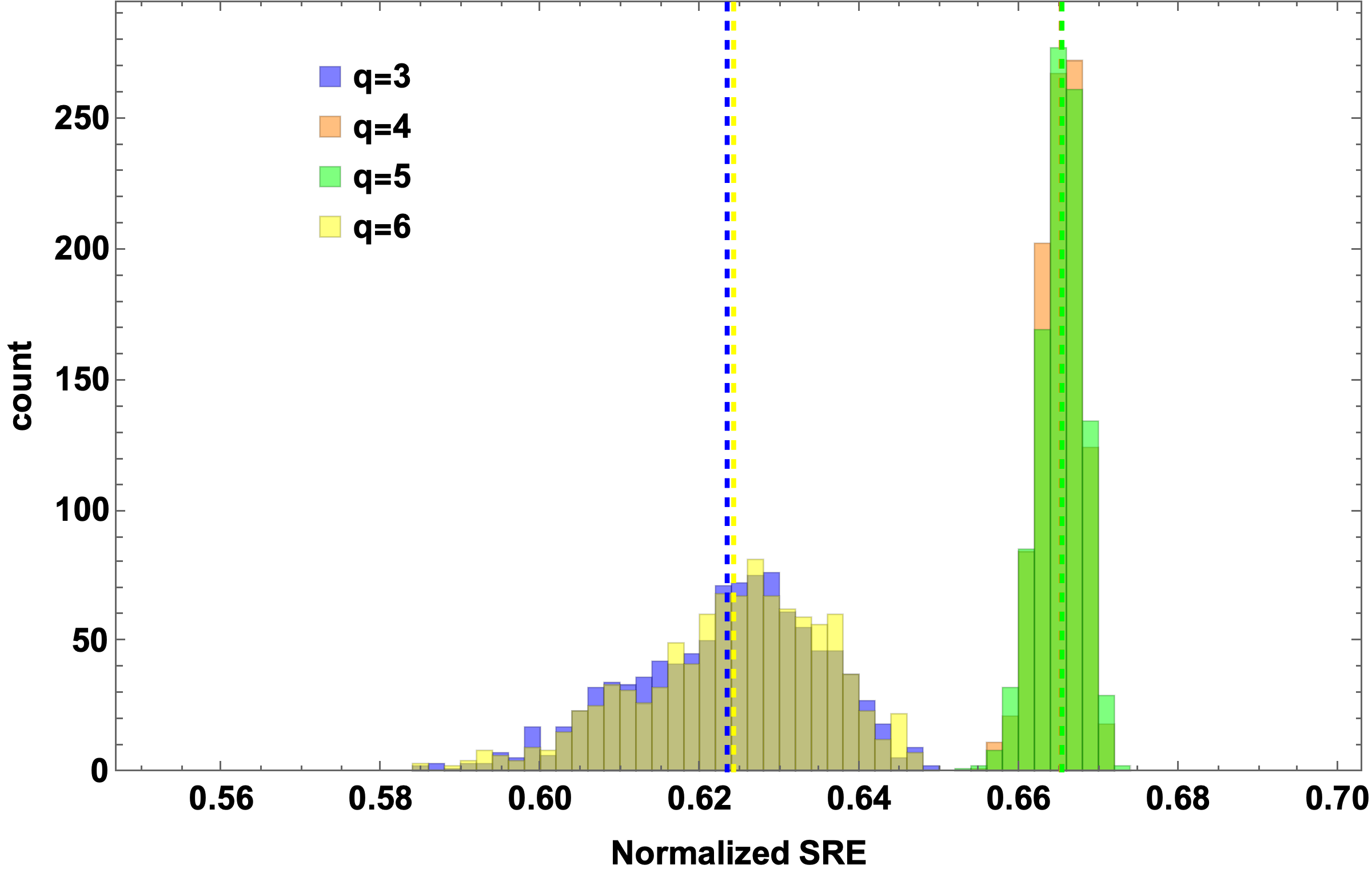}
    \subcaption{$N=9$}\label{SUSYBPSO3}
  \end{subfigure}

   \caption{\footnotesize{ The plot displays the distribution of the normalized stabilizer Rényi entropy among typical BPS states for different $R$-charge sectors for $N=7,8,9$.
}}\label{BPS789}
\end{figure}

We begin by evaluating the stabilizer Rényi entropy of the typical BPS states for $N = 7,8,9$, as shown in \Cref{BPS789}. For $N = 7$, the SRE values associated with the two allowed $R$-charge sectors, $q_R = 3$ and $q_R = 4$, are nearly indistinguishable and lie on top of each other. In contrast, for $N = 8$ and $N = 9$, a clear structure emerges: the BPS states residing in the central charge sectors namely $q_R = 4$ for $N = 8$ and $q_R = 4,5$ for $N = 9$ exhibit noticeably larger SRE compared to those in the edge charge sectors (  $q_R = 3,5$ for $N = 8$ and $q_R = 3,6$ for $N = 9$).

\begin{figure}[H]
  \centering
  \begin{subfigure}{.45\linewidth}
\includegraphics[height=4.5cm,width=\linewidth]{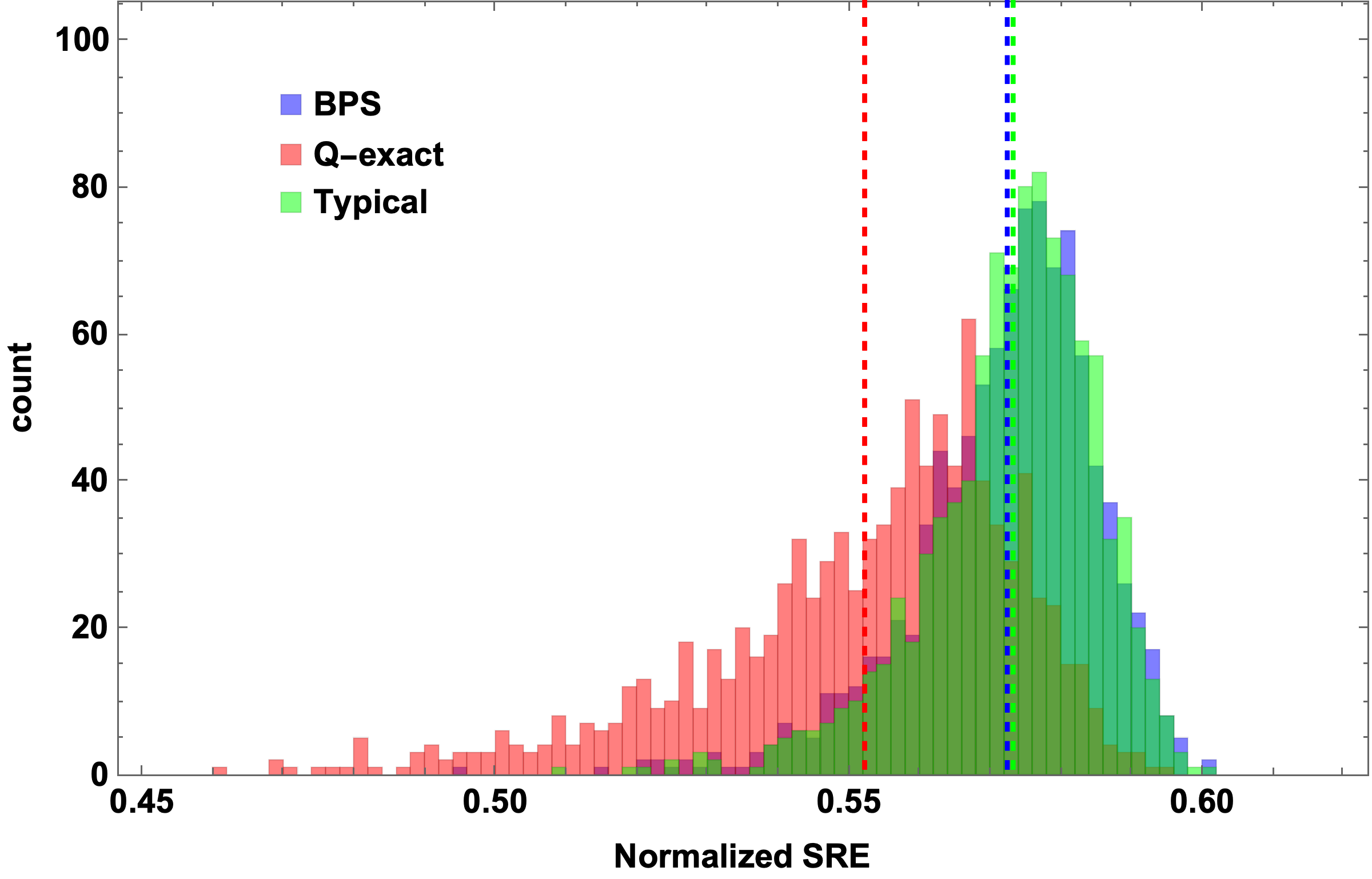}
    \caption{$N=7$, $q_R=3$}\label{SUSY73}
  \end{subfigure}\quad
  \begin{subfigure}{.45\linewidth}
    \includegraphics[height=4.5cm,width=\linewidth]{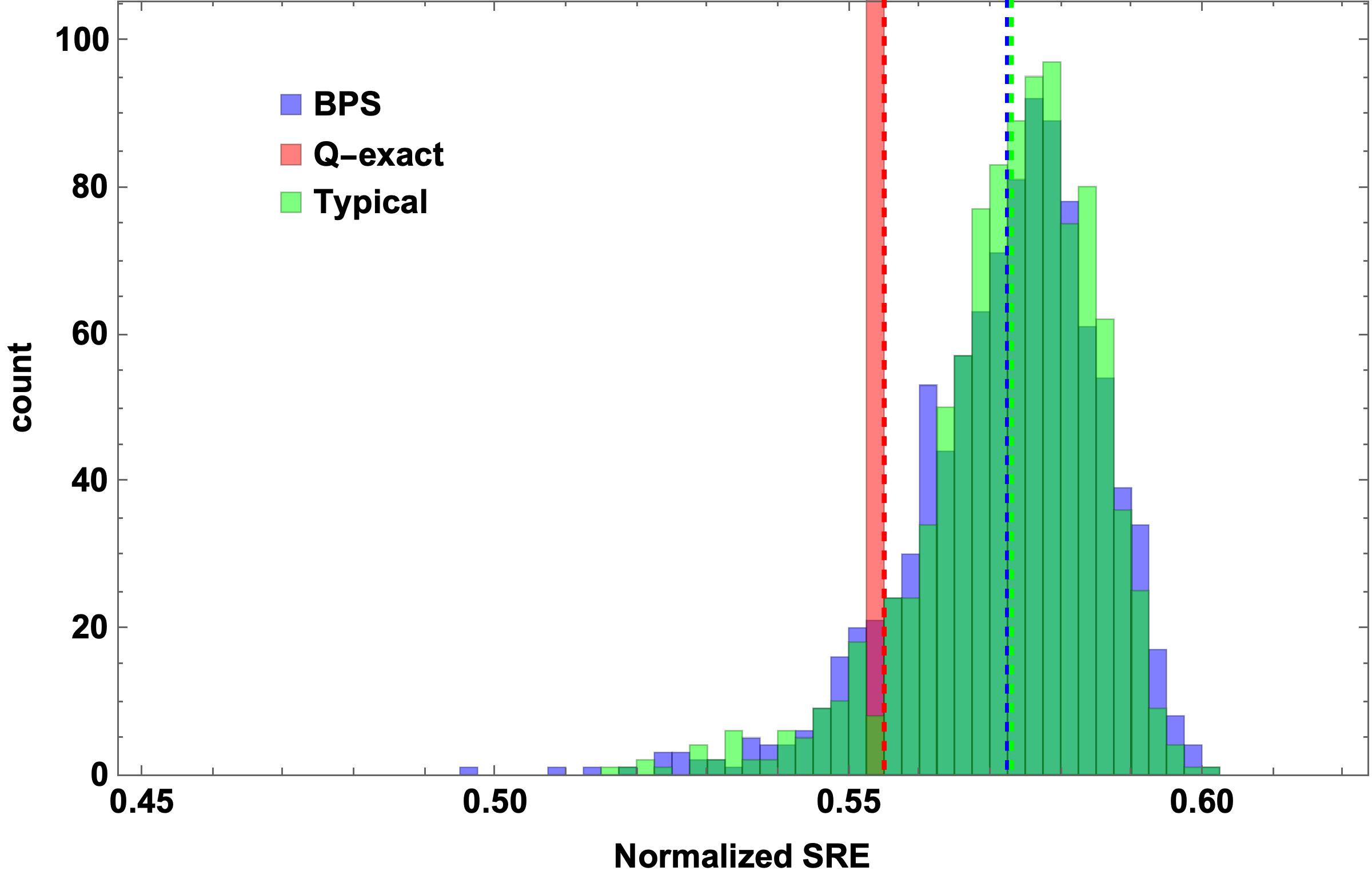}
    \subcaption{$N=7$, $q_R=4$}\label{SUSY74}
  \end{subfigure}
   \caption{\footnotesize{ The plot displays the distribution of the normalized stabilizer Rényi entropy among typical BPS states, typical $Q$-exact states, and states randomly sampled within the $R$-charge sector for $N=7$.
}}\label{SUSY7}
\end{figure}
  Following this analysis, in \Cref{SUSY7}, \Cref{SUSY8} and \Cref{SUSY9} we present the distributions of SRE for the BPS, $Q$-exact, and typical states across the allowed charge sectors. A consistent pattern emerges: in nearly every case, the average stabilizer Rényi entropy satisfies
\begin{align}
    SRE_2(\text{typical}) \gtrsim SRE_2(\text{BPS}).
\end{align}
This indicates that Haar-typical states within a fixed charge sector possess more non-stabilizerness than the BPS states, implying that such generic states are computationally relatively harder to simulate using stabilizer-based methods. In comparison, the BPS sector despite representing fortuitous black-hole-like microstates exhibits a somewhat reduced level of stabilizer complexity.

A more detailed structure appears when examining the central $R$-charge sectors where the BPS states reside. In these regions we observe the ordering
\begin{align}
    SRE_2(\text{typical}) \gtrsim SRE_2(\text{BPS}) > SRE_2(\text{Exact}),
\end{align}
showing that $Q$-exact states consistently display the smallest stabilizer content. This hierarchy highlights how the cohomological character of each ensemble influences its SRE: BPS states lie between the trivial $Q$-exact states and the fully random typical states. Thus, while BPS microstates possess nontrivial stabilizer entropy, they remain less complex than Haar-typical states within the same charge sector.

This hierarchy in the central sector has a natural structural origin. Typical states in a fixed R-charge sector are largely unstructured and therefore tend to exhibit large values of basis/frame-dependent complexity measures, like the SRE. Instead, BPS states are defined by the cohomological constraint imposed by supersymmetry, namely they are annihilated by $Q$ and $Q^\dagger$ and thus represent highly non-generic subspaces of state vectors satisfying non-trivial linear relations. These constraints reduce the degree to which BPS states can look like an arbitrary typical state, therefore leading to a mitigated spread in the stabilizer basis. This in turn leads to slighly reduced values of the SRE. Finally, $Q$-exact states exhibit an even stronger structure compared to a generic $Q$-closed state because they are explicitly in the image of $Q$. They are therefore expected to exhibits the most pronounced departure from the behavior of an arbitrary typical state. However, the contrasting behaviour in the edge charge sectors is something that needs to better understood.

\begin{figure}[H]
  \centering
  \begin{subfigure}{.3\linewidth}
\includegraphics[height=3.5cm,width=\linewidth]{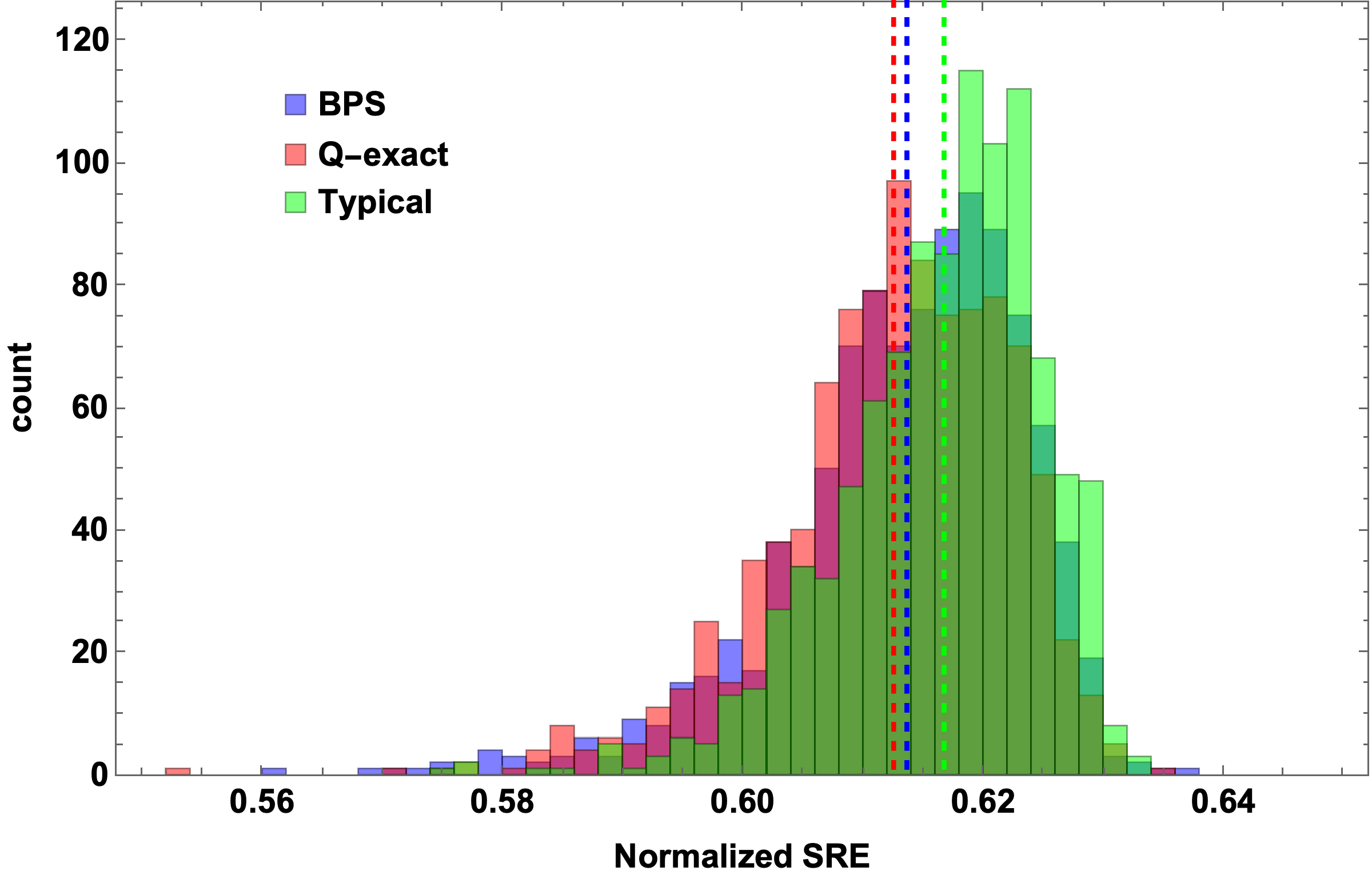}
    \caption{$N=8$, $q_R=3$}\label{SUSY83}
  \end{subfigure}\quad
  \begin{subfigure}{.3\linewidth}\includegraphics[height=3.5cm,width=\linewidth]{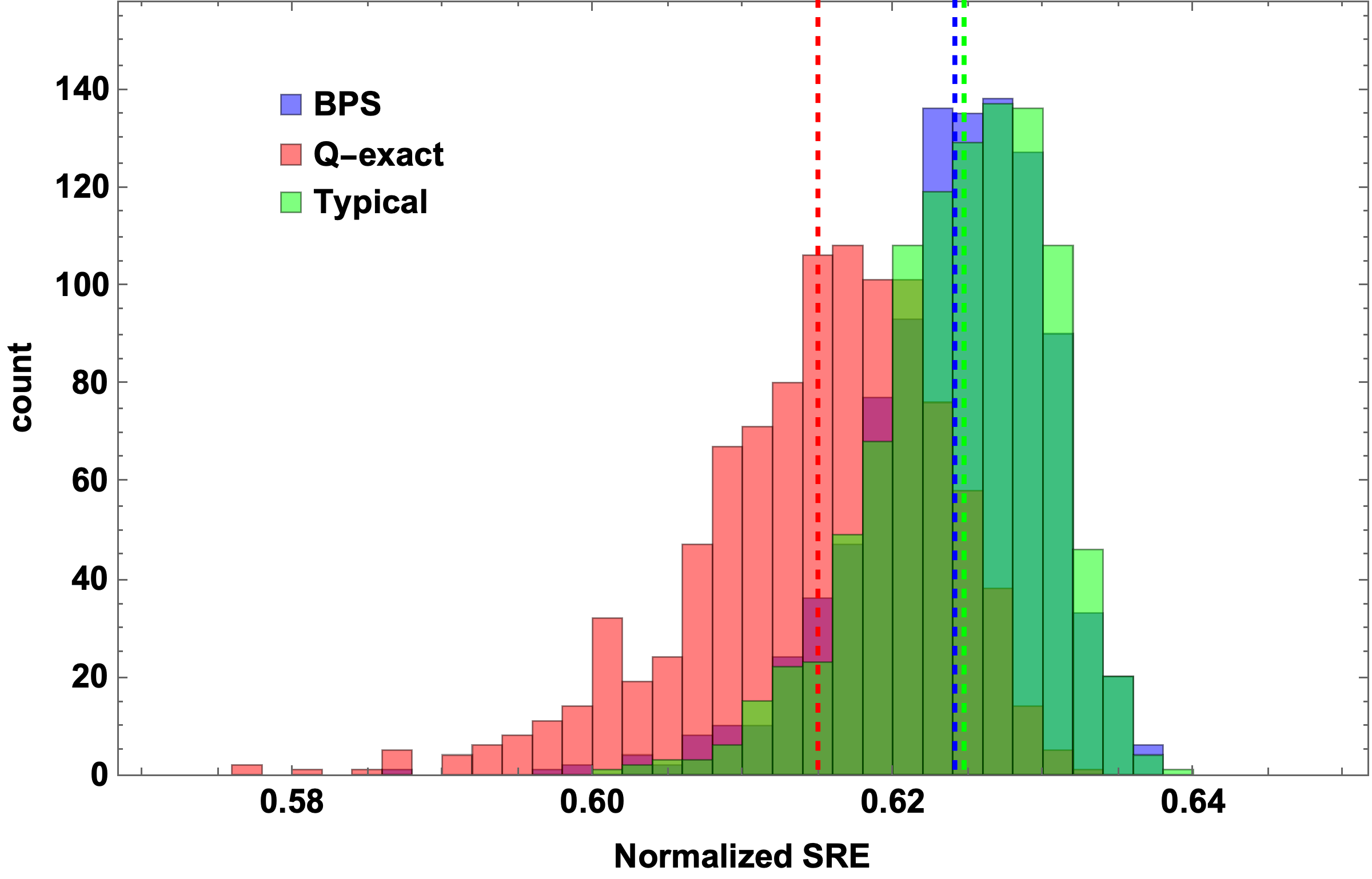}
    \subcaption{$N=8$, $q_R=4$}\label{SUSY84}
  \end{subfigure}\quad
    \begin{subfigure}{.3\linewidth}
\includegraphics[height=3.5cm,width=\linewidth]{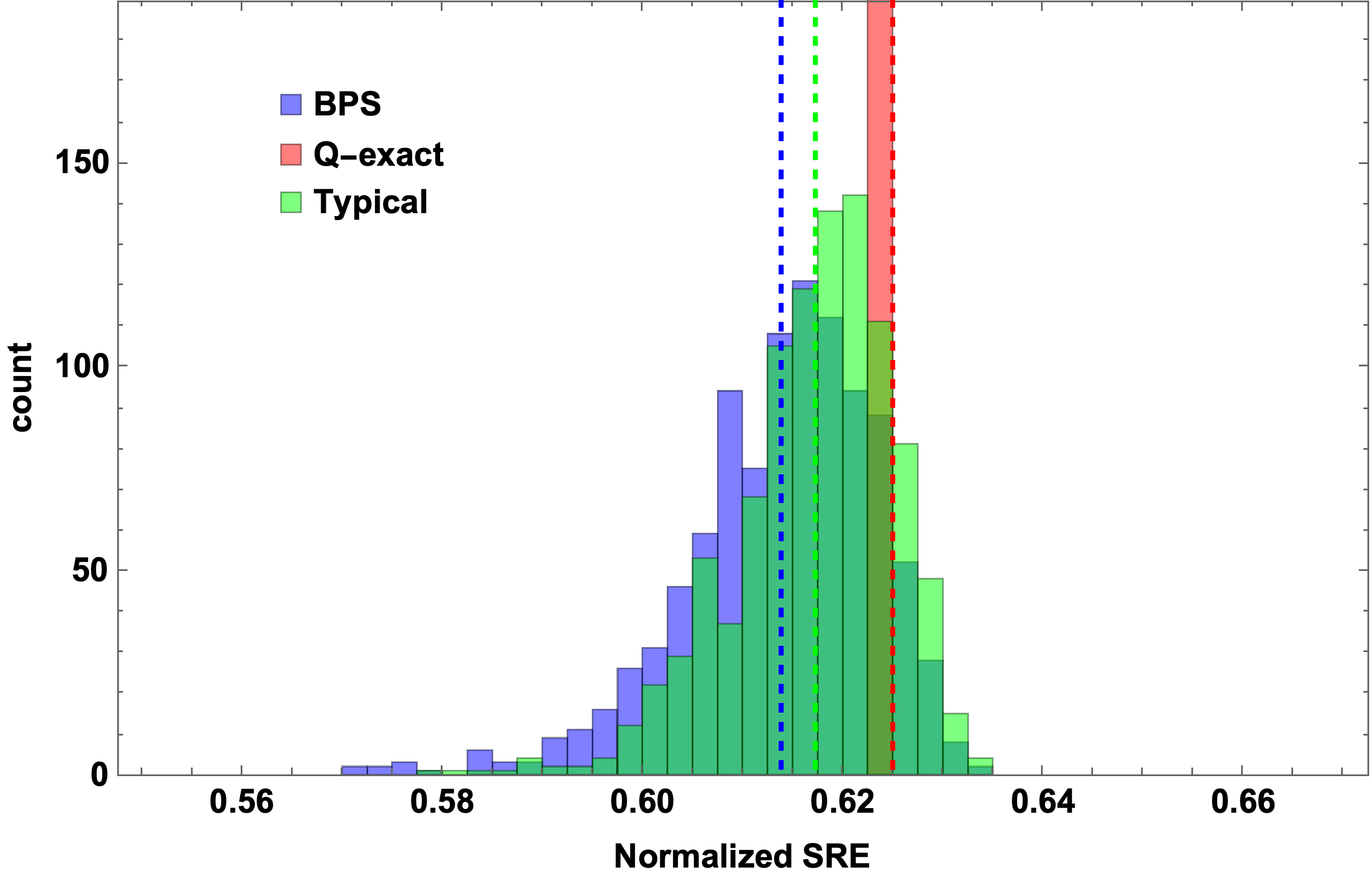}
    \caption{$N=8$, $q_R=5$}\label{SUSY85}
  \end{subfigure}
   \caption{\footnotesize{ The plot displays the distribution of the normalized stabilizer Rényi entropy among typical BPS states, typical $Q$-exact states, and states randomly sampled within the $R$-charge sector for $N=8$.
}}\label{SUSY8}
\end{figure}

\begin{figure}[H]
  \centering
  \begin{subfigure}{.45\linewidth}
\includegraphics[height=4.5cm,width=\linewidth]{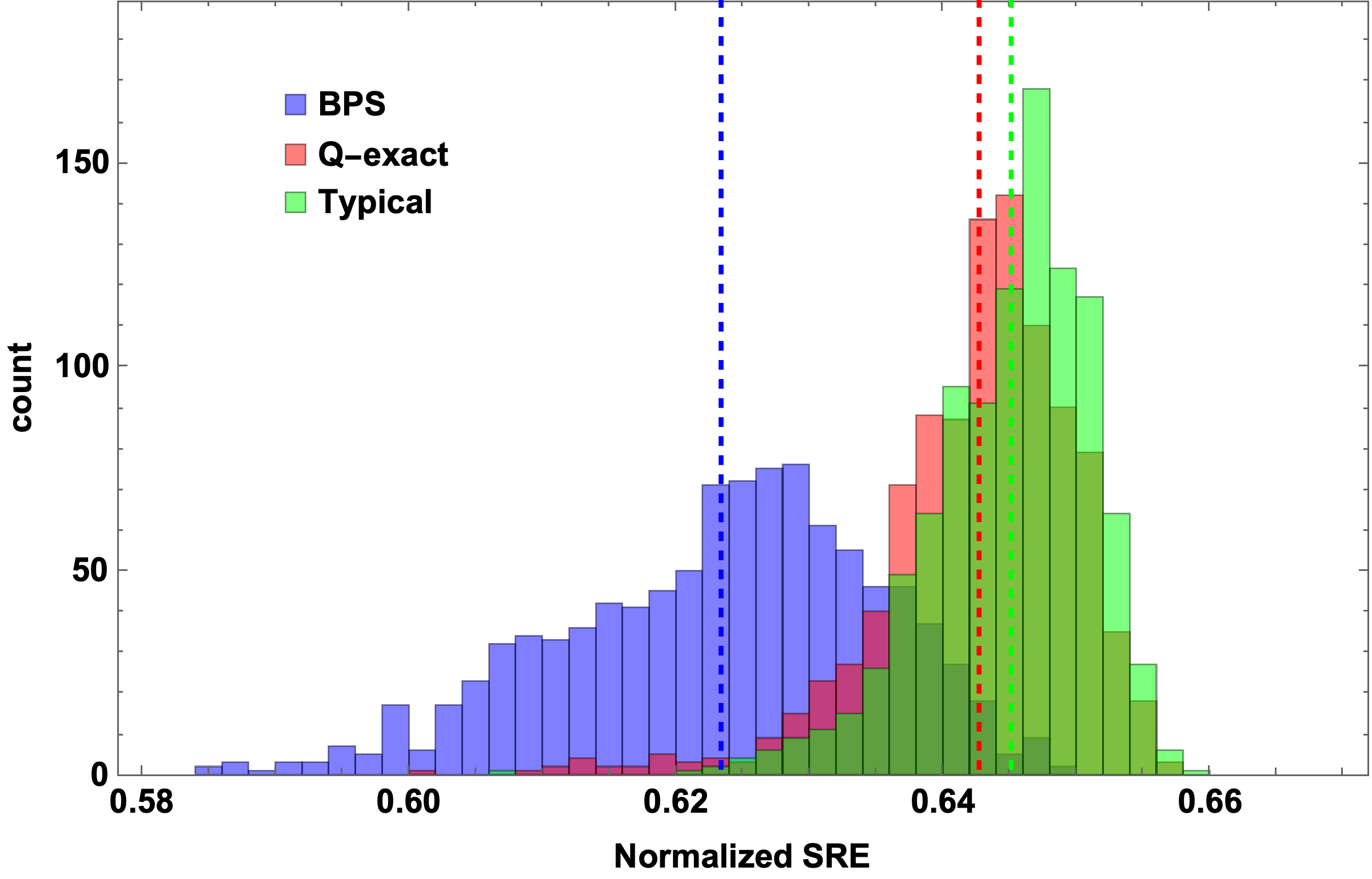}
    \caption{$N=9$, $q_R=3$}\label{SUSY93}
  \end{subfigure}\quad
  \begin{subfigure}{.45\linewidth}\includegraphics[height=4.5cm,width=\linewidth]{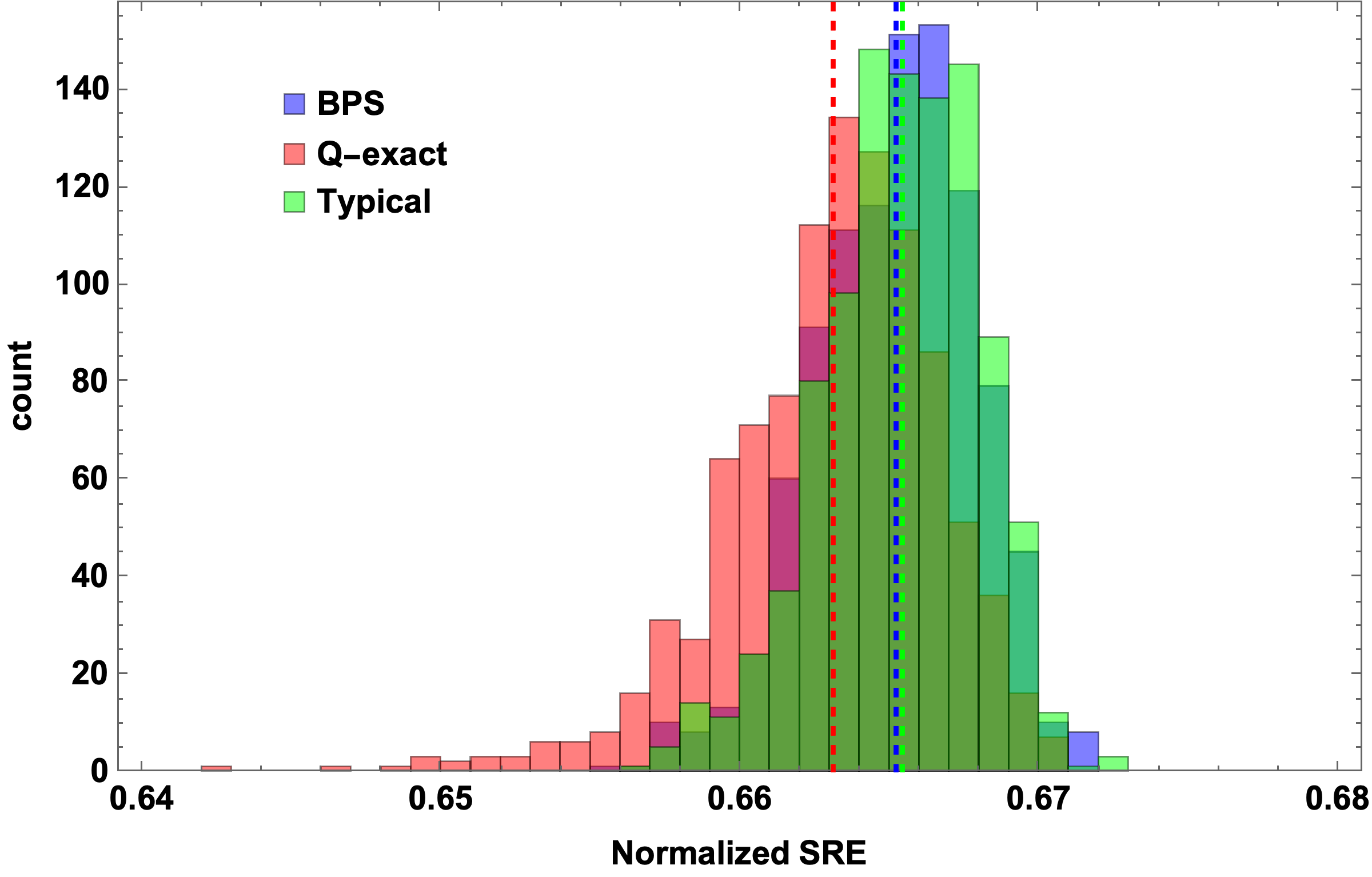}
    \subcaption{$N=9$, $q_R=4$}\label{SUSY94}
  \end{subfigure}\quad\\
    \begin{subfigure}{.45\linewidth}
\includegraphics[height=4.5cm,width=\linewidth]{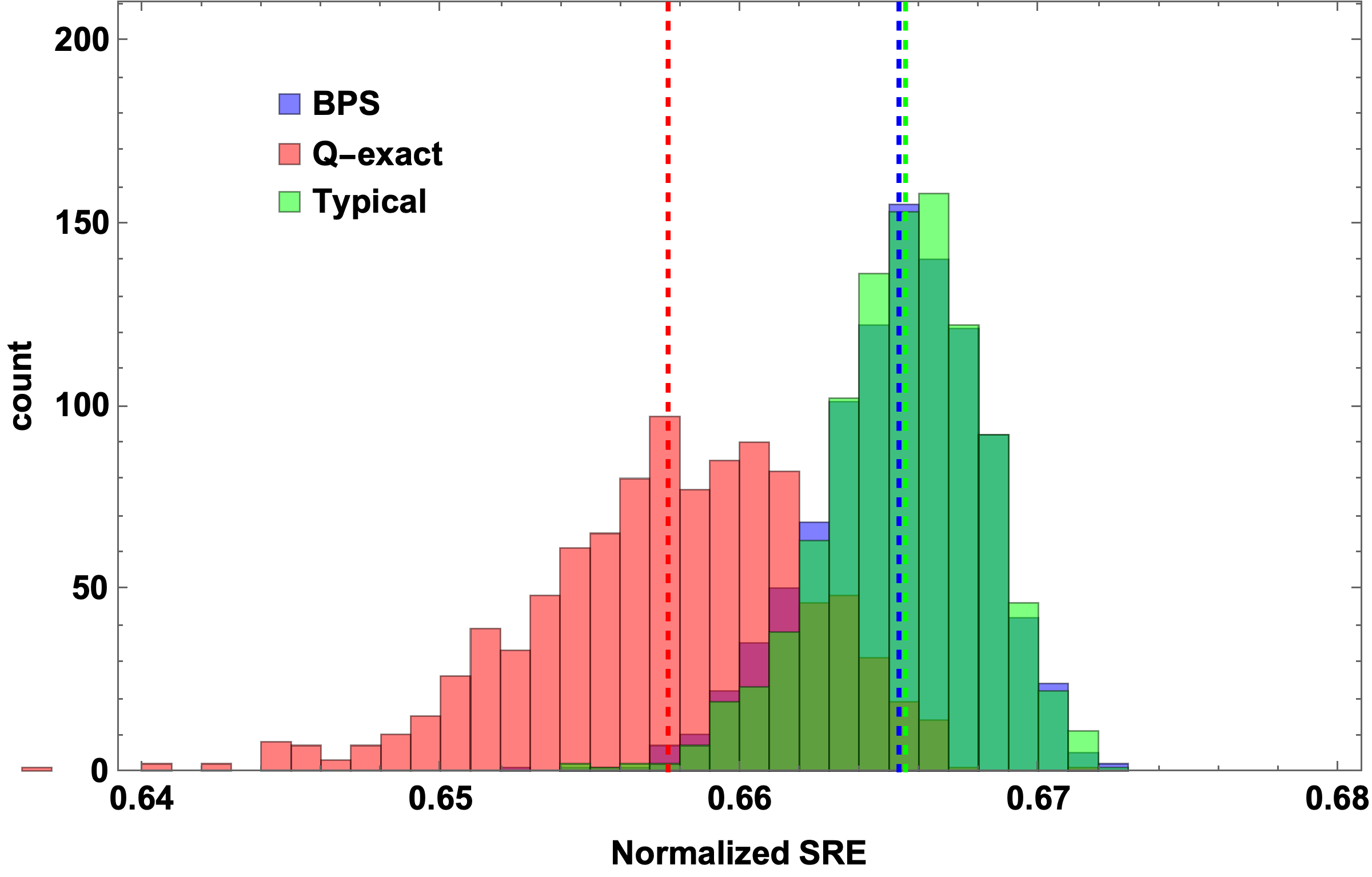}
    \caption{$N=9$, $q_R=5$}\label{SUSY95}
  \end{subfigure}
  \begin{subfigure}{.45\linewidth}
\includegraphics[height=4.5cm,width=\linewidth]{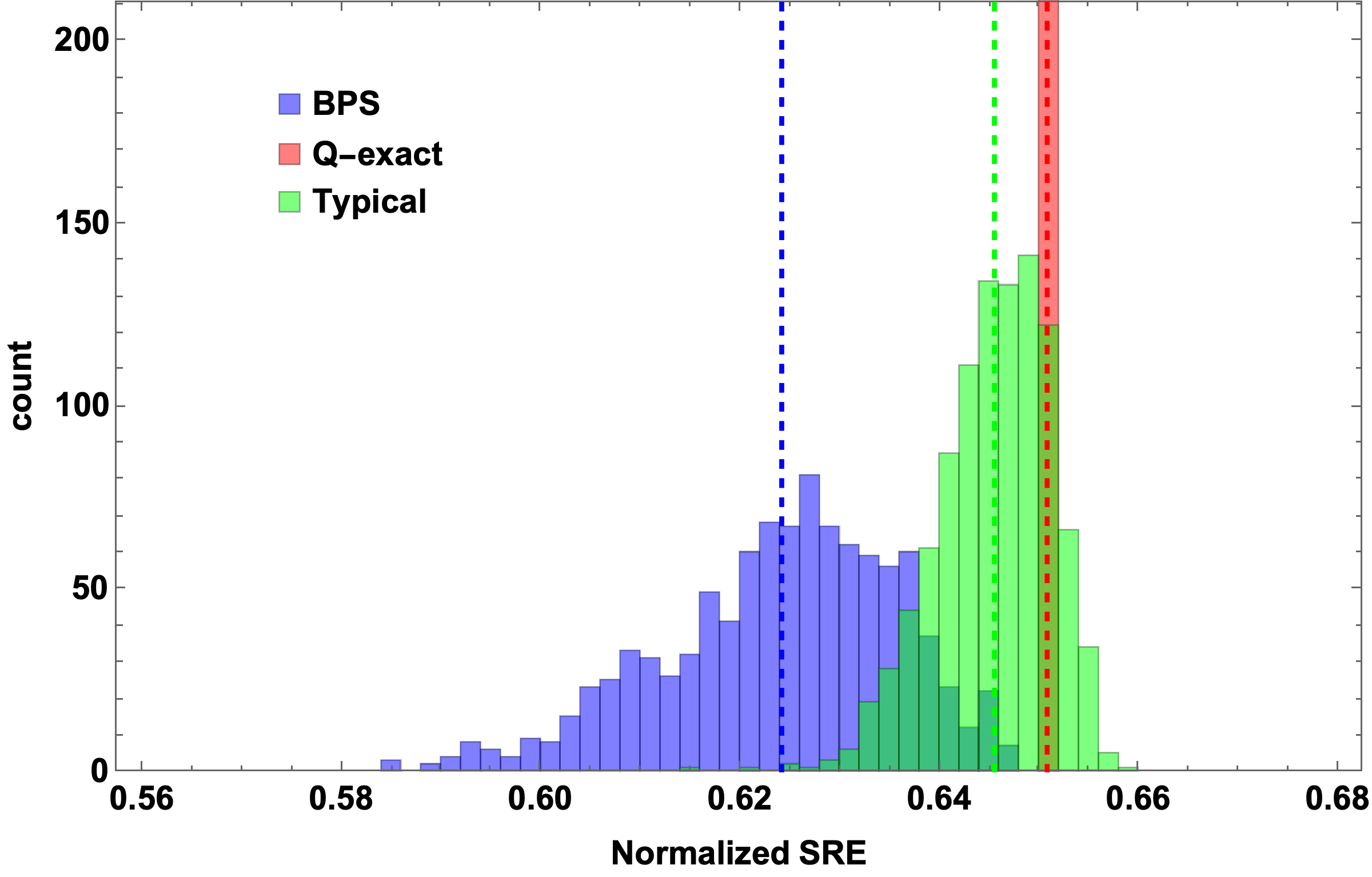}
    \caption{$N=9$, $q_R=6$}\label{SUSY96}
  \end{subfigure}
   \caption{\footnotesize{ The plot displays the histogram of the normalized stabilizer Rényi entropy among typical BPS states, typical $Q$-exact states, and states randomly sampled within the $R$-charge sector for $N=9$.
}}\label{SUSY9}
\end{figure}

\subsection{Multipartite non-local SRE}
Having characterized the behavior of the $SRE$ in the ${\cal N}=2$ supersymmetric SYK model, we now turn to the multipartite non-local SRE. This quantity captures the portion of non-stabilizerness that arises specifically from genuinely global, multipartite correlations rather than from local or few-body contributions. Examining it therefore provides a more refined understanding of how \textit{magic} is distributed across the many-body degrees of freedom in this model.

In \Cref{NLBPS789} we display the histograms of the multipartite SRE for the typical BPS states across the allowed $R$-charge sectors for $N = 7, 8, 9$. These plots illustrate how the strength of global non-local correlations varies with both system size and charge sector, shedding light on the structural features of fortuitous BPS microstates.

For $N = 7$, the multipartite SRE values corresponding to the two allowed $R$-charge sectors, $q_R = 3$ and $q_R = 4$, are nearly identical and lie essentially on top of each other. For $N = 9$, however, a much richer pattern emerges: the BPS states in the central charge sectors $q_R = 4,5$ exhibit significantly larger multipartite SRE compared to those in the edge sectors $q_R = 3,6$. This mirrors the behavior observed earlier for the stabilizer Rényi entropy.

Interestingly, the situation differs for $N = 8$. Although $q_R = 4$ is the central $R$-charge sector, the corresponding BPS states display \emph{lower} multipartite SRE than those in the neighboring edge sectors $q_R = 3$ and $q_R = 5$. This departure from the pattern seen in the ordinary SRE suggests that global multipartite correlations do not always follow the same charge-sector hierarchy as the local or few-body contributions captured by SRE.

\begin{figure}[H]
  \centering
  \begin{subfigure}{.3\linewidth}
    \includegraphics[height=3.5cm,width=\linewidth]{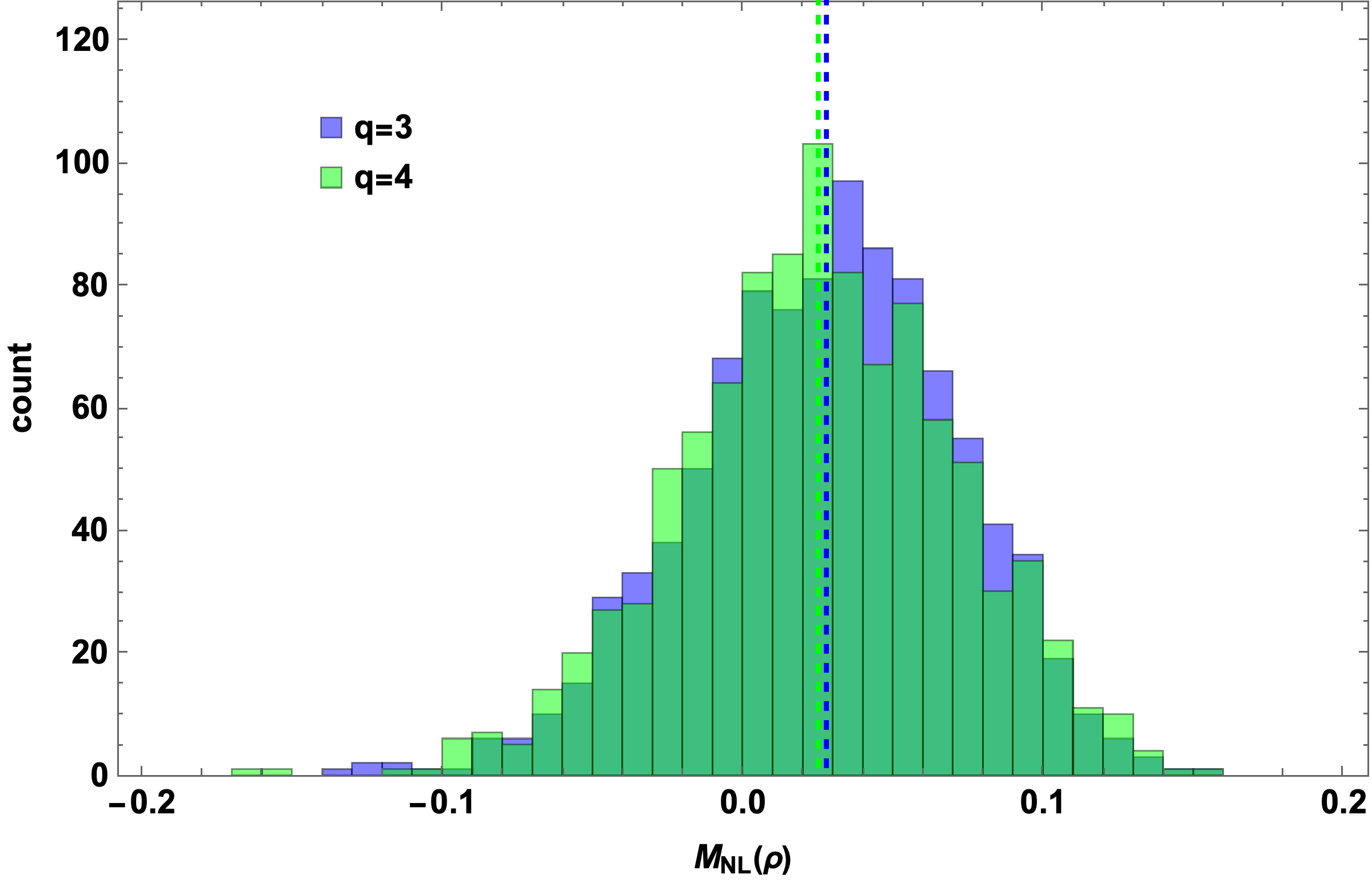}
    \subcaption{$N=7$}\label{NLSUSYBPSO1}
  \end{subfigure}\quad
  \begin{subfigure}{.3\linewidth}
    \includegraphics[height=3.5cm,width=\linewidth]{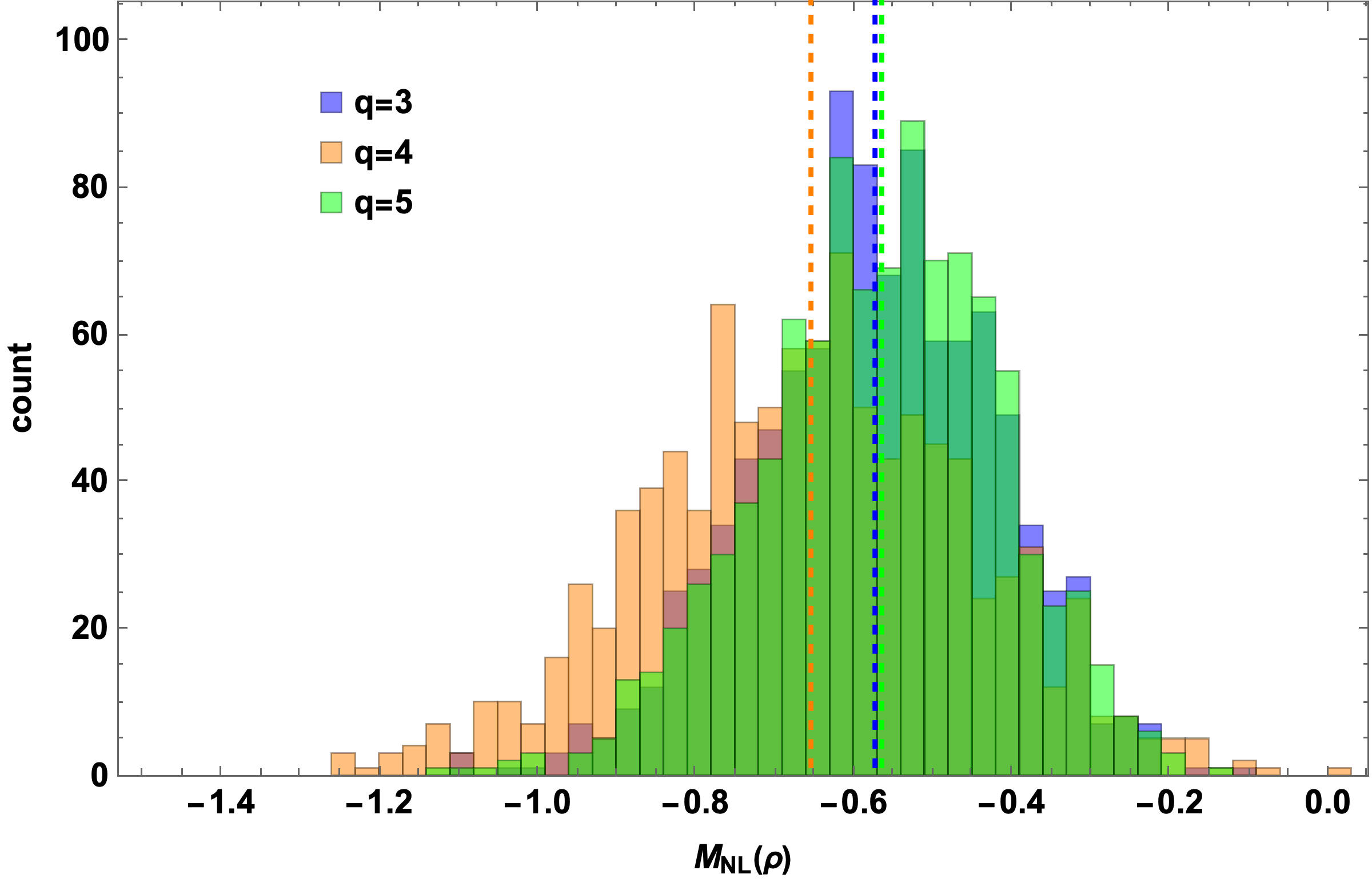}
    \subcaption{$N=8$}\label{NLSUSYBPSO2}
  \end{subfigure}\quad
  \begin{subfigure}{.3\linewidth}
    \includegraphics[height=3.5cm,width=\linewidth]{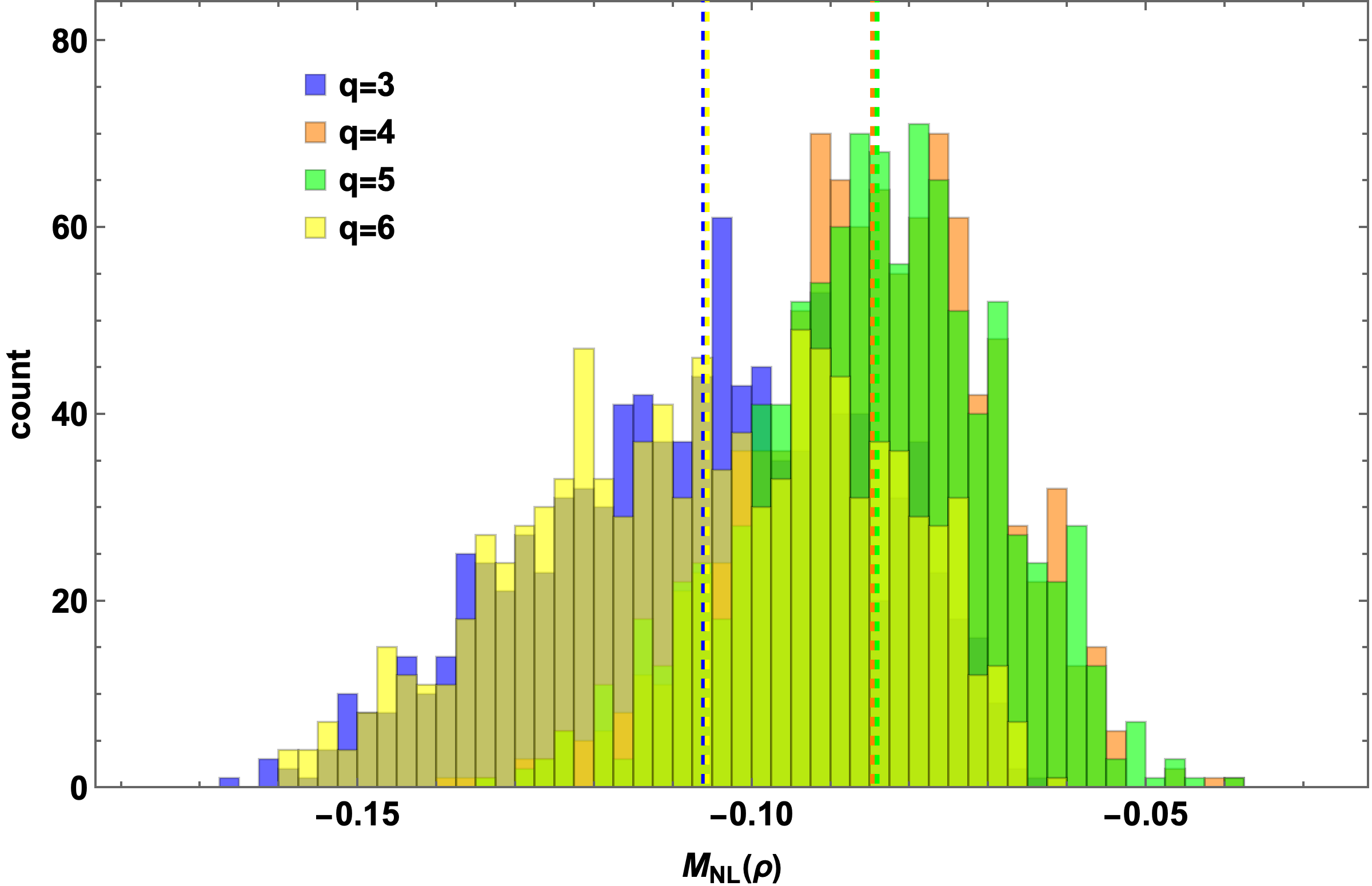}
    \subcaption{$N=9$}\label{NLSUSYBPSO3}
  \end{subfigure}

   \caption{\footnotesize{ The plot displays the distribution of the multipartite stabilizer Rényi entropy among typical BPS states for different $R$-charge sectors for $N=7,8,9$.
}}\label{NLBPS789}
\end{figure}

\begin{figure}[h]
  \centering
  \begin{subfigure}{.3\linewidth}
\includegraphics[height=3.5cm,width=\linewidth]{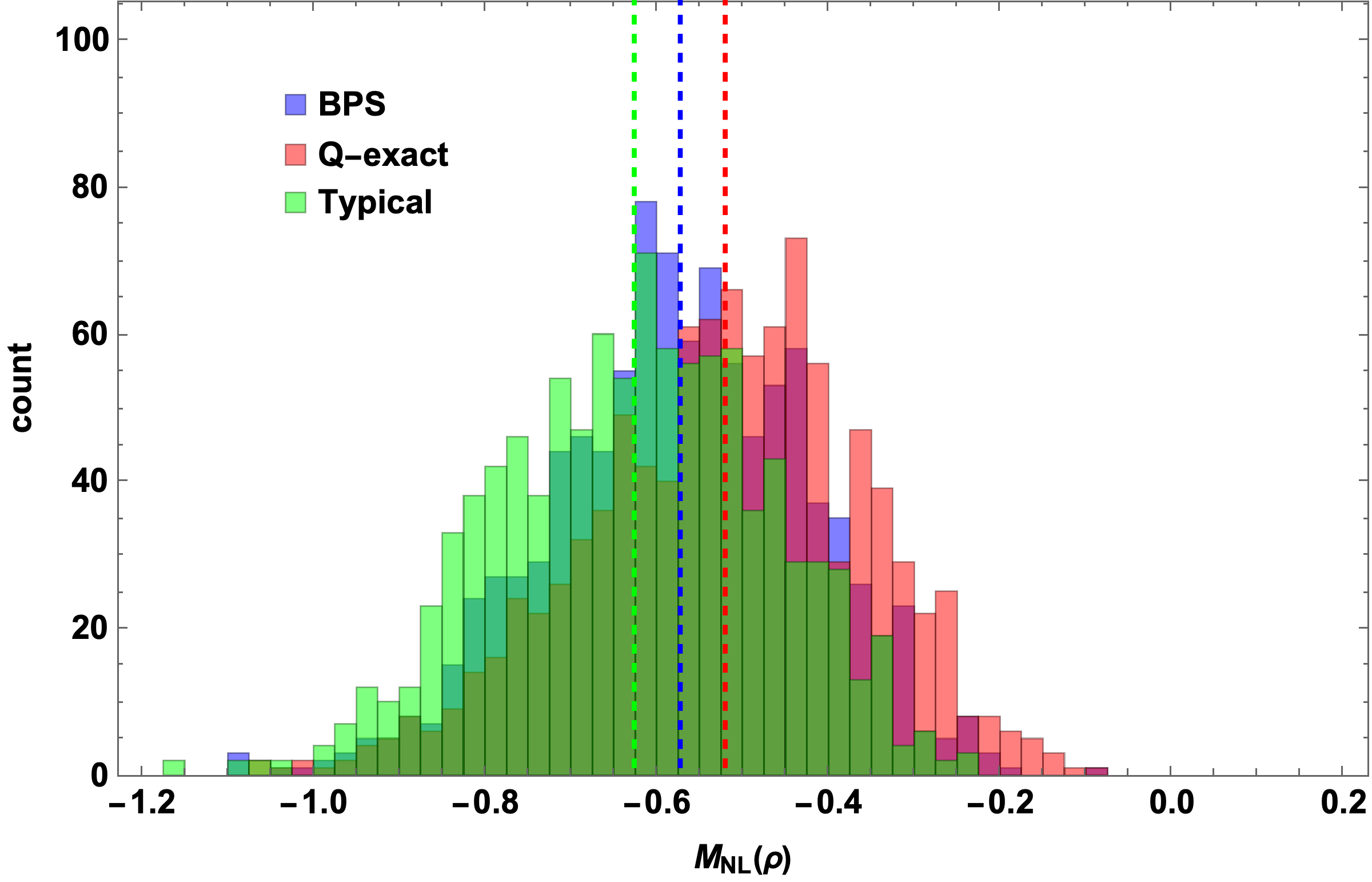}
    \caption{$N=8$, $q_R=3$}\label{NLSUSY83}
  \end{subfigure}\quad
  \begin{subfigure}{.3\linewidth}\includegraphics[height=3.5cm,width=\linewidth]{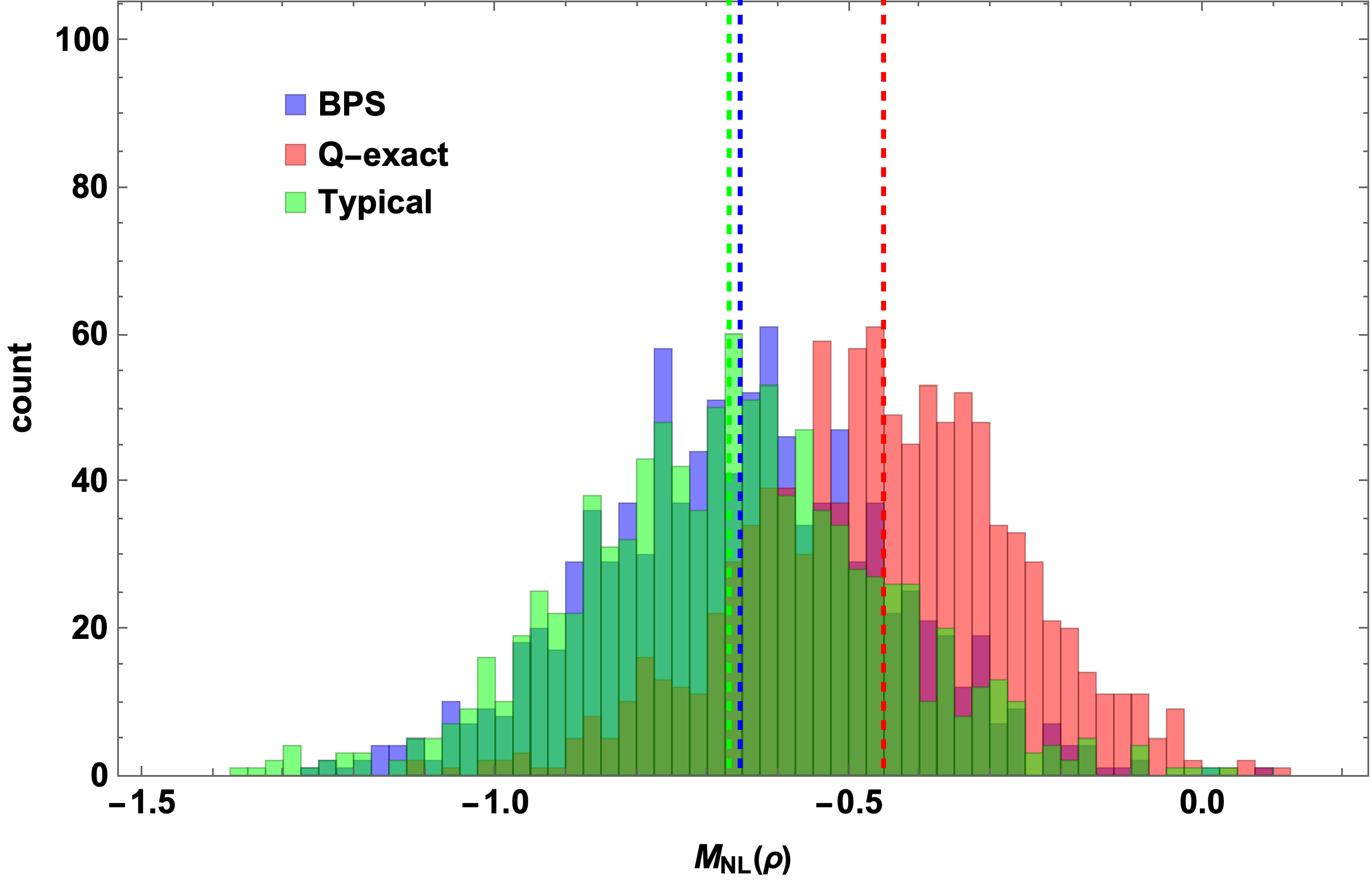}
    \subcaption{$N=8$, $q_R=4$}\label{NLSUSY84}
  \end{subfigure}\quad
    \begin{subfigure}{.3\linewidth}
\includegraphics[height=3.5cm,width=\linewidth]{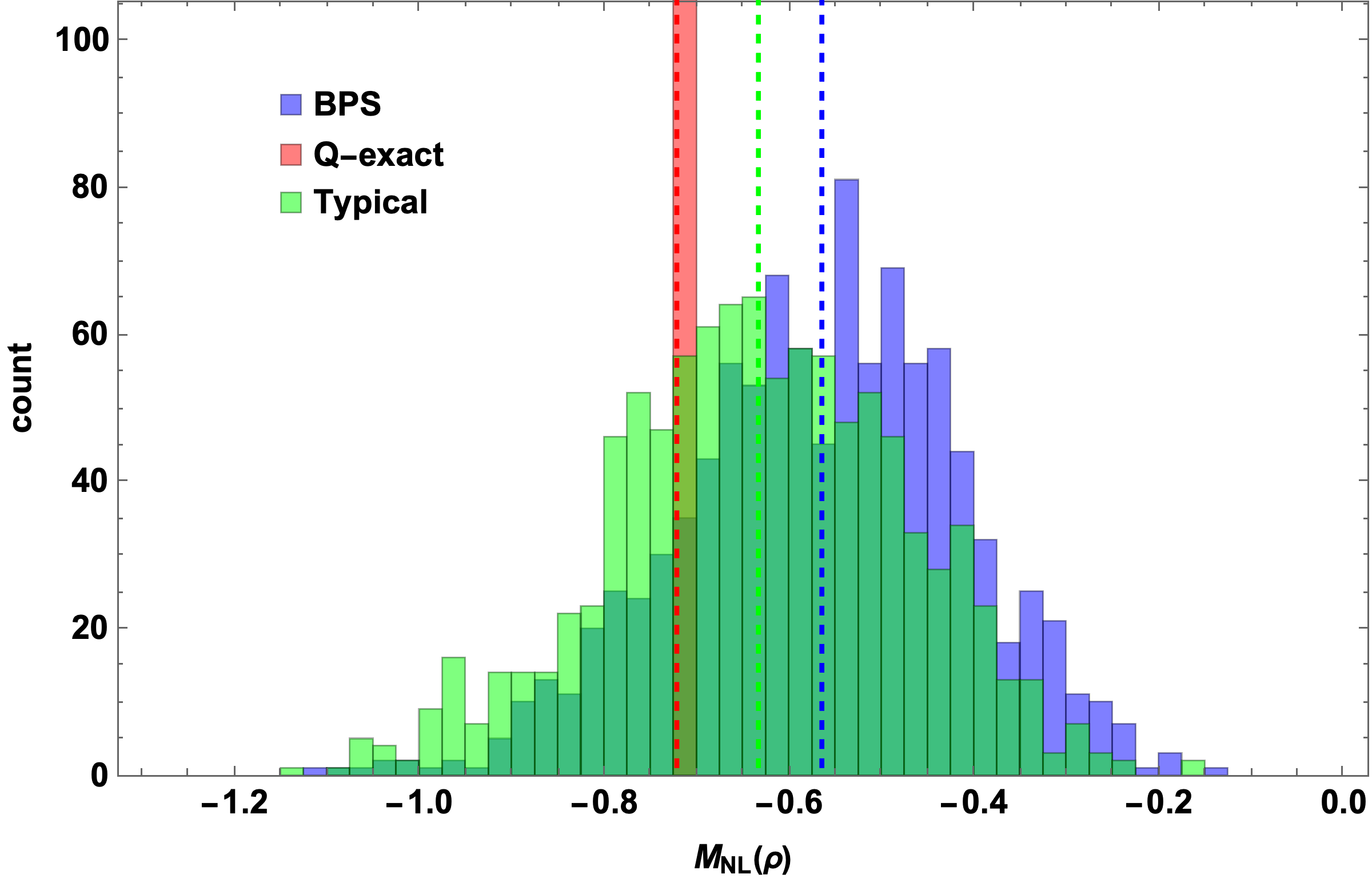}
    \caption{$N=8$, $q_R=5$}\label{NLSUSY85}
  \end{subfigure}
   \caption{\footnotesize{ The plot displays the distribution of the multipartite non-local stabilizer Rényi entropy among typical BPS states, typical $Q$-exact states, and states randomly sampled within the $R$-charge sector for $N=8$.
}}\label{NLSUSY8}
\end{figure}

We next present the histograms of multipartite SRE for typical states, BPS states, and $Q$-exact states across the allowed charge sectors in \Cref{NLSUSY8} for $N = 8$ and \Cref{NLSUSY9} for $N = 9$. A striking feature emerges in both system sizes: the BPS states located in the edge charge sectors $q_R = 3,5$ for $N = 8$ and $q_R = 4,6$ for $N = 9$ display larger multipartite SRE compared to the states in the corresponding central sectors ($q_R = 4$ for $N = 8$ and $q_R = 5$ for $N = 9$). This behavior contrasts with the ordinary stabilizer Rényi entropy, where central sectors often exhibited enhanced values. Within each edge charge sector, we further observe the ordering
\begin{align}
    |M_{\mathrm{NL}}(\text{typical})| > |M_{\mathrm{NL}}(\text{BPS})|,
\end{align}
indicating that typical states possess stronger global multipartite correlations than the BPS states in those sectors.

In the central charge sectors, however, the distinction is less pronounced; typical and BPS states yield nearly comparable multipartite SRE:
\begin{align}
    |M_{\mathrm{NL}}(\text{typical})| \approx |M_{\mathrm{NL}}(\text{BPS})|.
\end{align}
In contrast to both of these behaviors, the $Q$-exact states do not exhibit any consistent pattern in their multipartite SRE relative to the BPS or typical states. Their positions vary from sector to sector, and no stable ordering can be identified.

\begin{figure}[H]
  \centering
  \begin{subfigure}{.45\linewidth}
\includegraphics[height=4.5cm,width=\linewidth]{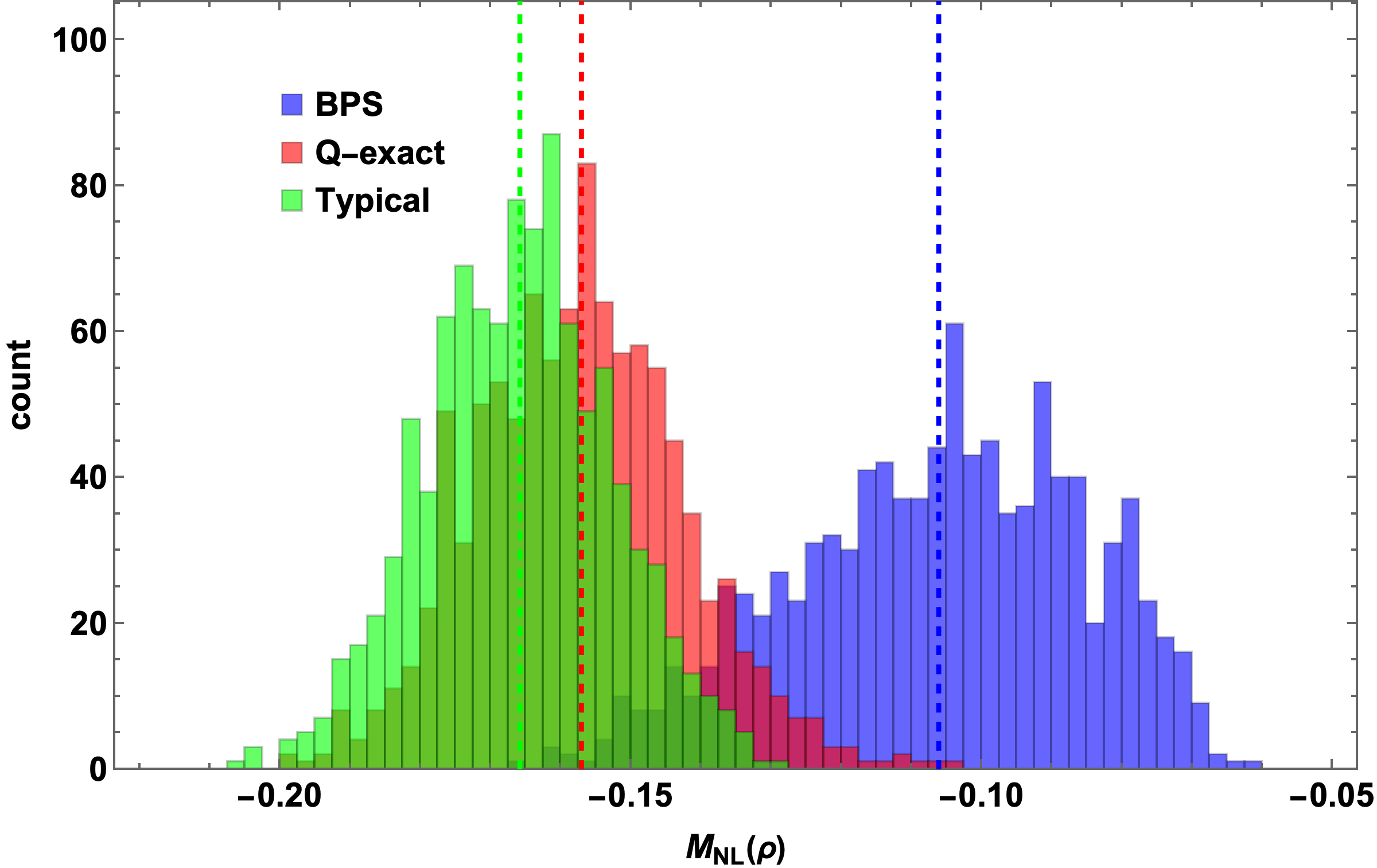}
    \caption{$N=9$, $q_R=3$}\label{NLSUSY93}
  \end{subfigure}\quad
  \begin{subfigure}{.45\linewidth}\includegraphics[height=4.5cm,width=\linewidth]{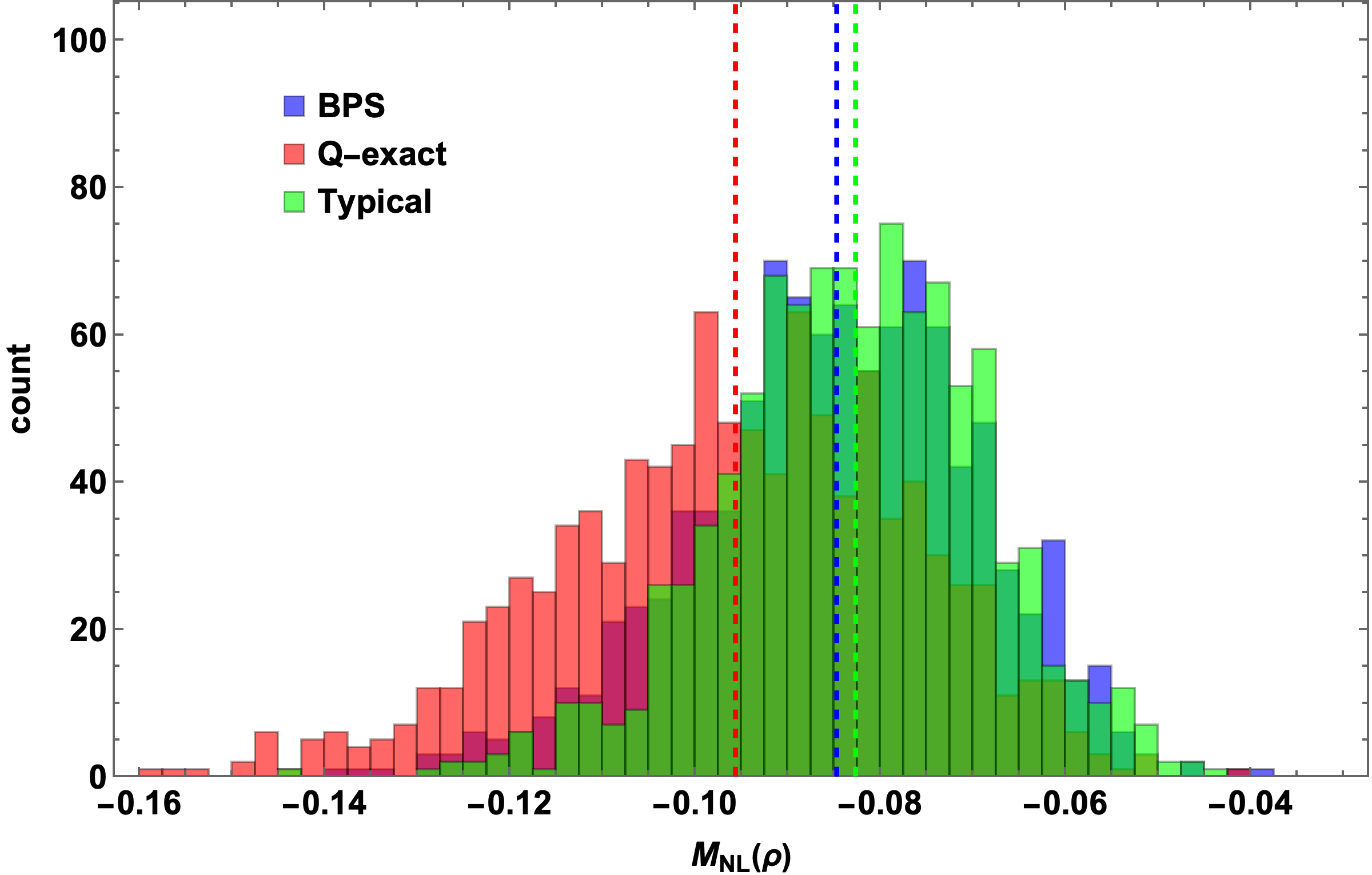}
    \subcaption{$N=9$, $q_R=4$}\label{NLSUSY94}
  \end{subfigure}\quad\\
    \begin{subfigure}{.45\linewidth}
\includegraphics[height=4.5cm,width=\linewidth]{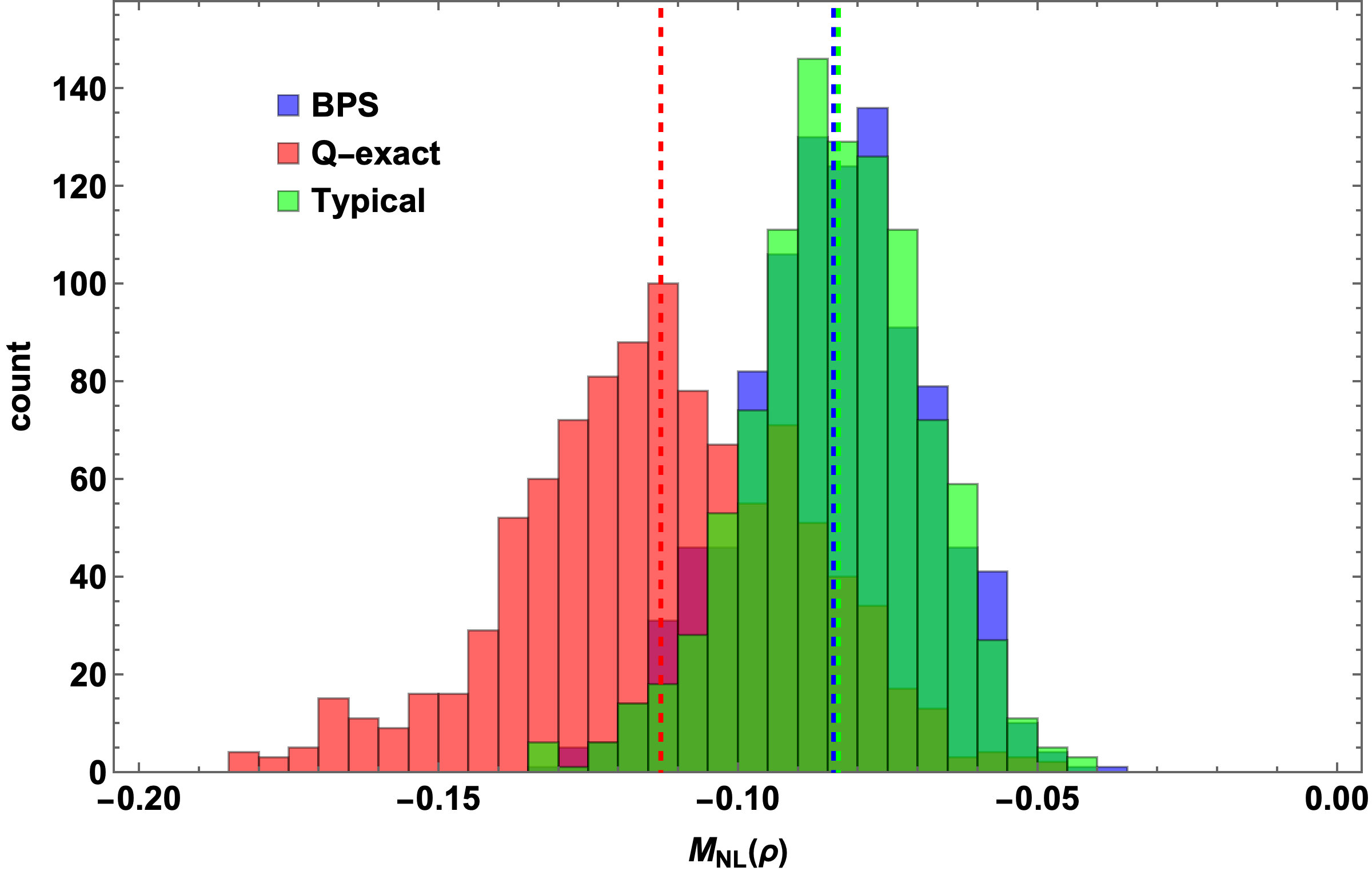}
    \caption{$N=9$, $q_R=5$}\label{NLSUSY95}
  \end{subfigure}
  \begin{subfigure}{.45\linewidth}
\includegraphics[height=4.5cm,width=\linewidth]{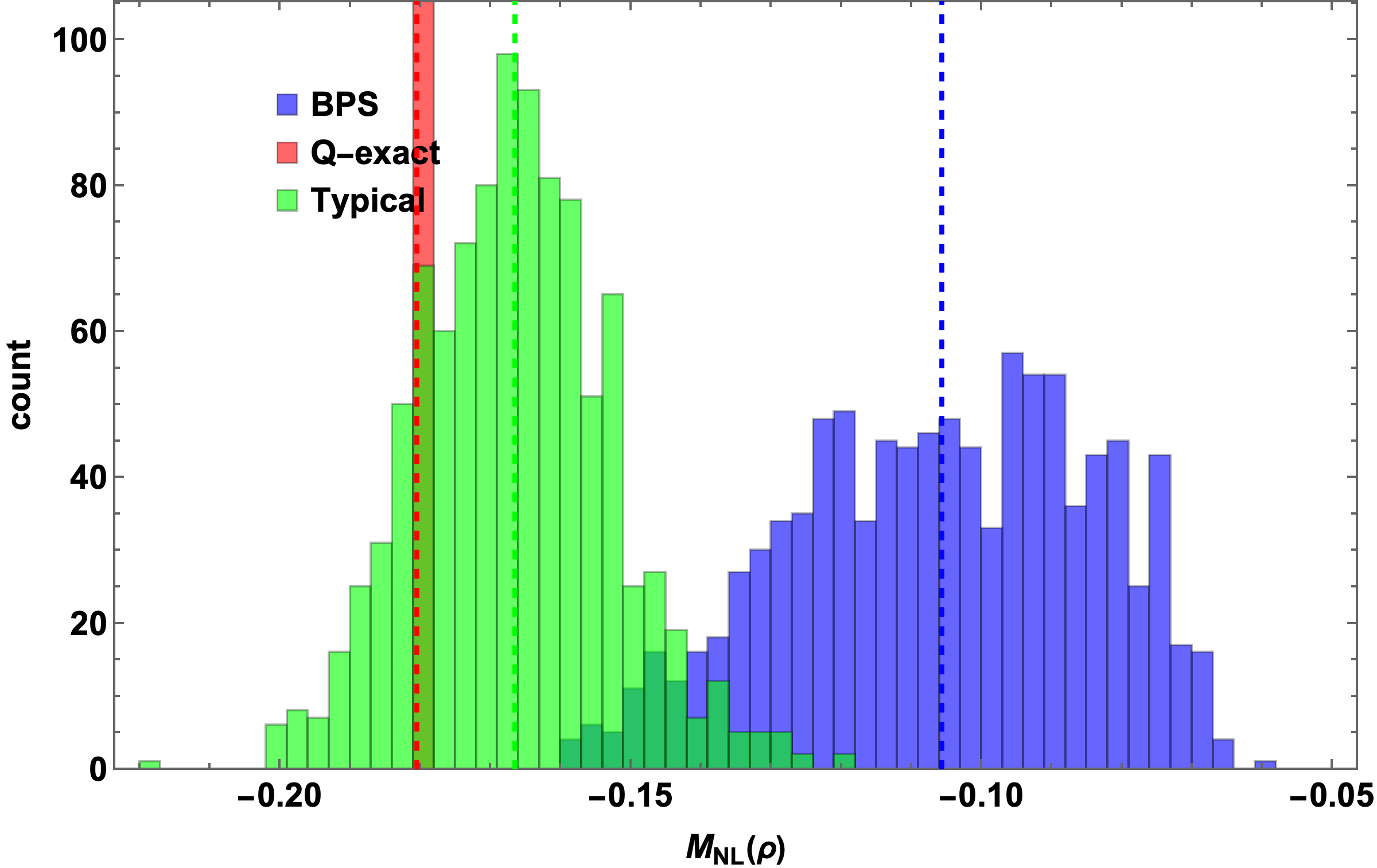}
    \caption{$N=9$, $q_R=6$}\label{NLSUSY96}
  \end{subfigure}
   \caption{\footnotesize{ The plot displays the distribution of the multipartite non-local stabilizer Rényi entropy among typical BPS states, typical $Q$-exact states, and states randomly sampled within the $R$-charge sector for $N=9$.
}}\label{NLSUSY9}
\end{figure}

\section{Conclusion}\label{sec 7}

To summarize, in this work we investigated the structure of quantum magic quantified through the stabilizer Rényi entropy and its multipartite non-local extension across the SYK model and several of its deformations. Our analysis began with a generalization of the notion of multipartite non-local stabilizer entropy, designed to isolate the component of magic arising specifically from genuinely global, multipartite correlations within a quantum state. We established that this measure satisfies several desirable properties, including non-negativity, additivity for independent $n$-party subsystems, invariance under Clifford unitaries, and the existence of natural reference states with vanishing multipartite non-local magic. To illustrate its behavior, we derived analytic expressions for the multipartite non-local magic of generalized GHZ states and complemented these results with numerical evaluations for $n$-partite $W$-states.

We next explored the behavior of the SRE and the multipartite SRE in the SYK model, analyzing both their time evolution and their dependence on temperature. In addition to these global measures, we examined how the average probabilities contributing to the SRE decompose according to structural properties of the underlying Majorana operators. Specifically, we studied their distribution as a function of the Majorana string length, as well as their dependence on the number of qubits on which the strings have support. This allowed us to identify which operator sectors contribute most prominently to the generation of magic in different dynamical and thermal regimes. Notably, our finite-temperature analysis revealed a significant disparity between the stabilizer Rényi entropy of Thermal Pure Quantum (TPQ) states and the thermal ensemble. This observation admits an intriguing interpretation in the context of holography: regarding TPQ states as representative black hole microstates, we uncover a \textit{concealed complexity} where the immense computational hardness is carried by the specific microstructure, hidden from the coarse-grained thermodynamic description. This implies that the quantumness required for unitary dynamics is intrinsically present in individual microstates but washed out in the ensemble average.

These findings suggest that the SRE serves as a fine-grained probe of the black hole interior, capable of detecting the magical resources that distinguish a unitary microstate from a semi-classical geometry. A compelling future direction is to investigate whether this concealed complexity correlates with other diagnostics of late-time quantum chaos, such as the spectral form factor or the emergence of replica wormholes, thereby bridging the gap between quantum information measures and gravitational path integrals.

Following that, we looked at the SRE and multipartite non-local SRE in variants of the SYK model namely the mass deformed SYK and the sparse SYK model. For the mass-deformed SYK$_4$ model, we observed that the SRE and multipartite SRE initially rise toward the SYK$_2$ saturation value and only later relax to the SYK$_4$ value, with the saturation time increasing monotonically with the deformation parameter $g$. The mass term modifies the operator content by activating contribution of Majorana strings of all even lengths, rather than only multiples of four as in pure SYK$_4$. The multipartite SRE also exhibits an early-time dip whose position shifts systematically with the deformation parameter. In the sparse SYK model, SRE decrease as the sparsity increases, with the saturation values suppressed for smaller $p$ (or $n_s$). The multipartite SRE displays a characteristic dip-and-rise structure, and the associated dynamical timescale  grow with increasing sparsity.

Finally, we analyzed the distribution of SRE and multipartite SRE in the ${\cal N}=2$ supersymmetric SYK model. Our focus was on the fortuitous BPS states, whose cohomological structure prevents a straightforward large-$N$ continuation and which are expected to represent black-hole microstates in the holographic dual. For each allowed $R$-charge sector, we computed the quantum magic through both SRE and its multipartite non-local extension for BPS states, $Q$-exact states, and Haar-typical states. We found a distinct structure in the stabilizer Rényi entropy (SRE) across the allowed $R$-charge sectors. For $N=8$ and $N=9$, the BPS states in the central charge sectors exhibit noticeably larger SRE than those in the edge sectors, while for $N=7$ the two sectors are nearly identical. Typical states consistently possess the largest SRE within each charge sector. In central charge sectors $Q$-exact states have the smallest, placing the BPS states at an intermediate level of stabilizer complexity. 

The multipartite non-local SRE, however, displays a more intricate pattern. For $N=7$ the two allowed charge sectors coincide, and for $N=9$ the multipartite SRE again peaks in the central sectors, mirroring the behavior of the ordinary SRE. In contrast, for $N=8$ the trend reverses: the edge charge sectors show larger multipartite non-local SRE than the central one. This illustrates that global multipartite correlations need not follow the same charge-sector hierarchy as the total stabilizer magic, revealing a richer structural dependence encoded in the multipartite measure.

An interesting future direction is to extend our multipartite magic diagnostics to larger system sizes and explore their behavior in the true large-$N$ limit of SYK-type models. It would also be valuable to investigate whether fortuitous BPS microstates exhibit distinctive signatures of magic dynamics in higher dimensions. Finally, applying these measures to other holographic systems may help clarify the role of quantum magic in gravitational dualities. We hope to come back to these interesting issues in near future.

\acknowledgments

V.M., Y.S. and J.Y. was supported by the National Research Foundation of Korea (NRF) grant funded by the Korean government (MSIT) (RS-2022-NR069038, RS-2025-25466315) and by the Brain Pool program funded by the Ministry of Science and ICT through the National Research Foundation of Korea (RS-2023-00261799).

\bibliographystyle{JHEP}

\bibliography{Manabib}

\end{document}